\documentclass[10pt,journal]{IEEEtran}
\IEEEoverridecommandlockouts
\usepackage{graphicx}
\usepackage{cite}
\usepackage{epstopdf}
\usepackage{amssymb, amsmath,amsthm, amsfonts}
\usepackage{ifthen}
\usepackage{tikz}
\usepackage{enumerate}
\usepackage{verbatim}
\usepackage{bbm,balance}
\usepackage{amsmath}
\usepackage{url}

\usetikzlibrary{arrows}
\usetikzlibrary{positioning}
\usetikzlibrary{calc}

\allowdisplaybreaks

\def\eE{\mathbb E}

\newtheorem{theorem}{Theorem}

\newtheorem{lemma}{Lemma}
\newtheorem{claim}{Claim}
\newtheorem{prop}{Proposition}
\newtheorem{defn}{Definition}
\theoremstyle{definition}\newtheorem{remark}{Remark}
\theoremstyle{definition}\newtheorem{example}{Example}

\newcommand{\be}{\begin{equation}}
\newcommand{\ee}{\end{equation}}
\newcommand{\ben}{\begin{equation*}}
\newcommand{\een}{\end{equation*}}
\newcommand{\ba}{\begin{eqnarray}}
\newcommand{\ea}{\end{eqnarray}}

\allowdisplaybreaks
\newcommand{\blue}{\color{black}}

\def \arxiv {1} 

\begin{document}
\title{Multiple Access Channel Simulation}
\author{Gowtham R. Kurri, \IEEEmembership{Member,~IEEE}, Viswanathan~Ramachandran, Sibi~Raj~B.~Pillai, \IEEEmembership{Member,~IEEE}, Vinod~M.~Prabhakaran, \IEEEmembership{Member,~IEEE}\thanks{This work
was supported by the Department of Atomic Energy, Government of India,
under Project RTI4001. The work of Viswanathan Ramachandran was supported in part by Bharti Centre for Communication, IIT Bombay. The work of Sibi~Raj~B.~Pillai was supported by DST-Imprint Project No. 2018/001280. The work of Vinod M. Prabhakaran was supported by the Science and Engineering Research Board, India under MTR/2020/000308. This article was presented in part at the 2021 IEEE International Symposium on Information Theory (ISIT)

Gowtham R. Kurri was with the Tata Institute of Fundamental Research,
Mumbai 400005, India. He is now with the School of Electrical, Computer and Energy Engineering, Arizona State University, Tempe, AZ 85287 USA
(email: \protect\url{gowthamkurri@gmail.com}). 

Viswanathan Ramachandran was with the Department of Electrical Engineering, Indian Institute of Technology Bombay, Mumbai 400076, India. He is now with the EE Department, TU Eindhoven, 5612 AZ Eindhoven, The Netherlands
(email: \protect\url{v.ramachandran@tue.nl}).

Sibi~Raj~B.~Pillai is with the Department of Electrical Engineering, Indian Institute of Technology Bombay, Mumbai 400076, India 
(email: \protect\url{bsraj@ee.iitb.ac.in}).

 Vinod M. Prabhakaran is with the School of Technology and Computer Science, Tata Institute of Fundamental Research, Mumbai 400005, India 
(email: \protect\url{vinodmp@tifr.res.in}). }}
\maketitle

\begin{abstract}
We study the problem of simulating a two-user multiple-access channel (MAC) over a multiple access network of noiseless links. Two encoders observe independent and identically distributed (i.i.d.) copies of a source random variable each, while a decoder observes i.i.d. copies of a side-information random variable. {\blue There are rate-limited noiseless communication links between each encoder and the decoder,  and there is independent pairwise shared randomness between all the three possible pairs of nodes}. The decoder has to output approximately i.i.d. copies of another random variable jointly distributed with the two sources and the side information. We are interested in the rate tuples which permit this simulation. This setting can be thought of as a multi-terminal  generalization of the point-to-point channel simulation problem studied by Bennett et al. (2002) and Cuff (2013). {\blue When the pairwise shared randomness between the encoders is absent, the setting reduces to a special case of MAC simulation using another MAC studied by Haddadpour et al.~(2013). We establish that the presence of encoder shared randomness can strictly improve the communication rate requirements. We first show that the inner bound derived from Haddadpour et al.~(2013) is tight when the sources at the encoders are conditionally independent given the side-information at the decoder. This result recovers the existing results on point-to-point channel simulation and function computation over such multi-terminal networks. We then explicitly compute the communication rate regions for an example both with and without the encoder shared randomness and demonstrate that its presence strictly reduces the communication rates. Inner and outer bounds for the general case are also obtained. }
\end{abstract}

\section{Introduction}

What is the minimum amount of communication required to create correlation remotely? The \emph{channel simulation} problem seeks to answer this fundamental question. In the point-to-point formulation, an encoder observing an independent and identically distributed (i.i.d.) source $X^n$ with distribution $q_X$ sends a message through a noiseless link to a decoder. The decoder has to output $Y^n$ such that the total variation distance between the joint distribution on $(X^n,Y^n)$ and the i.i.d. joint distribution induced by passing the source $X^n$ through a discrete memoryless channel $q_{Y|X}$ vanishes asymptotically. This requirement that the synthesized joint distribution be \emph{close} to the desired i.i.d. joint distribution in total variation distance has been termed as \emph{strong coordination}~\cite{cuff2010coordination}, which is also the focus of this paper. 
A source of common randomness accessible to both the encoder and the decoder may assist them in the aforementioned task. This framework was first investigated by Bennett et al.~\cite{BennettSST02} assuming unlimited common randomness, where they established a `reverse Shannon theorem' to synthesize a noisy channel from a noiseless channel\footnote{Referring to Shannon's channel coding theorem as the simulation of a noiseless channel using a noisy channel.}. It was shown that the minimum communication rate is nothing but the mutual information $I(X;Y)$ of the joint distribution. Harsha et al.~\cite{Harsha2010} studied the non-asymptotic version of this problem. Winter~\cite{winter2002compression} studied the setting with limited common randomness, albeit only for a certain extremal operating point. Cuff~\cite{cuff2013distributed} and Bennett et al.~\cite{6757002} independently determined the entire optimal trade-off between communication and shared randomness rates. Later, Wilde et al.~\cite{wilde2012information} obtained a similar trade-off in the quantum information-theoretic setting, generalizing the above result. Yassaee~et~al.~\cite{YassaeeGA15} established a similar trade-off for channel simulation in a point-to-point network with side-information at the decoder. {\blue Simulation of a channel using another channel (instead of the noiseless communication link) was studied by Haddadpour et al.~\cite{HaddadpourYBGAA17} and Cervia et al.~\cite{CerviaLLB20}.} A weaker form of coordination, namely, \emph{empirical coordination}, where only the empirical distribution of the sequence of samples is required to be close to the desired distribution, has also been studied in point-to-point networks~\cite{cuff2010coordination,YassaeeGA15,CerviaLBT16,Treust17,ChouBK18,MylonakisSS19}.

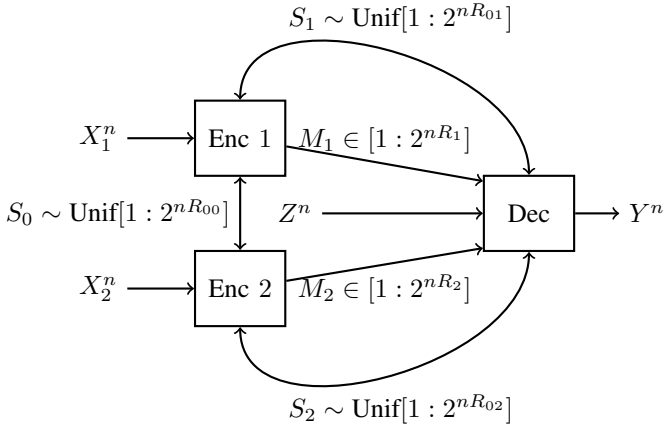
\begin{figure}[h]
\centering
\begin{tikzpicture}[thick]
\node (d1) at (-2,1) [rectangle, draw, right, minimum height=1.0cm, minimum width = 1.2cm]{Enc $1$};
\node (d2) at (1.84,0) [rectangle, draw, right, minimum height=1.0cm, minimum width = 1.2cm]{Dec};
\draw[->] (d1) -- (d2.145) node[midway,above] {$M_1 \in [1:2^{nR_1}]$};
\draw[->] (d2) --++(1.2,0) node[right]{$Y^n$};
\node (c1) at (0.75,2.6) {$S_1 \sim\textup{Unif}[1:2^{nR_{01}}]$};
\node (c1) at (0.75,-2.6) {$S_2 \sim\textup{Unif} [1:2^{nR_{02}}]$};
\draw[<-] (d1) --++(-1.5,0) node[left]{$X_1^n$};
\draw [<->] (d1.north) to [out=90,in=90] (d2.north);
\node (e1) at (-2,-1) [rectangle, draw, right, minimum height=1.0cm, minimum width = 1.2cm]{Enc $2$};
\draw[<-] (e1) --++(-1.5,0) node[left]{$X_2^n$};
\draw[->] (e1) -- (d2.-143) node[midway,below] {$M_2 \in [1:2^{nR_2}]$};
\draw [<->] (e1.south) to [out=-90,in=-90] (d2.south);
\draw[<-] (d2) --++(-2.75,0) node[left]{$Z^n$};
\draw[<->] (d1) -- (e1) node[midway,left] {$S_0\sim\textup{Unif}[1:2^{nR_{00}}]$};
\end{tikzpicture}
\caption{Strong coordination over a multiple-access network. Encoder $j\in \{1,2\}$ on observing the source $X_j^n$ and shared randomness $S_0,S_j$ sends a message $M_j$ over a noiseless link to the decoder, which has side-information $Z^n$. Here $(X_{1i},X_{2i},Z_i)$, $i=1,\dots,n$, are i.i.d. with $q_{X_1X_2Z}$. Also $(X_1^n,X_2^n,Z^n),S_0,S_1$, and $S_2$ are mutually independent. The decoder has to output $Y^n$ such that $(X_{1i},X_{2i},Z_i,Y_i)$, $i=1,\dots,n$, are approximately i.i.d. with $q_{X_1X_2ZY}$.} \label{fig:encSR}
\end{figure}
Channel simulation problems may also be thought of as distributed computation of randomized functions. Distributed computation of (deterministic) functions has received much attention in the computer science literature (see, e.g., \cite{kushilevitz_nisan_1996,BravermanR14} and references therein) and the information theory literature~\cite{KornerM79,Yamamoto82,Kaspi85,HanK87,AlonO96,OrlitskyR2001,nazer2007computation,NazerG07b,Gastpar08,ma2011some,AppuswamyFNZ11,KowshikK12,KuzuokaW15,SefidgaranGA15,SefidgaranT16,Watanabe20}. Two terminal interactive function computation was studied in \cite{Yamamoto82,Kaspi85,AlonO96,OrlitskyR2001,ma2011some}. Distributed multi-terminal function computation in a multiple-access network of noiseless links was studied by K\"{o}rner and Marton~\cite{KornerM79}, Han and Kobayashi~\cite{HanK87}, Kuzuoka and Watanabe~\cite{KuzuokaW15}, Watanabe~\cite{Watanabe20}, and Sefidgaran et al.~\cite{SefidgaranGA15}. Function computation in more general graph networks where a single node seeks to compute a function of the inputs at the other nodes was studied by Appuswamy et al.~\cite{AppuswamyFNZ11}, Kowshik and Kumar~\cite{KowshikK12}, and Sefidgaran and Tchamkerten~\cite{SefidgaranT16}. A related line of work is function computation over multiple-access channels studied by  Nazer and Gastpar~\cite{nazer2007computation,NazerG07b}, and Gastpar~\cite{Gastpar08}, which, in contrast to the above works, exploit the computation performed by the channel itself. This is fundamentally different from computation over a network of noiseless links where the communication channel does not perform any computation.

 There are relatively few conclusive results on channel simulation in multi-terminal networks where possibly randomized functions of the inputs need to be computed. A natural generalization of the point-to-point network to a multi-terminal setting is the cascade network~\cite{SatpathyC16,vellambi2017strong}. Satpathy and Cuff~\cite{SatpathyC16} considered a cascade network under a secrecy requirement and obtained the optimal trade-off between communication and common randomness rates. Vellambi et al.~\cite{vellambi2017strong} obtained the optimal rates for specific settings, e.g., when the communication topology matches the coordination structure. Empirical coordination in multi-terminal networks has been studied in~\cite{cuff2010coordination,BereyhiBMA13,MylonakisSS19a}.

In this paper, we study multiple access channel simulation (randomized function computation) over a multiple access network of noiseless links. In particular, there are two encoders who share noiseless communication links with a decoder -- see Figure~\ref{fig:encSR}. The two encoders and the decoder observe i.i.d. copies of the sources $X_1$ and $X_2$, and the side-information $Z$, respectively, that are generated according to a given distribution $q_{X_1X_2Z}$. Each encoder has access to a resource of pairwise shared randomness with the decoder. {\blue There is another resource of pairwise shared randomness between the encoders. The three pairwise shared randomness resources are independent of each other and also jointly independent of the sources and side information.} In addition, the encoders and the decoder may privately randomize. The encoders transmit messages through noiseless links to the decoder, whose output together with the input sources and decoder side-information should be approximately i.i.d. with $q_{X_1X_2ZY}$. 

{\blue The work which is closest to ours is by Haddadpour et al.~\cite[Section~IV]{HaddadpourYAG13}, who studied simulation of a multiple-access channel from another multiple-access channel, a more general resource than the multiple-access network of noiseless links considered here. The main difference is the presence of shared randomness between the encoders, i.e., $S_0$, in our model.
Haddadpour et al.~\cite{HaddadpourYAG13} obtained an inner bound to the rate-region for their setting. In this work, we investigate the role of this additional shared randomness resource and specifically ask the question: Can this additional pairwise shared randomness resource between the encoders (that is not available to the decoder) strictly improve the communication rate requirements for channel simulation? We answer this question in the affirmative. This complements the results of \cite{cuff2013distributed,HaddadpourYAG13} which established that shared randomness between encoder(s) and decoder is a useful resource for channel simulation. The configuration of a shared randomness resource between the encoders is reminiscent of the setting of multiple access channels with \emph{partially cooperating encoders}~\cite{willems1983discrete} -- see also \cite{bross2012dirty,noorzad2017unbounded}. However, in the present setting, the encoders are not allowed to cooperate after observing the sources.}    

The problem of finding the optimal communication rates for computing possibly randomized functions at the decoder in this multiple-access network of noiseless links remains largely open even for the case of independent sources. Sefidgaran and Tchamkerten~\cite{SefidgaranT16} determined the optimal communication rates for computing deterministic functions\footnote{It turns out that shared randomness does not aid in deterministic function computation -- see Remark~\ref{remark:CR}.} when the sources at the encoders are conditionally independent given the side-information at the decoder, i.e., $I(X_1;X_2|Z)=0$. Note that Sefidgaran and Tchamkerten~\cite{SefidgaranT16} in fact studied a more general setting consisting of multiple terminals over a rooted multi-level directed tree, where the multiple-access network of noiseless links is a special case. Atif et al.~\cite{AtifPP20} studied multiple access channel simulation in the presence of three-way common randomness instead of pairwise shared randomness as above and obtained an achievable inner bound. Atif et al.~\cite{atif2019faithful} obtained a similar inner bound in the quantum information-theoretic setting. {\blue After the submission of this work, the authors learnt about a concurrent work by Atif et al.~\cite{atif2020source}, which, like \cite{HaddadpourYAG13}, considered pairwise shared randomness between each encoder and the decoder (i.e., no encoder shared randomness), and derived inner and outer bounds on the rate region. Achievable schemes using algebraic-structured codes exploiting the specific structure of the function to be computed were further explored by Krithivasan and Pradhan~\cite{krithivasan2011distributed} and Atif and Pradhan~\cite{atifsandeep}, where the latter considered the quantum setting, in the spirit of K\"{o}rner and Marton~\cite{KornerM79}.}

{\blue\noindent \textbf{Main Contributions.}
We analyze the role of pairwise shared randomness in multiple-access channel simulation and establish that the presence of encoder shared randomness can strictly improve the communication rates required. 
\begin{itemize}
\item First, for the case when the encoder shared randomness is absent, we obtain the optimal trade-off between the communication rates and one of the two remaining shared randomness rates (under sufficiently large other shared randomness rate) when $I(X_1;X_2|Z)=0$ (Theorem~\ref{thm:indep}). This recovers the results on point-to-point channel simulation~\cite{cuff2013distributed}, and also shows that the inner bound of Haddadpour et al.~\cite[Theorem~3]{HaddadpourYAG13} is tight for the special case under consideration. When specialized to deterministic function computation, our Theorem~\ref{thm:indep} also recovers the result of Sefidgaran and Tchamkerten~\cite{SefidgaranT16} restricted to multiple-access network of noiseless links. However,  the techniques used there do not readily generalize to randomized function computation (see the discussion below Theorem~\ref{thm:indep}). 
\item Next, we explicitly compute the communication rate regions for an example both with and without the encoder shared randomness and show that its presence strictly reduces the communication rates required for coordination (see Example~\ref{example} and Section~\ref{section:exampleoptimality}).
\item We also derive general inner and outer bounds on the rate region with all the three pairwise shared randomness resources present (Theorems~\ref{thm:encsideIBsr}, \ref{thm:encsideOBsr}, and \ref{Theorem:outbndindp}). Our inner bound proofs are mostly in the spirit of Haddadpour et al.~\cite[Theorem~3]{HaddadpourYAG13}. 
\end{itemize}
}

{\blue The remainder of this paper is organized as follows. We present our system model in Section~\ref{UPDATEDsec:SMR}, and give the results for the case when the encoder shared randomness is absent in Section~\ref{secton:mainresults1}. In Section~\ref{sec:ex1}, we show through an example that the presence of encoder shared randomness can lead to a strict improvement of the communication rate region. General inner and outer bounds on the rate region are presented in Section~\ref{section:encSR}. The proofs of our main results are given in Section~\ref{proofs}. 
}

\section{System Model} \label{UPDATEDsec:SMR}
{\blue We study the problem of strong coordination of signals in a three-node multiple-access network. There are two encoders with inputs $X_1^n$ and $X_2^n$, respectively, and a decoder with side-information $Z^n$, where $(X_{1i},X_{2i},Z_i)$, $i=1,\dots,n$, are independent and identically distributed (i.i.d.) with distribution $q_{X_1X_2Z}$, with $X_1,X_2,$ and $Z$ taking values in finite alphabets $\mathcal{X}_1,\mathcal{X}_2$, and $\mathcal{Z}$, respectively. 
For $j=1,2,$ encoder $j$ and the decoder have access to a pairwise shared randomness $S_j$ uniformly distributed on $[1:2^{nR_{0j}}]$. There is another pairwise shared randomness $S_0$ between the two encoders that is uniformly distributed on $[1:2^{nR_{00}}]$. The random variables $S_0, S_1$ and $S_2$ are independent and also jointly independent of $(X_1^n,X_2^n,Z^n)$. Encoder $j \in \{1,2\}$ observes $X_j^n$ and shared randomness $(S_0,S_j)$, and sends a message $M_j\in[1:2^{nR_j}]$ over a noiseless communication link to the decoder. The decoder observes $M_1,M_2$,  in addition to the shared randomness $S_1,S_2,$ and side-information $Z^n$. The goal is to output $Y^n$ (where $Y_i$, $i=1,\dots,n,$ takes values in a finite alphabet $\mathcal{Y}$) which along with the input sources and decoder side-information is approximately distributed according to $q^{(n)}_{X_{1}X_{2}Z Y}(x_1^n,x_2^n,z^n,y^n):=\prod_{i=1}^nq_{X_1X_2ZY}(x_{1i},x_{2i},z_i,y_i)$ (see Figure~\ref{fig:encSR}).

\begin{defn}\label{defn:code}
A $(2^{nR_1}, 2^{nR_2}, 2^{nR_{00}}, 2^{nR_{01}}, 2^{nR_{02}}, n)$ \emph{code} consists of two randomized encoders $p^{\emph{E}_1}(m_1|s_0,s_1,x_1^n)$ and $p^{\emph{E}_2}(m_2|s_0,s_2,x_2^n)$ and a randomized decoder $p^{\emph{D}}(y^n|s_1,s_2,m_1,m_2,z^n)$, where $s_0\in[1:2^{nR_{00}}]$ and $s_j\in[1:2^{nR_{0j}}]$, $m_j\in[1:2^{nR_j}]$, $j=1,2$.
\end{defn}
The joint distribution of $(S_0,S_1,S_2,X_1^n,X_2^n,Z^n,M_1,M_2,Y^n)$ and the resulting induced joint distribution on $(X_1^n,X_2^n,Z^n,Y^n)$ are respectively given by
\begin{align*}
&p(s_0,s_1,s_2,x_1^n,x_2^n,z^n,m_1,m_2,y^n) \notag\\
&\hspace{12pt}=\frac{1}{2^{n(R_{00}+R_{01}+R_{02})}}p(x_1^n,x_2^n,z^n)\prod_{j=1}^2p^{\text{E}_j}(m_j|s_0,s_j,x_j^n) \\
& \hspace{24pt}\times p^{\text{D}}(y^n|s_1,s_2,m_1,m_2,z^n),
\end{align*}
and
\begin{align*}
&p^{\text{ind}}(x_1^n,x_2^n,z^n,y^n)\\
&\hspace{12pt}=\sum_{s_0,s_1,s_2,m_1,m_2}p(s_0,s_1,s_2,x_1^n,x_2^n,z^n,m_1,m_2,y^n).
\end{align*}

Recall that the total variation between two p.m.f.'s $p_X$ and $q_X$ on the same alphabet $\mathcal{X}$ is defined as
\begin{align*}
||p_X-q_X||_1 \triangleq \frac{1}{2} \sum_{x\in\mathcal{X}} |p_X(x)-q_X(x)|.
\end{align*}
\begin{defn} \label{def:ach}
A rate tuple $(R_1, R_2, R_{00}, R_{01}, R_{02})$ is said to be \emph{achievable for a distribution} $q_{X_1X_2ZY}$ if there exists a sequence of $(2^{nR_1},2^{nR_2},2^{R_{00}},2^{nR_{01}},2^{nR_{02}},n)$ codes such that
\end{defn} 
\begin{align}\label{eqn:correctness}
\lim_{n \to \infty} ||p^{\text{ind}}_{X_1^n,X_2^n,Z^n,Y^n}-q^{(n)}_{X_1X_2ZY}||_{1}=0,
\end{align}
where $q^{(n)}_{X_1X_2ZY}$ is the product distribution given by
\begin{align*}
q^{(n)}_{X_1X_2ZY}(x_1^n,x_2^n,z^n,y^n):=\prod_{i=1}^n q_{X_1X_2ZY}(x_{1i},x_{2i},z_i,y_i).
\end{align*} 
\begin{remark}\label{remark:CR}
If $q_{X_1X_2ZY}$ is such that $Y$ is a deterministic function of $(X_1,X_2,Z)$, then the pairwise shared randomness and the private randomness at the encoders and the decoder do not have any effect on the communication rates. In fact, more generally, common randomness available to both the encoders and the decoder  does not help to improve communication rates in this case.
\if \arxiv 0
This follows from standard probabilistic method arguments. 
\fi
\if \arxiv 1%
This follows from standard probabilistic method arguments\footnote{To see this, notice that if $Y=f(X_1,X_2,Z)$, then \eqref{eqn:correctness} reduces to $\lim\limits_{n\rightarrow \infty}P(Y^n\neq f(X_1^n,X_2^n,Z^n))=0$. Let $W$ denote the random variable corresponding to the common randomness. A simple application of the law of total probability implies that there exists a realization $w^*$ of $W$ such that $P(Y^n\neq f(X_1^n,X_2^n,Z^n)|W=w^*)\leq P(Y^n\neq f(X_1^n,X_2^n,Z^n))$. Therefore, new deterministic encoding and decoding functions can be defined by fixing $W=w^*$ and at the same time not increasing the probability of error.}.
\fi
\end{remark}
\begin{defn}\label{defn:new}
The \emph{rate region} $\mathcal{R}_{\textup{MAC-coord}}$ is the closure of the set of all achievable rate tuples $(R_1,R_2,R_{00},R_{01},R_{02})$. 
\end{defn}
Let $\mathcal{R}_{\textup{MAC-coord, UL-$(S_0,S_1,S_2)$}}$ be the rate region when all the three pairwise shared randomness are unlimited, i.e., 
\begin{align}&\mathcal{R}_{\textup{MAC-coord, UL-$(S_0,S_1,S_2)$}}=\{(R_1,R_2) : \exists \ R_{00}, R_{01}\ \text{and}\ R_{02} \nonumber\\
&\text{s.t.}\ (R_1,R_2,R_{00},R_{01},R_{02}) \in \mathcal{R}_{\textup{MAC-coord}}\}.
\end{align}
For purposes of comparison, we separately consider a special case when the encoder shared randomness $S_0$ is absent. A code, an achievable rate tuple, and the rate region can be defined analogously. In particular, the code and an achievable rate tuple can be defined similar to Definitions~\ref{defn:code} and \ref{def:ach} by removing the respective coordinates containing $R_{00}$ and by treating $S_0=\emptyset$. The \emph{rate region} $\mathcal{R}_{\textup{MAC-coord}}^{\text{NO-$S_0$}}$ is the closure of the set of all achievable rate tuples $(R_1,R_2,R_{01},R_{02})$ when $S_0$ is absent. Let $\mathcal{R}_{\textup{MAC-coord, UL-$S_2$}}^{\textup{NO-$S_0$}}$ be the rate region when the shared randomness $S_2$ is unlimited, i.e., 
\begin{align}
\mathcal{R}_{\textup{MAC-coord, UL-$S_2$}}^{\text{NO-$S_0$}}=\{(&R_1,R_2,R_{01}) : \exists \ R_{02} \ \text{s.t.}\nonumber\\
& (R_1,R_2,R_{01},R_{02}) \in \mathcal{R}_{\textup{MAC-coord}}^{\textup{NO-$S_0$}}\}.
\end{align}
The \emph{communication rate region} $\mathcal{R}_{\textup{MAC-coord, UL-$(S_1,S_2)$}}^{\textup{NO-$S_0$}}$ is given by $\{(R_1,R_2) : \exists \ R_{01}\  \text{and}\ R_{02} \ \text{s.t.}\ (R_1,R_2,R_{01},R_{02}) \in \mathcal{R}_{\textup{MAC-coord}}^{\textup{NO-$S_0$}}\}$. The communication rate region can be thought of as the trade-off between the communication rates $R_1$ and $R_2$ under sufficiently large pairwise shared randomness rates $R_{01}$ and $R_{02}$. 
}

\section{\blue Rate Region with No Encoder Shared Randomness}  \label{secton:mainresults1}
{\blue In this section, we present our results for the special case when the encoder shared randomness $S_0$ is absent (see Figure~\ref{fig:noiseless}). Our main result here is a complete characterization of the rate region $\mathcal{R}_{\textup{MAC-coord, UL-$S_2$}}^{\text{NO-$S_0$}}$ when the sources are conditionally independent given the side-information.
We first present an inner bound to the rate region $\mathcal{R}_{\textup{MAC-coord}}^{\text{NO-$S_0$}}$ that follows from Haddadpour~et al.~\cite[Theorem 3]{HaddadpourYAG13} which studies multiple-access channel simulation using another multiple-access channel as a resource (instead of a multiple-access network of noiseless links as in this work).} 
\begin{theorem}[Inner Bound with No Encoder Shared Randomness] \label{thm:encsideIB}
Given a p.m.f. $q_{X_1X_2ZY}$, the rate tuple $(R_1,R_2,R_{01},R_{02})$ is in {\blue$\mathcal{R}_{\textup{MAC-coord}}^{\emph{NO-$S_0$}}$} if
\begin{align*} 
R_1 &\geq I(U_1;X_1|U_2,Z,T) \\
R_2 &\geq I(U_2;X_2|U_1,Z,T) \\
R_1+R_2 &\geq I(U_1,U_2;X_1,X_2|Z,T) \\
R_1+R_{01} &\geq I(U_1;X_1,X_2,Y|Z,T) \notag\\
&\hspace{24pt}-I(U_1;U_2|Z,T) \\
R_2+R_{02} &\geq I(U_2;X_1,X_2,Y|Z,T) \notag\\
&\hspace{24pt}-I(U_1;U_2|Z,T) \\
R_1+R_2+R_{01} &\geq I(U_1;X_1,X_2,Y|Z,T)\notag\\
&\hspace{24pt}+I(U_2;X_2|U_1,Z,T) \\
R_1+R_2+R_{02} &\geq I(U_2;X_1,X_2,Y|Z,T)\notag\\
&\hspace{24pt}+I(U_1;X_1|U_2,Z,T) \\
R_1+R_2+R_{01}+R_{02} &\geq I(U_1,U_2;X_1,X_2,Y|Z,T),
\end{align*}
for some p.m.f. 
\begin{align}
p(x_1,&x_2,z,t,u_1,u_2,y)=\nonumber\\
&\hspace{12pt}p(x_1,x_2,z)p(t)p(u_1|x_1,t)p(u_2|x_2,t)p(y|u_1,u_2,z,t)\label{eq:pmfthm1}
\end{align}
such that 
\begin{align*}
\sum\limits_{u_1,u_2}p(x_1,x_2,z,u_1,u_2,y|t) &= q(x_1,x_2,z,y),\ \text{for all}\ t.
\end{align*}
\end{theorem}

\begin{remark}
The inner bound in Theorem~\ref{thm:encsideIB} without side-information $Z$ follows as a corollary of Haddadpour~et al.~\cite[Theorem 3]{HaddadpourYAG13}. In particular, let the resource mutiple access channel in \cite[Theorem 3]{HaddadpourYAG13} consists of two independent channels which can be converted to two noiseless links by operating at the rates of respective channel capacities. The details are analogous to how the inner bound of point-to-point channel simulation using a noiseless link \cite{cuff2013distributed} can be recovered from that of point-to-point channel simulation using another channel \cite[Theorem~1]{HaddadpourYAG13}  (see \cite[Remark~2]{HaddadpourYAG13}). However, for completeness, we present a proof of Theorem~\ref{thm:encsideIB} incorporating the side-information $Z$ with minor differences to that of \cite[Theorem~3]{HaddadpourYAG13} in Appendix~\ref{app:pfThm1}.
\end{remark}

{\blue The intuition behind the auxiliary random variables $U_1$ and $U_2$ is analogous to the auxiliary random variable in the point-to-point channel simulation setting~\cite{cuff2013distributed,6757002}. In particular, $U_1$ and $U_2$ may be thought of as quantized versions of the observations $X_1$ and $X_2$ respectively (with respect to the corresponding shared random variables). The Markov conditions on $U_1$ and $U_2$ in \eqref{eq:pmfthm1} arises naturally due to the information structure of the problem. We note that the inner bound in Theorem~\ref{thm:encsideIB} with $Z=\emptyset$ also appears in the concurrent work by Atif et al.~\cite[Theorem~1]{atif2020source}.}
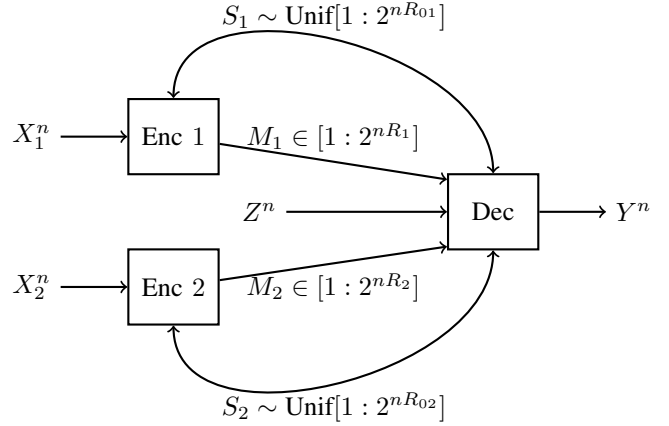
\begin{figure}[t]
\centering
\begin{tikzpicture}[thick]
\node (d1) at (-2,1) [rectangle, draw, right, minimum height=1.0cm, minimum width = 1.2cm]{Enc $1$};
\node (d2) at (2.25,0) [rectangle, draw, right, minimum height=1.0cm, minimum width = 1.2cm]{Dec};
\draw[->] (d1) -- (d2.145) node[midway,above] {$M_1 \in [1:2^{nR_1}]$};
\draw[->] (d2) --++(1.5,0) node[right]{$Y^n$};
\node (c1) at (0.75,2.6) {$S_1 \sim \textup{Unif}[1:2^{nR_{01}}]$};
\node (c1) at (0.75,-2.6) {$S_2 \sim \textup{Unif}[1:2^{nR_{02}}]$};
\draw[<-] (d1) --++(-1.5,0) node[left]{$X_1^n$};
\draw [<->] (d1.north) to [out=90,in=90] (d2.north);
\node (e1) at (-2,-1) [rectangle, draw, right, minimum height=1.0cm, minimum width = 1.2cm]{Enc $2$};
\draw[<-] (e1) --++(-1.5,0) node[left]{$X_2^n$};
\draw[->] (e1) -- (d2.-143) node[midway,below] {$M_2 \in [1:2^{nR_2}]$};
\draw [<->] (e1.south) to [out=-90,in=-90] (d2.south);
\draw[<-] (d2) --++(-2.75,0) node[left]{$Z^n$};
\end{tikzpicture}
\caption{Strong coordination over a multiple-access network with no encoder shared randomness. Encoder $j\in \{1,2\}$ on observing the source $X_j^n$ and shared randomness $S_j$ sends a message $M_j$ over a noiseless link to the decoder, which has side-information $Z^n$. Here $(X_{1i},X_{2i},Z_i)$, $i=1,\dots,n$, are i.i.d. with $q_{X_1X_2Z}$. Also $(X_1^n,X_2^n,Z^n),S_1$, and $S_2$ are mutually independent. The decoder has to output $Y^n$ such that $(X_{1i},X_{2i},Z_i,Y_i)$, $i=1,\dots,n$, are approximately i.i.d. with $q_{X_1X_2ZY}$.} \label{fig:noiseless}
\end{figure}
We now provide an outer bound to the region {\blue$\mathcal{R}_{\textup{MAC-coord}}^{\text{NO-$S_0$}}$}.
\begin{theorem}[Outer Bound with No Encoder Shared Randomness] \label{thm:encsideOB}
Given a p.m.f. $q_{X_1X_2ZY}$, any rate tuple $(R_1,R_2,R_{01},R_{02})$ in {\blue$\mathcal{R}_{\textup{MAC-coord}}^{\emph{NO-$S_0$}}$} satisfies, {\blue for every $\epsilon \in (0,\frac{1}{4}]$},
\begin{align*}
R_1 &\geq \max\{I(U_1;X_1|Z,T)\\
&\hspace{24pt}I(U_1;X_1|U_2,X_2,Z,T)\}\\
R_2 &\geq \max\{I(U_2;X_2|Z,T)\\
&\hspace{24pt}I(U_2;X_2|U_1,X_1,Z,T)\}\\
R_1+R_2 &\geq I(U_1,U_2;X_1,X_2|Z,T)\\
R_1+R_{01} &\geq I(U_1;X_1,X_2,Y|Z,T){\blue-g(\epsilon)}\\
R_2+R_{02} &\geq I(U_2;X_1,X_2,Y|Z,T){\blue-g(\epsilon)}\\
R_1+R_2+R_{01}+R_{02} &\geq I(U_1,U_2;X_1,X_2,Y|Z,T){\blue-g(\epsilon)},
\end{align*}
{\blue with $g(\epsilon)=2\sqrt{\epsilon}\left(H_q(X_1,X_2,Y,Z)+\log\frac{(|\mathcal{X}_1||\mathcal{X}_2||\mathcal{Y}||\mathcal{Z}|)}{\epsilon}\right)$ (which tends to $0$ as $\epsilon \to 0$)}, for some p.m.f.  
\begin{align*}
p(x_1,&x_2,z,t,u_1,u_2,y)=\\
&\hspace{12pt}p(x_1,x_2,z)p(t)p(u_1,u_2|x_1,x_2,t)p(y|u_1,u_2,z,t)
\end{align*}
such that 
\begin{align*}
p(u_1|x_1,x_2,z,t)& = p(u_1|x_1,t),\\
p(u_2|x_1,x_2,z,t) &= p(u_2|x_2,t),\\
{\blue||p(x_1,x_2,y,z|t)-q(x_1,x_2,y,z)||_1} &{\blue\leq \epsilon \: \textup{for all} \: t}.
\end{align*} 
\end{theorem}
\if \arxiv 1
Theorem~\ref{thm:encsideOB} is, in fact, a direct consequence of our more general result stated as Theorem~\ref{thm:encsideOBsr}. The details are given in Section~\ref{proof:thm2}. {\blue Notice that the outer bound in Theorem~\ref{thm:encsideOB} is only an epsilon rate region (as in \cite[Section VI-C]{cuff2013distributed}), i.e., the bound holds for every $\epsilon>0$. We do not know if this outer bound is continuous at $\epsilon=0$, i.e., it is unclear whether $\cap_{\epsilon>0}\mathcal{O}_\epsilon=\mathcal{O}_0$, where $\mathcal{O}_\epsilon$ denotes the epsilon rate region in Theorem~\ref{thm:encsideOB} for every $\epsilon\geq 0$ ($g(0):=0$ by continuous extension of the function $g$). Such a continuity argument requires cardinality bounds on both the auxiliary random variables $U_1$ and $U_2$ (in order to ensure the compactness of the  simplex, see \cite[Lemma VI.5]{cuff2013distributed}, \cite[Lemma~6]{YassaeeGA15}) and obtaining such cardinality bounds for the outer bound remains open even for the distributed rate-distortion problem~\cite{tung1979multiterminal,berger1977multiterminal} \footnote{\blue In a preliminary version of this work~\cite{KurriRPP21ISIT}, we incorrectly claimed the continuity of our outer bounds without obtaining cardinality bounds on the auxiliary random variables. We thank Sandeep Pradhan for pointing this out to us.}.}
\fi
\if \arxiv 0
The proof of Theorem~\ref{thm:encsideOB} uses standard arguments and is broadly along the lines of the converses for channel simulation~\cite{cuff2013distributed,YassaeeGA15}. A detailed proof can be found in the extended version~\cite[Appendix D]{KurriRPP21}. 
\fi

{\blue When the random variables $X_1$ and $X_2$ are conditionally independent given $Z$, {\blue and the shared randomness rate $R_{02}$ is unlimited,} we can show that the inner bound of Theorem~\ref{thm:encsideIB} is tight by obtaining an outer bound which is strictly stronger than that in Theorem~\ref{thm:encsideOB} along with cardinality bounds on the auxiliary random variables {\blue that allows us to prove the continuity of the outer bound at $\epsilon=0$}, thereby completely characterizing the rate region {\blue $\mathcal{R}_{\textup{MAC-coord, UL-$S_2$}}^{\text{NO-$S_0$}}$.}}
\begin{theorem}[Rate Region - Conditionally Independent Sources Given the Side Information] \label{thm:indep}
Consider a p.m.f. $q_{X_1X_2ZY}$ such that the random variables $X_1$ and $X_2$ are conditionally independent given $Z$, i.e., $I(X_1;X_2|Z)=0$. Then the rate region {\blue$\mathcal{R}_{\textup{MAC-coord, UL-$S_2$}}^{\emph{NO-$S_0$}}$} is given by the set of all rate tuples {\blue$(R_1,R_2,R_{01})$} such that
\begin{align*}
R_1&\geq I(U_1;X_1|Z,T)\\
R_2&\geq I(U_2;X_2|Z,T)\\
R_1+R_{01}&\geq I(U_1;X_1,Y|X_2,Z,T)\\
\end{align*}
for some p.m.f. 
\begin{align}
&p(x_1,x_2,z,t,u_1,u_2,y)=\nonumber\\
&p(z)p(x_1|z)p(x_2|z)p(t)p(u_1|x_1,t)p(u_2|x_2,t)p(y|u_1,u_2,z,t)
\end{align}
such that $\sum\limits_{u_1,u_2}p(x_1,x_2,z,u_1,u_2,y|t)= q(x_1,x_2,z,y)$, for all $t$, with {\blue$|\mathcal{U}_1| \leq |\mathcal{X}_1||\mathcal{X}_2||\mathcal{Y}||\mathcal{Z}|$, $|\mathcal{U}_2| \leq |\mathcal{U}_1||\mathcal{X}_1||\mathcal{X}_2||\mathcal{Y}||\mathcal{Z}|$, and $|\mathcal{T}| \leq 3$}.
\end{theorem}
In Remark~\ref{remark3} (on page~\pageref{remark3}), we show that the outer bound implicit in Theorem~\ref{thm:indep} is strictly stronger than that of Theorem~\ref{thm:encsideOB} {\blue(even after assuming its continuity at $\epsilon=0$)}. The non-trivial part in the converse of Theorem~\ref{thm:indep} is that we single-letterize the distributed protocol in order to obtain a p.m.f. structure matching that of the inner bound in Theorem~\ref{thm:encsideIB}, particularly leveraging the conditional independence of the sources given the side-information. In general, obtaining single-letter forms matching the inner bound is known to be notoriously difficult for distributed source coding problems~\cite{tung1979multiterminal,berger1977multiterminal}. It is interesting to note that for the case of deterministic function computation when the sources are conditionally independent given the side-information, the inner and outer bounds of Sefidgaran and Tchamkerten~\cite{SefidgaranT16} specialized to the two-user multiple-access network of noiseless links match, analogous to a result of Gastpar~\cite{Gastpar04}. However, for \emph{randomized} function computation, the inner and outer bounds in Theorems \ref{thm:encsideIB} and \ref{thm:encsideOB} {\blue(again after assuming the continuity of the outer bound at $\epsilon=0$ for the sake of comparison)} do not match for sources conditionally independent given side-information and we need a strictly stronger outer bound to show Theorem~\ref{thm:indep}. A detailed proof of Theorem~\ref{thm:indep} is given in Section~\ref{proof:achv-thm3} (achievability) and Section~\ref{proof:conv-thm3} (converse).

\begin{remark}\label{remark:detfn}
{\blue As mentioned earlier (see Remark~\ref{remark:CR}), when $Y$ is a deterministic function of $(X_1,X_2,Z)$, pairwise shared randomness (or common randomness shared by all users) does not have any effect on the communication rates. Indeed, the rate constraints in
Theorem~\ref{thm:indep} involving shared randomness rates become redundant as we show in Appendix~\ref{app:rem5}. We also show in Appendix~\ref{app:rem5} that Theorem~\ref{thm:indep} reduces to the main result of Sefidgaran and Tchamkerten~\cite[Theorem~3]{SefidgaranT16} specialized to the two-user multiple-access network of noiseless links, also reported in \cite[Theorem 3]{SefidgaranT11}. In appendix~\ref{app:rem2}, we show that for deterministic function computation, Theorem~\ref{thm:encsideIB} reduces to Theorem~2 of Sefidgaran and Tchamkerten~\cite{SefidgaranT16} specialized to the two-user multiple-access network of noiseless links, also reported in \cite[Proposition 1]{SefidgaranT11}.}
\end{remark}

{\blue\begin{remark}
Theorem~\ref{thm:indep} recovers the point-to-point channel simulation results~\cite[Theorem~II.1]{cuff2013distributed}, \cite[Theorem~1]{6757002}, and \cite[Theorem~1]{YassaeeGA15} (specialized to a single round of interaction) when $X_2=\emptyset$. 
\end{remark}
}
Sefidgaran and Tchamkerten~\cite{SefidgaranT11} already observed that their inner bound (and hence, our Theorem \ref{thm:encsideIB}, which recovers their inner bound as shown in Appendix~\ref{app:rem2}) is not tight, in general. The deterministic function computation problem of K\"{o}rner and Marton~\cite{KornerM79} illustrates this (see \cite[Example 2]{SefidgaranA11}). For computing the mod-$2$ sum of binary $X_1$ and $X_2$ with symmetric input distribution, K\"{o}rner and Marton~\cite{KornerM79} showed that structured codes can strictly outperform standard random coding schemes. Achievable schemes using algebraic-structured codes exploiting the specific structure of the function to be computed were further explored by Krithivasan and Pradhan~\cite{krithivasan2011distributed} and Atif and Pradhan~\cite{atifsandeep}, where the latter considered the quantum setting.

\section{\blue Encoder Shared Randomness Can Strictly Reduce the Communication Rates} \label{sec:ex1}
In this section, we show that if the encoders share additional independent randomness (see Figure~\ref{fig:encSR}), the communication rates in some cases can be strictly improved, even if the additional randomness is not available to the decoder. This is done via an example for which we first explicitly compute the communication rate region of Theorem \ref{thm:indep} when there is no shared randomness between the encoders, assuming sufficiently large pairwise shared randomness rates. Then we show that a rate pair outside this region is achievable in the presence of shared randomness between the encoders. 
\if \arxiv 1
This motivates the next section where we obtain general inner and outer bounds to the rate coordination region in the presence of shared randomness between the encoders.
\fi

\begin{example}\label{example}
Let $X_1=(X_{11},X_{12})$ be a vector of two independent and uniformly distributed binary random variables. Similarly, let $X_2=(X_{21},X_{22})$ be another vector of two independent and uniformly distributed binary random variables independent of $X_1$. Consider simulating a channel $q_{Y|X_1X_2}$ with $Y=(X_{1J},X_{2J})$, where $J$ is a random variable uniformly distributed on $\{1,2\}$ and independent of $(X_1,X_2)$. For simplicity, we let $Z=\emptyset$, i.e. there is no side information at the decoder. Let us assume unlimited rates $R_{01}$ and $R_{02}$.

When there is no additional shared randomness between the encoders, from Theorem~\ref{thm:indep}, the communication rate region $\mathcal{R}_{\textup{MAC-coord, UL-$(S_1,S_2)$}}^{\textup{NO-$S_0$}}$ is given by the set of all rate pairs $(R_1,R_2)$ such that
\begin{align}
R_1\geq I(U_1;X_1|T), \label{eq:exreg1}\\
R_2\geq I(U_2;X_2|T), \label{eq:exreg2}
\end{align}
for some p.m.f. 
\begin{align}
p(x_1,&x_2,t,u_1,u_2,y)=\nonumber\\ 
&p(x_1)p(x_2)p(t)p(u_1|x_1,t)p(u_2|x_2,t)p(y|u_1,u_2,t) \label{eq:pmfindep}
\end{align}
satisfying
\begin{align} \label{eq:excorrectness}
\sum\limits_{u_1,u_2}p(x_1,x_2,u_1,u_2,y|t)=q(x_1,x_2,y),
\end{align}
for all $t$. 
 The following proposition (proved at the end of this section) explicitly characterizes the communication rate region for this $q_{X_1X_2Y}$.

\begin{prop}\label{prop}
For the joint distribution $q_{X_1X_2Y}$ in Example~\ref{example}, the communication rate region of Theorem~\ref{thm:indep} under sufficiently large shared randomness rates is equal to the region defined by the constraints $R_1\geq 1$, $R_2\geq 1$, and $R_1+R_2\geq 3$.
\end{prop}

Now we show that if there exists an additional source of shared randomness between the encoders, then the rate pair $(R_1,R_2)=(1,1)$ is achievable (see Figure~\ref{region}). In particular, we prove that if this additional shared randomness is of rate at least $1$, then a rate pair $(R_1,R_2)=(1,1)$ is achievable under sufficiently large shared randomness rates $R_{01}$ and $R_{02}$. To see this, notice that both the encoders, using a shared randomness of rate $1$, can sample a sequence $W_1, W_2, \cdots, W_n$ i.i.d. distributed on $\{1,2\}$ such that $p_W(1)=p_W(2)=0.5$ and independent of $(X_1^n,X_2^n)$. Now we invoke Theorem~\ref{thm:encsideIB} with $(X_i,W)$ as the input source to the encoder-$i$, $i=1,2$. This implies that  a rate pair $(R_1,R_2)$ is achievable under sufficiently large shared randomness rates $R_{01}$ and $R_{02}$ if
\begin{align*}
R_1&\geq I(U_1;X_1,W|U_2),\\
R_2&\geq I(U_2;X_2,W|U_1),\\
R_1+R_2&\geq I(U_1,U_2;X_1,X_2,W),
\end{align*}
for some p.m.f.
\begin{align}
p(x_1,&x_2,w,u_1,u_2,y)=\nonumber\\ 
&q(x_1,x_2,w)p(u_1|x_1,w)p(u_2|x_2,w)p(y|u_1,u_2).
\end{align}
\begin{figure}[htbp]
\begin{center}
\includegraphics{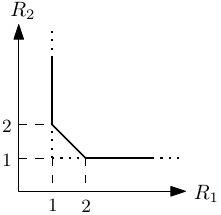}
\end{center}
\caption{The communication rate regions for $X_k=(X_{k1},X_{k2}), k \in \{1,2\}$ ($X_{11},X_{12},X_{21},X_{22}$ are mutually independent, uniform binary random variables) and $Y=(X_{1J},X_{2J})$, where $J$ is uniform on $\{1,2\}$ and independent of $(X_1,X_2)$. The region is defined by the constraints $R_1\geq 1$, $R_2\geq 1$, and $R_1+R_2\geq 3$ (shown via solid line) when the shared randomness between each encoder and the decoder are of sufficiently large rates. When an additional shared randomness is available between the encoders, the region is defined by the constraints $R_1\geq 1$ and $R_2\geq 1$ (shown via dotted line) provided the shared randomness between all the three pairs are of sufficiently large rates\if \arxiv 0 (see the discussion towards the end of this section)\fi \if \arxiv 1 (see Section~\ref{section:exampleoptimality})\fi.}\label{region}
\end{figure}
It is easy to see that $U_1=X_{1W}$ and $U_2=X_{2W}$ satisfy the conditions on the structure of the probability distribution $p(x_1,x_2,w,u_1,u_2,y)$. This gives
\begin{align*}
I(X_{1W};X_1,W|X_{2W})&=1,\\
I(X_{2W};X_2,W|X_{1W})&=1,\\
I(X_{1W},X_{2W};X_1,X_2,W)&=2,
\end{align*}
which imply that  a rate pair $(1,1)$ is achievable, thereby strictly improving over the rate region without encoder shared randomness, defined by the constraints $R_1\geq 1$, $R_2\geq 1$, and $R_1+R_2\geq 3$. 
\if \arxiv 1
This motivates the next section where we analyze the model with shared randomness between the encoders. In Section~\ref{section:exampleoptimality}, we will show that the rate region defined by $R_1\geq 1$, $R_2\geq 1$ is indeed optimal for our example with sufficiently large pairwise shared randomness rates (all three pairs). 
\fi
\end{example}
\begin{remark}\label{remark3}
We remark that the outer bound implicit in Theorem~\ref{thm:indep} is strictly stronger than that of Theorem~\ref{thm:encsideOB} {\blue(even after assuming its continuity at $\epsilon=0$)} for the p.m.f. $q_{X_1X_2ZY}$ in Example~\ref{example}. We first observe that that the communication rate pair $(1,1)$ is contained in the outer bound given by Theorem~\ref{thm:encsideOB} (under unlimited shared randomness rates {\blue and the assumption of continuity at $\epsilon=0$}). To see this, first notice that the choice of $U_1=(X_{1J},J)$, $U_2=(X_{2J},J)$, where $J$ is a random variable uniformly distributed on $\{1,2\}$ and independent of $(X_1,X_2)$, satisfies the conditions on the structure of the p.m.f. in the outer bound of Theorem~\ref{thm:encsideOB}. Now we evaluate the bounds on communication rates in Theorem~\ref{thm:encsideOB} for this choice of auxiliary random variables.
\begin{align*}
\max\{I(X_{1J},J;X_1),I(X_{1J};X_1|X_{2J},J,X_2)\}&=1\\
\max\{I(X_{2J},J;X_2),I(X_{2J};X_2|X_{1J},J,X_1)\}&=1\\
I(X_{1J},X_{2J},J;X_1,X_2)&=2.
\end{align*}  
This implies that a rate pair $(R_1,R_2)=(1,1)$ is contained in the outer bound on the communication rate region implied by Theorem~\ref{thm:encsideOB} (under sufficiently large shared randomness rates {\blue and the assumption of continuity at $\epsilon=0$}). However, by Proposition~\ref{prop}, this rate pair lies outside the communication rate region implied by Theorem~\ref{thm:indep}.
\end{remark}

We conclude the section with a proof of Proposition~\ref{prop}.
\if \arxiv 0
Motivated by this example, we also study the rate region for the problem where encoders have shared randomness that is not available to the decoder. In particular, consider the setting in Section \ref{sec:SMR} augmented with a pairwise shared randomness $S_0$ (uniformly distributed on $[1:2^{nR_{00}}]$) between the two encoders -- see Figure \ref{fig:encSR}. 
\fi%
\begin{proof}[Proof of Proposition~\ref{prop}]
 For the achievability, it suffices to show that the corner points $(1,2)$ and $(2,1)$ are in the region defined by \eqref{eq:exreg1} -- \eqref{eq:excorrectness}. By symmetry, it is enough to show that there exists a p.m.f. $$p(x_1,x_2,u_1,u_2,y)=p(x_1,x_2)p(u_1|x_1)p(u_2|x_2)p(y|u_1,u_2)$$ such that $\sum\limits_{u_1,u_2}p(x_1,x_2,u_1,u_2,y)=q(x_1,x_2,y)$, $I(U_1;X_1)=2,$ and $I(U_2;X_2)=1$.  It is easy to see that $U_1=X_1$ and $U_2=(J,X_{2J})$, where $J$ is a random variable uniformly distributed on $\{1,2\}$ and independent of $(X_1,X_2)$, satisfy the conditions on the structure of the joint probability distribution $p(x_1,x_2,u_1,u_2,y)$. Now $I(U_1;X_1)=H(X_1)=2$ and $I(U_2;X_2)=I(J,X_{2J};X_2)=I(X_{2J};X_2|J)=1$. 

For the converse, it suffices to show that for any p.m.f. $p(x_1,x_2,u_1,u_2,y)= p(x_1,x_2)p(u_1|x_1)p(u_2|x_2)p(y|u_1,u_2)$ such that $\sum\limits_{u_1,u_2}p(x_1,x_2,u_1,u_2,y)=q(x_1,x_2,y)$, we have $I(U_1;X_1)\geq 1$, $I(U_2;X_2)\geq 1$, and $I(U_1;X_1)+I(U_2;X_2)\geq 3$. 
Consider the p.m.f. in \eqref{eq:pmfindep} for a fixed value of $t$, i.e.,
\begin{subequations}\label{eq:pmfind:main}
\begin{equation} 
p(x_1,x_2,u_1,u_2,y)= p(x_1)p(x_2)p(u_1|x_1)p(u_2|x_2)p(y|u_1,u_2) \label{eq:pmfind:A}
\end{equation}
such that 
\begin{equation}
\sum\limits_{u_1,u_2}p(x_1,x_2,u_1,u_2,y)=q(x_1,x_2,y). \label{eq:pmfind:B}
\end{equation}
\end{subequations}
Note that the independence of $X_1$ and $X_2$ along with the long Markov chain $U_1 \to X_1 \to X_2 \to U_2$ implies the independence of $(U_1,X_1)$ and $(U_2,X_2)$. The following three cases now arise based on $H(X_1|U_1)$ and $H(X_2|U_2)$.

\noindent Case $1$: $H(X_1|U_1)=0$,\\
Case $2$: $H(X_2|U_2)=0$,\\
Case $3$: $H(X_1|U_1)>0$ and $H(X_2|U_2)>0$.\\

\noindent\underline{Case $1$} ($H(X_1|U_1)=0$):

We have $I(X_1;U_1)=H(X_1)=2.$
Now we prove that $I(U_2;X_2)\geq 1$. We show this by contradiction. Suppose that $I(U_2;X_2)<1$. Then $H(X_2|U_2)=H(X_2)-I(U_2;X_2)=2-I(U_2;X_2)>1$ and hence there exists a $u_2$ with $P(U_2=u_2)>0$ such that $p_{X_2|U_2=u_2}$ has a support whose size is larger than 2. Notice that the Markov chain $Y\rightarrow (X_1,U_2)\rightarrow X_2$ holds because
\begin{align*}
I(Y;X_2|X_1,U_2)&\leq I(Y,U_1;X_2|X_1,U_2)\\
&=\!I(U_1;X_2|X_1,U_2)+I(Y;X_2|U_1,U_2,X_1)\\
&\leq\!  I(U_1,X_1;U_2,X_2)\!+I(Y;X_1,X_2|U_1,U_2)\\
&=0,
\end{align*}
where the last equality follows because $(U_1,X_1)$ is independent of $(U_2,X_2)$ and the Markov chain $Y\rightarrow (U_1,U_2)\rightarrow (X_1,X_2)$ holds.

Suppose $p_{X_2|U_2=u_2}$ has the support {\blue which is a superset of}  $\{(0,0),(1,0),(0,1)\}$. Consider the induced distribution $p_{Y|X_1=(0,1),U_2=u_2}$. This is well-defined because $P(X_1=(0,1),U_2=u_2)>0$ as $X_1$ is independent of $U_2$ and $P(X_1=(0,1)),P(U_2=u_2)>0$.  Since $P(X_2=(0,1)|X_1=(0,1),U_2=u_2)=P(X_2=(0,1)|U_2=u_2)>0$ and $Y\rightarrow (X_1,U_2)\rightarrow X_2$, we have $P(Y=(1,0)|X_1=(0,1),U_2=u_2)=P(Y=(0,1)|X_1=(0,1),U_2=u_2)=0$.  Since $P(X_2=(1,0)|X_1=(0,1),U_2=u_2)=P(X_2=(1,0)|U_2=u_2)>0$ and $Y\rightarrow (X_1,U_2)\rightarrow X_2$, we have $P(Y=(1,1)|X_1=(0,1),U_2=u_2)=P(Y=(0,0)|X_1=(0,1),U_2=u_2)=0$. This is a contradiction since $p_{Y|X_1=(0,1),U_2=u_2}$ has to be a probability distribution.

Suppose $p_{X_2|U_2=u_2}$ has the support {\blue which is a superset of} $\{(0,0),(1,1),(0,1)\}$. Since $P(X_2=(0,0)|X_1=(0,0),U_2=u_2)=P(X_2=(0,0)|U_2=u_2)>0$ and $Y\rightarrow (X_1,U_2)\rightarrow X_2$, we have $P(Y=(1,0)|X_1=(0,0),U_2=u_2)=P(Y=(0,1)|X_1=(0,0),U_2=u_2)=P(Y=(1,1)|X_1=(0,0),U_2=u_2)=0$. Since $P(X_2=(1,1)|X_1=(0,0),U_2=u_2)=P(X_2=(1,1)|U_2=u_2)>0$ and $Y\rightarrow (X_1,U_2)\rightarrow X_2$, we have $P(Y=(0,0)|X_1=(0,0),U_2=u_2)=0$. This is a contradiction since $p_{Y|X_1=(0,0),U_2=u_2}$ is a probability distribution.

The other supports $\{(1,1),(0,1),(1,0)\}$ and $\{(0,0),(1,1),(1,0)\}$ can be analysed in a similar manner to arrive at a contradiction. Hence, $I(U_2;X_2)\geq 1$.

\noindent\underline{Case $2$} ($H(X_2|U_2)=0$):

By symmetry, the analysis for this case is similar to that of Case~$1$.

\noindent\underline{Case $3$} ($H(X_1|U_1)>0$ and $H(X_2|U_2)>0$):

We prove the following claim in Appendix \ref{apendix:example}.
\begin{claim} \label{claim:ex}
When $H(X_1|U_1)>0$ and $H(X_2|U_2)>0$, there exists a $k\in\{1,2\}$ such that for all $u_1$ and $u_2$ with $P(U_1=u_1)>0$, $P(U_2=u_2)>0$, we have $H(X_{1k}|U_1=u_1)=0$ and $H(X_{2k}|U_2=u_2)=0$.
\end{claim}

Claim~\ref{claim:ex} roughly states that, under Case $3$, $U_1$ and $U_2$ always reveal the $k^{\text{th}}$ components of $X_1$ and $X_2$, respectively, for a fixed $k\in\{1,2\}$. Let us assume $k=1$ without loss of generality, i.e., $U_i$ conveys atleast $X_{i1}$ losslessly, for $i=1,2$. Let $A_i$ be the event ``$U_i \ \text{conveys both}\ X_{i1}\ \text{and}\ X_{i2}\ \text{losslessly}$'', i.e., for $i=1,2$,
\begin{align}
A_i = \underset{u_i: H(X_{i1}|U_i=u_i)=H(X_{i2}|U_i=u_i)=0}{\bigcup} (U_i=u_i).
\end{align}
Let, for $i=1,2$,
\begin{align*}
p_i:= P(A_i) = \sum_{\substack{u_i:\\ H(X_{i1}|U_i=u_i)=H(X_{i2}|U_i=u_i)=0}}P(U_i=u_i).
\end{align*}
For $i=1,2$, it follows that $A_i^c$ is the event ``$U_i$ conveys only $X_{i1}$ losslessly'' from Claim~\ref{claim:ex} along with our assumption that $k=1$. 
Now, in view  of the independence of $(U_1,X_1)$ and $(U_2,X_2)$, it follows that $1-p_1p_2 = 1-P(A_1 \cap A_2) = P(A_1^c \cup A_2^c)$.
Notice that this event $A_1^c \cup A_2^c$ must be a subset of the event $(J=1)$, otherwise the correctness condition \eqref{eq:pmfind:B} is violated. Thus we have
\begin{align}
1-p_1p_2&\leq P(J=1)=0.5\label{eqn:example2}.
\end{align}
Also, we have
\begin{align}
I(U_1&;X_1)\nonumber\\
&= H(X_1)-H(X_1|U_1)\nonumber\\
&= 2-\sum_{u_1}P(U_1=u_1) H(X_1|U_1=u_1) \nonumber\\
&\stackrel{(a)}= 2-\!\!\!\sum_{\substack{u_1:\\ H(X_{11}|U_1=u_1)=0,\\ H(X_{12}|U_1=u_1)>0}}P(U_1=u_1) H(X_1|U_1=u_1) \nonumber\\
&\phantom{www} -\!\!\!\!\sum_{\substack{u_1:\\ H(X_{11}|U_1=u_1)=0, \\H(X_{12}|U_1=u_1)=0}}P(U_1=u_1) H(X_1|U_1=u_1) \nonumber\\
&\geq 2-\!\!\!\!\!\!\!\!\!\!\sum_{\substack{u_1:\\ H(X_{11}|U_1=u_1)=0,\\ H(X_{12}|U_1=u_1)>0}}\!\!\!\!\!\!P(U_1=u_1) H(X_{12}|U_1=u_1) \nonumber\\
&\geq 2-(1-p_1)\nonumber\\
&=1+p_1,
\end{align}
where (a) follows from Claim~\ref{claim:ex} along with our assumption that $k=1$, since $U_i$ reveals atleast $X_{i1}$ losslessly for $i=1,2$.
Similarly, we have $I(U_2;X_2) \geq 1+p_2$. From \eqref{eqn:example2}, we have $p_1p_2\geq 0.5$. Now since 
\begin{align}
\min\limits_{\substack{p_1,p_2:\\0\leq p_1,p_2\leq1\ \text{and}\ p_1p_2\geq 0.5}}{p_1+p_2}>1,
\end{align}
we have $I(U_1;X_1)+I(U_2;X_2)\geq 2+p_1+p_2\geq 3$.
\if \arxiv 0
We show this in the extended version~\cite[Appendix H]{KurriRPP21}. 
\fi
This proves that the communication rate region is equal to the region defined by the constraints $R_1\geq 1$, $R_2\geq 1$, and $R_1+R_2\geq 3$.
\end{proof}
\section{\blue Inner and Outer Bounds on the Rate Region}\label{section:encSR}
{\blue In this section, we present our results for the general setup when a pairwise shared randomness of limited rate $R_{00}$ is present between the encoders (see Figure~\ref{fig:encSR}). In addition, for Example~\ref{example}, we will show that the (achievable) rate region defined by $R_1\geq 1$, $R_2\geq 1$ is indeed optimal with sufficiently large pairwise shared randomness rates (all three pairs). We will also exploit common components~\cite{gacs1973common} between the two sources, i.e. random variables $X_0$ such that there exist deterministic functions $f_1$ and $f_2$ with
\begin{align}
X_0 = f_1(X_1) = f_2(X_2) \: \:  \textup{a.s}.
\end{align}}
The following theorem provides an inner bound to the region $\mathcal{R}_{\textup{MAC-coord}}$.
\begin{theorem}[Inner Bound with Encoder Shared Randomness] \label{thm:encsideIBsr}
Given a p.m.f. $q_{X_1X_2ZY}$, the rate tuple $(R_1,R_2,R_{00},R_{01},R_{02})$ is in $\mathcal{R}_{\textup{MAC-coord}}$ if
\begin{align} 
R_{00} &\geq H(U_0|X_0,T) \label{eq:t4e1}\\
R_1 &\geq I(U_1;X_1,U_0|U_2,Z,T) \label{eq:t4e2} \\
R_2 &\geq I(U_2;X_2,U_0|U_1,Z,T) \\
R_1+R_2 &\geq I(U_1,U_2;X_1,X_2,U_0|Z,T) \\
R_1+R_{01} &\geq I(U_1;X_1,X_2,U_0,Y|Z,T)\notag\\
&\hspace{24pt}-I(U_1;U_2|Z,T) \\
R_2+R_{02} &\geq I(U_2;X_1,X_2,U_0,Y|Z,T)\notag\\
&\hspace{24pt}-I(U_1;U_2|Z,T) \\
R_1+R_2+R_{01} &\geq I(U_1;X_1,X_2,U_0,Y|Z,T)\notag\\
&\hspace{24pt}+I(U_2;X_2,U_0|U_1,Z,T) \\
R_1+R_2+R_{02} &\geq I(U_2;X_1,X_2,U_0,Y|Z,T)\notag\\
&\hspace{24pt}+I(U_1;X_1,U_0|U_2,Z,T) \\
R_1+R_2+R_{01}+R_{02} &\geq I(U_1,U_2;X_1,X_2,U_0,Y|Z,T) \label{eq:t4e9}, 
\end{align}
for some p.m.f. 
\begin{align*}
p(x_1,&x_2,z,t,u_0,u_1,u_2,y)=\\
&p(x_1,x_2,z)p(t)p(u_0|x_0,t)p(u_1|x_1,u_0,t) \\
&\hspace{12pt}\times p(u_2|x_2,u_0,t)p(y|u_1,u_2,z,t)
\end{align*}
such that 
\begin{align*}
\sum\limits_{u_0,u_1,u_2}p(x_1,x_2,u_0,u_1,u_2,y,z|t)=q(x_1,x_2,z,y),\ \text{for all}\ t.
\end{align*}
\end{theorem}
The main idea behind the proof is to make use of the shared randomness between the encoders in order to simulate a common description of $X_0^n$, \emph{viz.} $U_0^n$ at both the encoders approximately distributed according to $q^{(n)}_{U_0^n|X_0^n}(u_0^n|x_0^n):=\prod_{i=1}^np_{U_0|X_0}(u_{0i}|x_{0i})$. Then we invoke Theorem \ref{thm:encsideIB} with $X_j$ replaced by $(X_j,U_0)$ for $j \in \{1,2\}$. A detailed proof is given in Section~\ref{proof:thm4}. 

We now provide an outer bound to the region $\mathcal{R}_{\textup{MAC-coord}}$.
\begin{theorem}[Outer Bound with Encoder Shared Randomness] \label{thm:encsideOBsr}
Given a p.m.f. $q_{X_1X_2ZY}$, any rate tuple $(R_1,R_2,R_{00},R_{01},R_{02})$ in $\mathcal{R}_{\textup{MAC-coord}}$ satisfies, {\blue for every $\epsilon \in (0,\frac{1}{4}]$},
\begin{align}
R_1 &\geq \max\{I(U_1;X_1|U_0,Z,T)\nonumber\\
&\hspace{22pt}I(U_1;X_1|U_0,U_2,X_2,Z,T)\!\}\!\label{eq:t5e1}\\
R_2 &\geq \max\{I(U_2;X_2|U_0,Z,T)\nonumber\\
&\hspace{22pt}I(U_2;X_2|U_0,U_1,X_1,Z,T)\!\}\!\label{eq:t5e2}\\
R_1+R_2 &\geq I(U_1,U_2;X_1,X_2|U_0,Z,T) \label{eq:t5e3}\\
R_1+R_{01} &\geq I(U_1;X_1,X_2,Y|Z,T){\blue-g(\epsilon)} \label{eq:t5e4}\\
R_2+R_{02} &\geq I(U_2;X_1,X_2,Y|Z,T){\blue-g(\epsilon)} \label{eq:t5e5}\\
R_{00}+R_1+R_{01} &\geq I(U_0,U_1;X_1,X_2,Y|Z,T){\blue-g(\epsilon)} \label{eq:t5e6}\\
R_{00}+R_2+R_{02} &\geq I(U_0,U_2;X_1,X_2,Y|Z,T){\blue-g(\epsilon)} \label{eq:t5e7}\\
R_1+R_2+R_{01}+R_{02} &\geq I(U_1,U_2;X_1,X_2,Y|Z,T){\blue-g(\epsilon)} \label{eq:t5e8}\\
R_{00}+R_1+R_2+R_{01}+&R_{02}+g(\epsilon)\nonumber \\
&\geq I(U_0,U_1,U_2;X_1,X_2,Y|Z,T)\!\label{eq:t5e9}
\end{align}
with $g(\epsilon)=2\sqrt{\epsilon}\left(H_q(X_1,X_2,Y,Z)+\log\frac{(|\mathcal{X}_1||\mathcal{X}_2||\mathcal{Y}||\mathcal{Z}|)}{\epsilon}\right)$ (which tends to 0 as $\epsilon \to 0$), for some p.m.f. 
\begin{align}
&p(x_1,x_2,z,t,u_0,u_1,u_2,y)=\nonumber\\
&p(x_1,x_2,z)p(t)p(u_0|t)p(u_1,u_2|x_1,x_2,u_0,t)p(y|u_1,u_2,z,t)\label{eqn:pmfstructure21}
\end{align}
such that 
\begin{align}
p(u_1|x_1,x_2,u_0,z,t)& = p(u_1|x_1,u_0,t)\label{eqn:pmfstructure22}\\
p(u_2|x_1,x_2,u_0,z,t) &= p(u_2|x_2,u_0,t)\label{eqn:pmfstructure23}\\
{\blue||p(x_1,x_2,y,z|t)-q(x_1,x_2,y,z)||_1} &{\blue\leq \epsilon \: \textup{for all} \: t.}
\end{align} 
\end{theorem}
A detailed proof is given in Section~\ref{proof:thm5}.
{\blue Once again, the outer bound in Theorem~\ref{thm:encsideOBsr} is only an epsilon rate region, whose continuity at $\epsilon=0$ is unknown.} 

When the random variables $X_1$ and $X_2$ are conditionally independent given $Z$, {\blue and the shared randomness rates $(R_{00},R_{01},R_{02})$ are unlimited,} we obtain a potentially stronger outer bound {\blue which is also continuous at $\epsilon=0$}.
\begin{theorem}[Outer Bound - Conditionally Independent Sources Given the Side Information] \label{Theorem:outbndindp}
Consider a p.m.f. $q_{X_1X_2ZY}$ such that the random variables $X_1$ and $X_2$ are conditionally independent given $Z$, i.e., $I(X_1;X_2|Z)=0$. Then any rate tuple in {\blue$\mathcal{R}_{\textup{MAC-coord, UL-$(S_0,S_1,S_2)$}}$} satisfies
\begin{align*}
R_1&\geq I(U_0,U_1;X_1|Z,T)\\
R_2&\geq I(U_0,U_2;X_2|Z,T)\\
R_1+R_{2}&\geq I(U_0,U_1,U_2;X_1,X_2|Z,T)\\
\end{align*}
for some p.m.f. 
\begin{align}
&p(x_1,x_2,z,t,u_0,u_1,u_2,y)=\nonumber\\
&p(z)p(x_1|z)p(x_2|z)p(t)p(u_0|t)p(u_1|x_1,u_0,t)p(u_2|x_2,u_0,t)\nonumber\\
&\hspace{1cm}\times p(y|u_1,u_2,z,t)
\end{align}
such that $\sum\limits_{u_0,u_1,u_2}p(x_1,x_2,z,u_0,u_1,u_2,y|t)= q(x_1,x_2,z,y)$, for all $t$, with  {\blue $|\mathcal{U}_0| \leq |\mathcal{X}_1||\mathcal{X}_2||\mathcal{Y}||\mathcal{Z}|$, $|\mathcal{U}_1| \leq |\mathcal{U}_0||\mathcal{X}_1||\mathcal{X}_2||\mathcal{Y}||\mathcal{Z}|$, $|\mathcal{U}_2| \leq |\mathcal{U}_0||\mathcal{X}_1||\mathcal{X}_2||\mathcal{Y}||\mathcal{Z}|$, and $|\mathcal{T}| \leq 3$}.
\end{theorem}
Notice that the improvement is in the structure of the p.m.f. compared to that of Theorem~\ref{thm:encsideOBsr}. A detailed proof can be found in Section~\ref{proof:thm6}.

\subsection{Optimal Region for Example \ref{example} with unlimited shared randomness between all three pairs}\label{section:exampleoptimality}
Here, we show that in the setting of Example \ref{example}, the region $R_1 \geq 1$ and $R_2 \geq 1$ is indeed the optimal rate region (not just achievable as shown in Section \ref{sec:ex1}) with unlimited pairwise shared randomness (all three pairs). The achievability can also be inferred from Theorem \ref{thm:encsideIBsr} with the choice of $U_0=J$, $U_1=X_{1J}$ and $U_2=X_{2J}$, where $J$ is a random variable uniformly distributed on $\{1,2\}$ and independent of $(X_1,X_2)$. To prove the converse, first note that Theorem~\ref{Theorem:outbndindp} (with $Z=\emptyset$) implies that any achievable rate pair $(R_1,R_2)$ must satisfy  
\begin{align*}
R_1&\geq I(U_0,U_1;X_1|T)\\
R_2&\geq I(U_0,U_2;X_2|T) \\
R_1+R_2 &\geq I(U_0,U_1,U_2;X_1,X_2|T),
\end{align*}
for some p.m.f. 
\begin{align*}
p(x_1,&x_2,t,u_0,u_1,u_2,y)=p(x_1)p(x_2)p(t)p(u_0|t)\\
&\hspace{12pt} \times p(u_1|x_1,u_0,t)p(u_2|x_2,u_0,t)p(y|u_1,u_2,t)
\end{align*}
such that $\sum\limits_{u_0,u_1,u_2}p(x_1,x_2,u_0,u_1,u_2,y|t)=q(x_1,x_2,y)$, for all $t$. 

For the converse, it suffices to show that for any p.m.f. 
\begin{align*}
p(x_1,&x_2,u_0,u_1,u_2,y)=\nonumber\\
& p(x_1)p(x_2)p(u_0)p(u_1|x_1,u_0)p(u_2|x_2,u_0)p(y|u_1,u_2)
\end{align*}
 with $\sum\limits_{u_0,u_1,u_2}p(x_1,x_2,u_0,u_1,u_2,y)=q(x_1,x_2,y)$, we have $I(U_0,U_1;X_1)\geq 1$ and $I(U_0,U_2;X_2)\geq 1$.  Equivalently, by marginalizing away $U_2$ and letting $(U_0,U_1) \triangleq U$, it suffices to show that for any p.m.f. $$p(x_1,x_2,u,y)= p(x_1)p(x_2)p(u|x_1)p(y|u,x_2)$$ with $\sum\limits_{u}p(x_1,x_2,u,y)=q(x_1,x_2,y)$, we have $I(U;X_1)\geq 1$ (The condition $I(U_0,U_2;X_2)\geq 1$ can be shown analogously.). This can be established by proving that for each $U=u$ such that $H(X_1|U=u)>0$, there exists $k(u) \in \{1,2\}$ such that $H(X_{1k}|U=u)=0$. Indeed this yields
\begin{align}
I(U;X_1) &= H(X_1)-H(X_1|U) \notag\\
&= 2-\sum_{u} P(U=u) H(X_1|U=u) \notag\\
&= 2-\sum_{u} P(U=u) \left(H(X_{1k}|U=u) \right. \notag\\
&\phantom{wwwwwwwww} \left.+H(X_{1k^{\prime}}|U=u,X_{1k})\right) \notag\\
&\geq 2-\sum_{u} P(U=u) (0+1) \notag\\
&= 1.
\end{align}
We have the following claim.
\begin{claim}\label{claim2}
When $H(X_1|U)>0$, for all $u$ with $H(X_1|U=u)>0$ there exists a $k$ in $\{1,2\}$ such that $H(X_{1k}|U=u)=0$.
\end{claim}
\begin{proof}[Proof of Claim~\ref{claim2}]
We prove this by contradiction. Suppose $H(X_{1i}|U=u)>0$, for $i=1,2$. Then the support of $p_{X_1|U=u}$ has to be a superset of either $\{(0,1),(1,0)\}$ or $\{(0,0),(1,1)\}$. However, it turns out that, the support cannot be a superset of $\{(0,0),(1,1)\}$. To see this, first notice that $P(X_1=(0,0)|U=u,X_2=(0,0))=P(X_1=(0,0)|U=u)>0$, where the equality follows from the independence of $(U,X_1)$ and $X_2$. Similarly, $P(X_1=(1,1)|U=u,X_2=(0,0))>0$. Now since $X_{1J}-(U,X_2)-X_1$, we have $P(X_{1J}=0|U=u,X_2=(0,0))=P(X_{1J}=0|U=u,X_2=(0,0),X_1=(1,1))=0$, where the last equality follows from the correctness of the output $Y=(X_{1J},X_{2J})$. Similarly, $P(X_{1J}=1|U=u,X_2=(0,0))=P(X_{1J}=1|U=u,X_2=(0,0),X_1=(0,0))=0$. This is a contradiction since $p_{Y_1|U=u,X_2=(0,0)}$ has to be a probability distribution. The only other possibility is that $p_{X_1|U=u}$ has a support that is a superset of $\{(0,1),(1,0)\}$. Consider the induced distribution $p_{Y|X_2=(0,1),U=u}$. This is well-defined because $P(X_2=(0,1),U=u)>0$ as $X_2$ is independent of $U$ and $P(X_2=(0,1)),P(U=u)>0$.  Since $P(X_1=(0,1)|X_2=(0,1),U=u)=P(X_1=(0,1)|U=u)>0$ and $Y\rightarrow (X_2,U)\rightarrow X_1$, we have $P(Y=(1,0)|X_2=(0,1),U=u)=P(Y=(0,1)|X_2=(0,1),U=u)=0$.  Since $P(X_1=(1,0)|X_2=(0,1),U=u)=P(X_1=(1,0)|U=u)>0$ and $Y\rightarrow (X_2,U)\rightarrow X_1$, we have $P(Y=(1,1)|X_2=(0,1),U=u)=P(Y=(0,0)|X_2=(0,1),U=u)=0$. This is a contradiction since $p_{Y|X_2=(0,1),U=u}$ has to be a valid probability distribution.
\end{proof}

This proves that the optimal communication rate region with unlimited pairwise shared randomness (all three pairs) \linebreak$\mathcal{R}_{\textup{MAC-coord, UL-$(S_0,S_1,S_2)$}}$ is indeed defined by the constraints $R_1\geq 1$ and $R_2\geq 1$.

\section{Proofs}\label{proofs}
\subsection{Achievability Proofs}\label{Finalproofsachv}
{\blue The proof of Theorem~\ref{thm:encsideIB} is in Appendix~\ref{app:pfThm1}.}

\subsubsection{Achievability Proof of Theorem \ref{thm:indep}}\label{proof:achv-thm3}
\begin{proof}
{\blue We argue that achievability follows from Theorem~\ref{thm:encsideIB} by enforcing the constraint $p(x_1,x_2,z) = p(z)p(x_1|z)p(x_2|z)$ along with unlimited shared randomness rate $R_{02}$}. In this case, the joint distribution on $(X_1,X_2,Z,T,U_1,U_2,Y)$ decomposes as $p(x_1,x_2,z,t,u_1,u_2,y)=p(z)p(x_1|z)p(x_2|z)p(t)p(u_1|x_1,t)p(u_2|x_2,t)p(y|u_1,u_2,z,t)$. 
{\blue We first write down the inner bound from Theorem~\ref{thm:encsideIB} for this case.}
\begin{align}
R_1 &\geq I(U_1;X_1|U_2,Z,T) \nonumber\\
&= I(U_1;X_1|Z,T)\notag \\
R_2 &\geq I(U_2;X_2|U_1,Z,T) \nonumber\\
&= I(U_2;X_2|Z,T)\notag\\
R_1+R_2 &\geq I(U_1,U_2;X_1,X_2|Z,T) \notag\\
&= I(U_1;X_1|Z,T)+I(U_2;X_2|Z,T)\notag\\
R_1+R_{01} &\geq I(U_1;X_1,X_2,Y|Z,T)\nonumber\\
&\hspace{24pt}-I(U_1;U_2|Z,T) \notag\\
&= I(U_1;X_1,X_2,Y|Z,T)\notag \\
&=I(U_1;X_1,Y|X_2,Z,T)\notag\\
R_2+R_{02} &\geq I(U_2;X_1,X_2,Y|Z,T)\nonumber\\
&\hspace{24pt}-I(U_1;U_2|Z,T) \notag\\
&= I(U_2;X_1,X_2,Y|Z,T)\notag\\
&=I(U_2;X_2,Y|X_1,Z,T)\notag\\
R_1+R_2+R_{01} &\geq I(U_1;X_1,\!X_2,\!Y|Z,T)\!\notag\\
&\hspace{24pt}+\!I(U_2;X_2|U_1,Z,T) \notag\\
&= I(U_1;X_1,X_2,Y|Z)\nonumber\\
&\hspace{24pt}+I(U_2;X_2|Z,T) \notag\\
R_1+R_2+R_{02} &\geq I(U_2;X_1,\!X_2,\!Y|Z,T)\!\notag\\
&\hspace{24pt}+\!I(U_1;X_1|U_2,Z,T) \notag\\
&= I(U_2;X_1,X_2,Y|Z)\nonumber\\
&\hspace{24pt}+I(U_1;X_1|Z,T)\notag\\
R_1+R_2+R_{01}+R_{02} &\geq I(U_1,U_2;X_1,X_2,Y|Z,T) \label{eqn:thm3acvproof1}.
\end{align}
Notice that the constraints on $R_1+R_2$, $R_1+R_2+R_{01}$ and $R_1+R_2+R_{02}$ are redundant. {\blue This region is an inner bound to $\mathcal{R}_{\textup{MAC-coord}}^{\textup{NO-$S_0$}}$. Considering only the constraints that exclude $R_{02}$ completes the achievability of Theorem~\ref{thm:indep}.}
 \end{proof}

\subsubsection{Proof of Theorem~\ref{thm:encsideIBsr}}\label{proof:thm4}
\begin{proof} 
{\blue We prove the achievability for $|\mathcal{T}|=1$, and the rest of the proof follows by using time sharing argument similar to that in Theorem~\ref{thm:encsideIB} (in particular, see the paragraph after Lemma~\ref{lem:a2}}). Firstly, we show that the shared randomness between the encoders can be harnessed to simulate $U_0^n$ approximately distributed according to $q^{(n)}_{U_0^n|X_0^n}(u_0^n|x_0^n)=\prod_{i=1}^np_{U_0|X_0}(u_{0i}|x_{0i})$. The rate of shared randomness needed here will turn out to be $R_{00} \geq H(U_0|X_0)$, i.e. the constraint \eqref{eq:t4e1} on $R_{00}$ in Theorem \ref{thm:encsideIBsr}.

We make use of the following setup from Cuff~\cite[Corollary VII.5]{cuff2013distributed}.
\begin{figure}[h]
\centering
\begin{tikzpicture}[thick]
\node (e1) at (1,1.5) [rectangle, draw, minimum height=1.0cm, text width=0.8cm]{$f(\cdot)$}; 
\node (c1) at (4,1.5) [rectangle, draw, minimum height=1.0cm, text width=1.5cm]{$P_{V|U_0,X_0}$}; 
\node (s1) at (-1.3,3.5) {$X_0^n$};
\draw[->] (e1.east) -- (c1) node[midway, above] {$U_0^n$};
\draw[->] (c1.east) --++(1.75,0) node[midway, above] {$V^n$};
\draw[<-] (e1) --++(-1.75,0) node[left]{$R_{00}$};
\draw[->] (s1)  -| (e1);
\draw[->] (s1)  -| (c1);
\end{tikzpicture}
\caption{Implications of soft covering~\cite{cuff2013distributed}} \label{fig:cuffsim}
\end{figure}
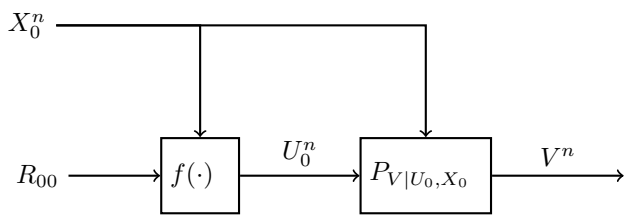
Let $X_0^n$ be an i.i.d. sequence with distribution $p_{X_0}$ and $p_{U_0|X_0}p_{V|U_0,X_0}$ be a memoryless channel. A deterministic encoder $f(\cdot)$ receives both $X_0^n$ and a uniformly distributed random variable $S_0 \in [1:2^{nR_{00}}]$. Cuff~\cite{cuff2013distributed} gave the following sufficient condition on the rate $R_{00}$ so that the induced distribution on the channel output in Figure~\ref{fig:cuffsim} is i.i.d. in the limit of large $n$.
\begin{lemma} \cite[Corollary VII.5]{cuff2013distributed} \label{lem:unifV}
Let $p_{V^n}$ be the induced distribution on the channel output in Figure~\ref{fig:cuffsim} and $q_{V^n}$ be the i.i.d. output distribution specified by $\sum_{u_0,x_0} p_{X_0}p_{U_0|X_0}p_{V|U_0,X_0}$. Then if
\begin{align}
R \geq I(X_0,U_0;V)-H(X_0),
\end{align}
we have
\begin{align}
\lim_{n \to \infty} ||p_{V^n}-q_{V^n}||_{1}=0.
\end{align}
\end{lemma}

For our purposes, we take $p_{V|U_0,X_0}$ to be an identity channel i.e. $V=(U_0,X_0)$ in Lemma \ref{lem:unifV}. Hence 
\begin{align}\label{eqn:cufflemma}
\lim_{n \to \infty} ||p_{U_0^n,X_0^n}-q_{U_0^n,X_0^n}||_{1}=0,
\end{align}
provided that $R_{00}$ satisfies
\begin{align}
R_{00} &\geq I(X_0,U_0;X_0,U_0)-H(X_0) \notag\\
&= H(U_0|X_0).
\end{align}

Suppose if $(X_1^n,X_2^n,U_0^n)$ are generated exactly i.i.d., then we can invoke Theorem \ref{thm:encsideIB} with $X_j$ replaced by $(X_j,U_0)$ for $j \in \{1,2\}$ and it can be verified that this exactly yields the eight rate constraints \eqref{eq:t4e2}--\eqref{eq:t4e9} involving $(R_1,R_2,R_{01},R_{02})$ in Theorem \ref{thm:encsideIBsr}. Here, we show that the same set of rate constraints suffices even if $(U_0^n,X_0^n)$ are approximately i.i.d. in the sense of \eqref{eqn:cufflemma}. All that remains to be shown is that under the rate constraints in Theorem~\ref{thm:encsideIBsr}, there exists a sequence of codes with an induced p.m.f. on $(X_1^n,X_2^n,U_0^n,Z^n,Y^n)$ such that the total variation distance between this p.m.f. and the desired i.i.d. p.m.f. $\prod q_{X_1,X_2,U_0,Z,Y}$ vanishes in the limit of large blocklength $n$. This can be argued out as follows.

If $(X_1^n,X_2^n,U_0^n)$ were exactly i.i.d., by invoking Theorem~\ref{thm:encsideIB} there exists
\begin{align*}
&p(m_1,m_2,y^n|x_1^n,x_2^n,u_0^n,z^n,s_1,s_2)\\
&:=p(m_1|s_1,x_1^n,u_0^n)p(m_2|s_2,x_2^n,u_0^n)p(y^n|m_1,m_2,s_1,s_2,z^n)
\end{align*}
such that
\begin{align}
||p_{X_1^n,X_2^n,U_0^n,Z^n,Y^n}-\prod q_{X_1X_2U_0ZY}||_{1} \leq \epsilon. \label{eq:tv1}
\end{align}
On the other hand, let $\tilde{p}(x_1^n,x_2^n,u_0^n,z^n,s_1,s_2)$ be the distribution in which $(U_0^n,X_0^n)$ are approximately i.i.d. in the sense of \eqref{eqn:cufflemma}. Let us define
\begin{align}
\tilde{p}(x_1^n,&x_2^n,u_0^n,z^n,s_1,s_2,m_1,m_2,y^n)\nonumber\\:=
&\tilde{p}(x_1^n,x_2^n,u_0^n,z^n,s_1,s_2)\nonumber\\
&\hspace{12pt}\times p(m_1,m_2,y^n|x_1^n,x_2^n,u_0^n,z^n,s_1,s_2).
\end{align} 
Now, the total variation distance of interest can be bounded using triangle inequality as
\begin{align}
&||\tilde{p}_{X_1^n,X_2^n,{U}_0^n,Z^n,Y^n}-\prod q_{X_1X_2U_0ZY}||_{1} \notag\\
&\phantom{www}\leq ||\tilde{p}_{X_1^n,X_2^n,{U}_0^n,Z^n,Y^n}-p_{X_1^n,X_2^n,U_0^n,Z^n,Y^n}||_{1} \notag\\
&\phantom{wwwwww} +||p_{X_1^n,X_2^n,U_0^n,Z^n,Y^n}-\prod q_{X_1X_2U_0Y}||_{1} \notag\\
&\phantom{www}\leq ||\tilde{p}_{X_1^n,X_2^n,{U}_0^n,Z^n,Y^n}-p_{X_1^n,X_2^n,U_0^n,Z^n,Y^n}||_{1}+\epsilon, \label{eq:tv3}
\end{align}
where \eqref{eq:tv3} follows from \eqref{eq:tv1}. Next consider the first term on the RHS of \eqref{eq:tv3}.
\begin{align}
&||\tilde{p}_{X_1^n,X_2^n,U_0^n,Z^n,Y^n}-p_{X_1^n,X_2^n,U_0^n,Z^n,Y^n}||_{1} \notag\\
& =\sum_{\substack{x_1^n,x_2^n,u_0^n,\\ z^n,y^n}} |\sum_{\substack{m_1,m_2,\\ s_1,s_2}} \tilde{p}(x_1^n,x_2^n,u_0^n,z^n,y^n,m_1,m_2,s_1,s_2) \notag\\
&\hspace{50pt}-p(x_1^n,x_2^n,u_0^n,z^n,y^n,m_1,m_2,s_1,s_2)| \notag\\
& \stackrel{(a)}\leq \sum_{\substack{x_1^n,x_2^n,u_0^n,z^n,\\ y^n,m_1,m_2,s_1,s_2}} \frac{p(y^n,m_1,m_2|x_1^n,x_2^n,u_0^n,z^n,s_1,s_2)}{2^{n(R_{01}+R_{02})}} \notag\\
&\hspace{50pt}\times \left|\tilde{p}(x_1^n,x_2^n,u_0^n,z^n)-p(x_1^n,x_2^n,u_0^n,z^n)\right| \notag\\
& = \sum_{x_1^n,x_2^n,u_0^n,z^n} \left|\tilde{p}(x_1^n,x_2^n,u_0^n,z^n)-p(x_1^n,x_2^n,u_0^n,z^n)\right| \notag\\
& \stackrel{(b)}\leq\sum_{x_1^n,x_2^n,x_0^n,u_0^n,z^n} |\tilde{p}(x_1^n,x_2^n,x_0^n,u_0^n,z^n)-\notag\\
&\hspace{50pt}p(x_1^n,x_2^n,x_0^n,u_0^n,z^n)| \notag\\
&=\sum_{x_1^n,x_2^n,x_0^n,u_0^n,z^n} p(x_1^n,x_2^n,x_0^n,z^n)|\tilde{p}(u_0^n|x_0^n)-p(u_0^n|x_0^n)|\notag\\
&= \sum_{x_0^n,u_0^n} \left|\tilde{p}(x_0^n,u_0^n)-p(x_0^n,u_0^n)\right| \notag\\
&\leq \epsilon, \label{eq:tv4}
\end{align}
{\blue where (a) follows from the triangle inequality and the fact that $S_1$ and $S_2$ are uniformly distributed on their respective ranges, (b) follows since $X_0$ is a common function of $X_1$ and $X_2$, and \eqref{eq:tv4} follows from the statement below \eqref{eq:tv1}.}
From expressions \eqref{eq:tv3} and \eqref{eq:tv4}, we conclude that
\begin{align*}
||\tilde{p}_{X_1^n,X_2^n,{U}_0^n,Z^n,Y^n}-\prod q_{X_1,X_2,U_0,Z,Y}||_{1} \leq 2\epsilon.
\end{align*}
\end{proof}

\subsection{Converse Proofs}\label{Finalproofsconv}
We will show that Theorem~\ref{thm:encsideOB} is a direct consequence of Theorem~\ref{thm:encsideOBsr}, which we prove first.

\subsubsection{Proof of Theorem \ref{thm:encsideOBsr}}\label{proof:thm5}
\begin{proof}

Consider a code that induces a joint distribution on $(X_1^n,X_2^n,Z^n,Y^n)$ such that
\begin{align}
\lVert p_{X_1^n,X_2^n,Z^n,Y^n}-q^{(n)}_{X_1X_2ZY}\rVert_1< \epsilon,\label{eqn:totalSR}.
 \end{align}
for $\epsilon\in(0,\frac{1}{4}]$.
For ease of notation, for a vector $A^n$, we write $A_{\sim i} \triangleq (A^{i-1},A_{i+1}^n)$. Also let $M_{[0:2]} \triangleq (M_0,M_1,M_2)$ and $S_{[0:2]} \triangleq (S_0,S_1,S_2)$. We quote the following lemmas that will prove useful in the outer bound.
\begin{lemma} \cite[Lemma 6]{cervia2020strong} \label{lem:c1}
Let $p_{X^n}$ be such that $||p_{X^n}-q_X^{(n)}||_1 \leq \epsilon$, where $q^{(n)}_X(x^n)=\prod_{i=1}^nq_X(x_i)$, then
\begin{align}
\sum_{i=1}^n I_p(X_i;X_{\sim i}) \leq ng_1(\epsilon),
\end{align}
where ${\blue g_1(\epsilon)=2\sqrt{\epsilon}\left(H(X)+\log{|\mathcal{X}|+\log{\frac{1}{\sqrt{\epsilon}}}}\right)}\to 0$ as $\epsilon \to 0$.
\end{lemma}

\begin{lemma} \cite[Lemma VI.3]{cuff2013distributed} \label{lem:c2}
Let $p_{X^n}$ be such that $||p_{X^n}-q_X^{(n)}||_1 \leq \epsilon$, where $q^{(n)}_X(x^n)=\prod_{i=1}^nq_X(x_i)$, then for any RV $T \in [1:n]$ independent of $X^n$,
\begin{align}
I_p(T;X_T) \leq g_2(\epsilon),
\end{align}
where ${\blue g_2(\epsilon)=4\epsilon\left(\log{|\mathcal{X}|}+\log{\frac{1}{\epsilon}}\right)} \to 0$ as $\epsilon \to 0$.
\end{lemma}
{\blue Notice that for $\epsilon\in(0,\frac{1}{4}]$, we have $\max\{g_1(\epsilon),g_2(\epsilon)\}\leq g(\epsilon):=2\sqrt{\epsilon}\left(H(X)+\log{|\mathcal{X}|+2\log{\frac{1}{\sqrt{\epsilon}}}}\right)$. So, we can replace $g_1(\epsilon)$ and $g_2(\epsilon)$ in Lemmas~\ref{lem:c1} and \ref{lem:c2} by $g(\epsilon)$, which also tends to $0$ as $\epsilon\to 0$.} 
 Let us consider the first lower bound on $R_j$ for $j \in \{1,2\}$ in \eqref{eq:t5e1}--\eqref{eq:t5e2}.
\begin{align}
nR_j &\geq H(M_j)\notag\\
& \geq H(M_j|S_0,S_j,Z^n) \notag\\
&\geq I(M_j;X_j^n|S_0,S_j,Z^n)\notag\\
& \stackrel{(a)}= I(M_j,S_j;X_j^n|S_0,Z^n) \notag\\
&= \sum_{i=1}^n I(M_j,S_j;X_{ji}|X_{j,i+1}^n,S_0,Z_i,Z_{\sim i}) \notag\\
&\stackrel{(b)}= \sum_{i=1}^n I(M_j,S_j,X_{j,i+1}^n,Z_{\sim i};X_{ji}|S_0,Z_i) \notag\\
&\geq \sum_{i=1}^n I(M_j,S_j,X_{j,i+1}^n,Z^{i-1};X_{ji}|S_0,Z_i) \notag\\
&\stackrel{(c)}\geq \sum_{i=1}^n I(U_{ji};X_{ji}|U_{0i},Z_i) \notag\\
&\stackrel{(d)}= nI(U_{jT};X_{jT}|U_{0T},Z_T,T) \notag\\
&\stackrel{(e)}= nI(U_j;X_j|U_0,Z,T), \label{eq:refe1SR}
\end{align}
where (a) follows since $S_j$ is independent of $(S_0,X_j^n,Z^n)$, (b) follows since $(X_{ji},Z_i),i=1,\dots,n,$ are jointly i.i.d. and $S_0$ is independent of $(X_j^n,Z^n)$, (c) follows by defining $U_{0i} = S_0$, $U_{1i} = (M_1,S_1,X_{1,i+1}^n,Z^{i-1})$ and $U_{2i} = (M_2,S_2,X_{2,i+1}^n)$, (d) follows by introducing a uniform time-sharing random variable $T \in [1:n]$ that is independent of everything else, while (e) follows by defining $U_0:=U_{0T}$, $U_1:=U_{1T}$, $U_2:=U_{2T}$, $X_1:=X_{1T}$, $X_2:=X_{2T}$, $Y:=Y_T$ and $Z:=Z_T$. The second lower bound on $R_1$ in \eqref{eq:t5e1} is obtained as follows.
\begin{align}
nR_1&\geq H(M_1)\notag\\
&\geq H(M_1|X_2^n,S_{[0:2]},Z^n) \notag\\
&\geq I(M_1;X_1^n|X_2^n,S_{[0:2]},Z^n) \notag\\
&\stackrel{(a)} = I(M_{[1:2]};X_1^n|X_2^n,S_{[0:2]},Z^n) \notag\\
&\stackrel{(b)}= I(M_{[1:2]},S_{[1:2]};X_1^n|S_0,X_2^n,Z^n) \notag\\
&= \sum_{i=1}^n I(M_{[1:2]},S_{[1:2]};X_{1i}|X_{1,i+1}^n,S_0,X_2^n,Z^n) \notag\\
&\stackrel{(c)}= \sum_{i=1}^n I(M_{[1:2]},S_{[1:2]},X_{1,i+1}^n,X_{2,\sim i},Z_{\sim i};X_{1i}|S_0,X_{2i},Z_i) \notag\\
&\geq  \sum_{i=1}^n I(M_{[1:2]}\!,S_{[1:2]}\!,X_{1,i+1}^n\!,X_{2,i+1}^n\!,Z^{i-1};X_{1i}|S_0,X_{2i},Z_i) \notag\\
&\stackrel{(d)}= \sum_{i=1}^n I(U_{1i},U_{2i};X_{1i}|U_{0i},X_{2i},Z_i) \notag\\
&\stackrel{(e)}= \sum_{i=1}^n I(U_{1i};X_{1i}|U_{0i},U_{2i},X_{2i},Z_i) \notag\\
&= nI(U_{1T};X_{1T}|U_{0T},U_{2T},X_{2T},Z_T,T) \notag\\
&= nI(U_1;X_1|U_0,U_2,X_2,Z,T), \label{eq:refe2SR}
\end{align}
where (a) follows from the Markov chain $M_2 \to (X_2^n,S_0,S_2) \to (M_1,S_1,X_1^n,Z^n)$, (b) follows since $S_{[1:2]}$ is independent of $(S_0,X_1^n,X_2^n,Z^n)$, {\blue(c) follows since $S_0$ is independent of $(X_1^n,X_2^n,Z^n)$ and the fact that $(X_{1i},X_{2i},Z_i), i=1,\cdots,n$ are i.i.d., (d) follows from the identifications $U_{0i} = S_0$, $U_{1i} = (M_1,S_1,X_{1,i+1}^n,Z^{i-1})$ and $U_{2i} = (M_2,S_2,X_{2,i+1}^n)$}, while (e) follows from the Markov chain $U_{2i} \to (U_{0i},X_{2i}) \to (X_{1i},Z_i)$. Similarly, we obtain
\begin{align}
nR_2 &\geq nI(U_2;X_2|U_0,U_1,X_1,Z,T).
\end{align}
We next derive the lower bound on $(R_1+R_2)$ in \eqref{eq:t5e3}.
\begin{align}
&n(R_1+R_2) \notag\\
& \geq H(M_{[1:2]})\notag\\
&\geq H(M_{[1:2]}|S_{[0:2]},Z^n) \notag\\
& \geq I(M_{[1:2]};X_1^n,X_2^n|S_{[0:2]},Z^n) \notag\\
&\stackrel{(a)}= I(M_{[1:2]},S_{[1:2]};X_1^n,X_2^n|S_0,Z^n) \notag\\
&= \sum_{i=1}^n I(M_{[1:2]},S_{[1:2]};X_{1i},X_{2i}|X_{1,i+1}^n,X_{2,i+1}^n,S_0,Z^n) \notag\\
&\stackrel{(b)}= \sum_{i=1}^n I(M_{[1:2]},\!S_{[1:2]},\!X_{1,i\!+\!1}^n,\!X_{2,i\!+\!1}^n,\!Z_{\sim i};X_{1i},\!X_{2i}|S_0,Z_i) \notag\\
&\geq \sum_{i=1}^n I(M_{[1:2]},\!S_{[1:2]},\!X_{1,i\!+\!1}^n,\!X_{2,i\!+\!1}^n,\!Z^{i-1};X_{1i},\!X_{2i}|S_0,Z_i) \notag\\
&= \sum_{i=1}^n I(U_{1i},U_{2i};X_{1i},X_{2i}|U_{0i},Z_i) \notag\\
&= nI(U_{1T},U_{2T};X_{1T},X_{2T}|U_{0T},Z_T,T) \notag\\
& = nI(U_1,U_2;X_1,X_2|U_0,Z,T), \label{eq:refe3SR}
\end{align}
where (a) follows since $S_{[1:2]}$ is independent of $(S_0,X_1^n,X_2^n,Z^n)$ while (b) follows since $(X_{1i},X_{2i},Z_i),i=1,\dots,n,$ are jointly i.i.d. and $S_0$ is independent of $ (X_1^n,X_2^n,Z^n)$. We next derive the lower bound on $(R_j+R_{0j})$ for $j\in\{1,2\}$ in \eqref{eq:t5e4}--\eqref{eq:t5e5}.
\begin{align}
&n(R_j+R_{0j}) \notag\\ 
&\geq H(M_j,S_j)\notag\\
& \geq H(M_j,S_j|Z^n) \notag\\
& \geq I(M_j,S_j;X_1^n,X_2^n,Y^n|Z^n) \notag\\
&= \sum_{i=1}^n I(M_j,S_j;X_{1i},X_{2i},Y_i|X_{1,i+1}^n,X_{2,i+1}^n,Y_{i+1}^n,Z^n) \notag\\
&= \sum_{i=1}^n I(M_j,S_j,X_{1,i+1}^n,X_{2,i+1}^n,Y_{i+1}^n,Z_{\sim i};X_{1i},X_{2i},Y_i|Z_i) \notag\\
&\phantom{www} - \sum_{i=1}^n I(X_{1,i+1}^n,X_{2,i+1}^n,Y_{i+1}^n,Z_{\sim i};X_{1i},X_{2i},Y_i|Z_i) \notag\\
&\stackrel{(a)}\geq \sum_{i=1}^n I(M_j,S_j,X_{j,i+1}^n,Z^{i-1};X_{1i},X_{2i},Y_i|Z_i)-ng(\epsilon) \notag\\
&\stackrel{(b)}\geq \sum_{i=1}^n I(U_{ji};X_{1i},X_{2i},Y_i|Z_i)-ng(\epsilon) \notag\\
&= nI(U_{jT};X_{1T},X_{2T},Y_T|Z_T,T) -ng(\epsilon) \notag\\
& = nI(U_j;X_1,X_2,Y|Z,T)-ng(\epsilon), \label{eq:refe4}
\end{align}
 where (a) follows since
 \begin{align}
\sum_{i=1}^n &I(X_{1,i+1}^n,X_{2,i+1}^n,Y_{i+1}^n,Z_{\sim i};X_{1i},X_{2i},Y_i|Z_i) \notag\\
 & \leq \sum_{i=1}^n I(X_{1 \sim i},X_{2 \sim i},Y_{\sim i},Z_{\sim i};X_{1i},X_{2i},Y_i,Z_i)\notag\\
 & \leq ng(\epsilon)\label{eqn:newrepeat}
 \end{align}
by \eqref{eqn:totalSR} and Lemma \ref{lem:c1}, {\blue while (b) follows from the identifications $U_{1i} = (M_1,S_1,X_{1,i+1}^n,Z^{i-1})$ and $U_{2i} = (M_2,S_2,X_{2,i+1}^n)$}.

We next derive the lower bound on $(R_{00}+R_j+R_{0j})$ for $j\in\{1,2\}$ in \eqref{eq:t5e6}--\eqref{eq:t5e7}.
\begin{align}
&n(R_{00}+R_j+R_{0j}) \notag\\ 
&\geq H(M_j,S_0,S_j)\notag\\
& \geq H(M_j,S_0,S_j|Z^n) \notag\\
& \geq I(M_j,S_0,S_j;X_1^n,X_2^n,Y^n|Z^n) \notag\\
&= \sum_{i=1}^n I(M_j,S_0,S_j;X_{1i},X_{2i},Y_i|X_{1,i+1}^n,X_{2,i+1}^n,Y_{i+1}^n,Z^n) \notag\\
&= \sum_{i=1}^n I(M_j,\!S_0,\!S_j,\!X_{1,i+1}^n,\!X_{2,i+1}^n,\!Y_{i+1}^n,\!Z_{\sim i};X_{1i},\!X_{2i},\!Y_i|Z_i) \notag\\
&\phantom{www} - \sum_{i=1}^n I(X_{1,i+1}^n,X_{2,i+1}^n,Y_{i+1}^n,Z_{\sim i};X_{1i},X_{2i},Y_i|Z_i) \notag\\
&\stackrel{(a)}\geq \sum_{i=1}^n I(M_j,S_0,S_j,X_{j,i+1}^n,Z^{i-1};X_{1i},X_{2i},Y_i|Z_i)-ng(\epsilon) \notag\\
&\stackrel{(b)}\geq \sum_{i=1}^n I(U_{0i},U_{ji};X_{1i},X_{2i},Y_i|Z_i)-ng(\epsilon) \notag\\
&= nI(U_{0T},U_{jT};X_{1T},X_{2T},Y_T|Z_T,T) -ng(\epsilon) \notag\\
& = nI(U_0,U_j;X_1,X_2,Y|Z,T)-ng(\epsilon), \label{eq:refe4SR}
\end{align}
where (a) follows from \eqref{eqn:newrepeat}, {\blue while (b) follows from the identifications $U_{0i} = S_0$, $U_{1i} = (M_1,S_1,X_{1,i+1}^n,Z^{i-1})$ and $U_{2i} = (M_2,S_2,X_{2,i+1}^n)$}.
We next derive the lower bound on $(R_1+R_2+R_{01}+R_{02})$ in \eqref{eq:t5e8}.
\begin{align}
&n(R_1+R_2+R_{01}+R_{02}) \notag\\
& \geq H(M_{[1:2]},S_{[1:2]})\notag\\
&\geq H(M_{[1:2]},S_{[1:2]}|Z^n) \notag\\
& \geq I(M_{[1:2]},S_{[1:2]};X_1^n,X_2^n,Y^n|Z^n) \notag\\
&= \sum_{i=1}^n I(M_{[1:2]},S_{[1:2]};X_{1i},X_{2i},Y_i|X_{1,i+1}^n,X_{2,i+1}^n,Y_{i+1}^n,Z^n) \notag\\
&= \sum_{i=1}^n I(M_{[1:2]},\!S_{[1:2]},\!X_{1,i\!+\!1}^n,\!X_{2,i\!+\!1}^n,\!Y_{i\!+\!1}^n,\!Z_{\sim i};X_{1i},\!X_{2i},\!Y_i|Z_i) \notag\\
&\phantom{wwww} - \sum_{i=1}^n I(X_{1,i+1}^n,X_{2,i+1}^n,Y_{i+1}^n,Z_{\sim i};X_{1i},X_{2i},Y_i|Z_i) \notag\\
&\stackrel{(a)}\geq \sum_{i=1}^n I(M_{[1:2]},\!S_{[1:2]},\!X_{1,i\!+\!1}^n,\!X_{2,i\!+\!1}^n,\!Z^{i-1};X_{1i},\!X_{2i},\!Y_i|Z_i) \notag\\
&\phantom{wwww}-ng(\epsilon) \notag\\
&= \sum_{i=1}^n I(U_{1i},U_{2i};X_{1i},X_{2i},Y_i|Z_i) - ng(\epsilon) \notag\\
&= nI(U_{1T},U_{2T};X_{1T},X_{2T},Y_T|Z_T,T)-ng(\epsilon) \notag\\
& = nI(U_1,U_2;X_1,X_2,Y|Z,T)-ng(\epsilon), \label{eq:refe5}
\end{align}
where (a) follows from \eqref{eqn:newrepeat}.
We finally derive the lower bound on $(R_{00}+R_1+R_2+R_{01}+R_{02})$ in \eqref{eq:t5e9}.
\begin{align}
&n(R_{00}+R_1+R_2+R_{01}+R_{02}) \notag\\
& \geq H(M_{[1:2]},S_{[0:2]})\notag\\
&\geq H(M_{[1:2]},S_{[0:2]}|Z^n) \notag\\
& \geq I(M_{[1:2]},S_{[0:2]};X_1^n,X_2^n,Y^n|Z^n) \notag\\
&= \sum_{i=1}^n I(M_{[1:2]},S_{[0:2]};X_{1i},X_{2i},Y_i|X_{1,i+1}^n,X_{2,i+1}^n,Y_{i+1}^n,Z^n) \notag\\
&= \sum_{i=1}^n I(M_{[1:2]},\!S_{[0:2]},\!X_{1,i\!+\!1}^n,\!X_{2,i\!+\!1}^n,\!Y_{i\!+\!1}^n,\!Z_{\sim i};X_{1i},\!X_{2i},\!Y_i|Z_i) \notag\\
&\phantom{wwww} - \sum_{i=1}^n I(X_{1,i+1}^n,X_{2,i+1}^n,Y_{i+1}^n,Z_{\sim i};X_{1i},X_{2i},Y_i|Z_i) \notag\\
&\stackrel{(a)}\geq \sum_{i=1}^n I(M_{[1:2]},\!S_{[0:2]},\!X_{1,i\!+\!1}^n,\!X_{2,i\!+\!1}^n,\!Z^{i-1};X_{1i},\!X_{2i},\!Y_i|Z_i) \notag\\
&\phantom{wwww}-ng(\epsilon) \notag\\
&= \sum_{i=1}^n I(U_{0i},U_{1i},U_{2i};X_{1i},X_{2i},Y_i|Z_i) - ng(\epsilon) \notag\\
&= nI(U_{0T},U_{1T},U_{2T};X_{1T},X_{2T},Y_T|Z_T,T)-ng(\epsilon) \notag\\
& = nI(U_0,U_1,U_2;X_1,X_2,Y|Z,T)-ng(\epsilon), \label{eq:refe5SR}
\end{align}
where (a) follows from \eqref{eqn:newrepeat}.

We now prove the Markov chains $U_{1i} \to (X_{1i},U_{0i}) \to (X_{2i},Z_i)$, $U_{2i} \to (X_{2i},U_{0i}) \to (X_{1i},Z_i)$, $(U_{1i},U_{2i}) \to (X_{1i},X_{2i},U_{0i}) \to Z_i$ and $Y_i \to (U_{1i},U_{2i},Z_i) \to (X_{1i},X_{2i},U_{0i})$. Note that this implies that the joint p.m.f. satisfies \eqref{eqn:pmfstructure21}-\eqref{eqn:pmfstructure23}. Recall the auxiliary random variable identifications
\begin{align*}
U_{0i} &= S_0, \\
U_{1i} &= (M_1,S_1,X_{1,i+1}^n,Z^{i-1}), \\
U_{2i} &= (M_2,S_2,X_{2,i+1}^n).
\end{align*}
Let us first show that $(U_{1i},U_{2i}) \to (X_{1i},X_{2i},U_{0i}) \to Z_i$.
\begin{align}
&I(U_{1i},U_{2i};Z_i|X_{1i},X_{2i},U_{0i}) \notag\\
&\phantom{w} = I(M_{[1:2]},S_{[1:2]},X_{1,i+1}^n,X_{2,i+1}^n,Z^{i-1};Z_i|X_{1i},X_{2i},S_0) \notag\\
&\phantom{w} \leq I(M_{[1:2]},S_{[1:2]},X_{1 \sim i},X_{2 \sim i},Z^{i-1};Z_i|X_{1i},X_{2i},S_0) \notag\\
&\phantom{w} = I(S_{[1:2]},X_{1 \sim i},X_{2 \sim i},Z^{i-1};Z_i|X_{1i},X_{2i},S_0) \notag\\
&\phantom{www}+I(M_{[1:2]};Z_i|X_1^n,X_2^n,S_{[0:2]},Z^{i-1}) \notag\\
&\phantom{w}\stackrel{(a)}= 0+0 = 0,
\end{align}
where in (a), the first term is zero since $S_{[0:2]}$ is independent of $(X_1^n,X_2^n,Z^n)$ and $(X_{1i},X_{2i},Z_i),i=1,\dots,n,$ are jointly i.i.d., while the second term is zero because of the Markov chain $M_{[1:2]} \to (X_1^n,X_2^n,S_{[0:2]}) \to Z^n$.

We now show that $U_{1i} \to (X_{1i},U_{0i}) \to (X_{2i},Z_i)$ holds.
\begin{align}
&I(U_{1i};X_{2i},Z_i|X_{1i},U_{0i}) \notag\\
&\phantom{w} = I(M_1,S_1,X_{1,i+1}^n,Z^{i-1};X_{2i},Z_i|X_{1i},S_0) \notag\\
&\phantom{w} \leq I(M_1,S_1,X_{1 \sim i},Z^{i-1};X_{2i},Z_i|X_{1i},S_0) \notag\\
&\phantom{w} = I(S_1,X_{1 \sim i},Z^{i-1};X_{2i},Z_i|X_{1i},S_0) \notag\\
&\phantom{www}+I(M_1;X_{2i},Z_i|X_1^n,S_1,Z^{i-1},S_0) \notag\\
&\phantom{w}\stackrel{(a)}= 0+0 = 0,
\end{align}
where in (a), the first term is zero since $(S_0,S_1)$ is independent of $(X_1^n,X_2^n,Z^n)$ and $(X_{1i},X_{2i},Z_i),i=1,\dots,n,$ are jointly i.i.d., while the second term is zero because of the Markov chain $M_1 \to (X_1^n,S_0,S_1) \to (X_2^n,Z^n)$. In a similar fashion, we can show that $U_{2i} \to (X_{2i},U_{0i}) \to (X_{1i},Z_i)$ holds as well. Finally, let us show that $Y_i \to (U_{1i},U_{2i},Z_i) \to (X_{1i},X_{2i},U_{0i})$ is a Markov chain as well.
\begin{align*}
&I(Y_i;X_{1i},X_{2i},U_{0i}|U_{1i},U_{2i},Z_i) \notag\\
& = I(Y_i;X_{1i},X_{2i},S_0|M_{[1:2]},S_{[1:2]},X_{1,i+1}^n,X_{2,i+1}^n,Z^{i-1},Z_i) \notag\\
& \leq I(Y_i,Z_{i+1}^n;X_{1i},X_{2i},S_0|M_{[1:2]},S_{[1:2]},X_{1,i+1}^n,X_{2,i+1}^n,Z^i) \notag\\
& = I(Z_{i+1}^n;X_{1i},X_{2i},S_0|M_{[1:2]},S_{[1:2]},X_{1,i+1}^n,X_{2,i+1}^n,Z^i) \notag\\
&\phantom{ww} +I(Y_i;X_{1i},X_{2i},S_0|M_{[1:2]},S_{[1:2]},X_{1,i+1}^n,X_{2,i+1}^n,Z^n) \notag\\
& \stackrel{(a)}= I(Z_{i+1}^n;X_{1i},X_{2i},S_0|M_{[1:2]},S_{[1:2]},X_{1,i+1}^n,X_{2,i+1}^n,Z^i)\\
& \leq I(Z_{i+1}^n;X_1^i,X_2^i,S_0|M_{[1:2]},S_{[1:2]},X_{1,i+1}^n,X_{2,i+1}^n,Z^i) \notag\\
& \leq I(Z_{i+1}^n;X_1^i,X_2^i,M_{[1:2]},S_0|S_{[1:2]},X_{1,i+1}^n,X_{2,i+1}^n,Z^i) \notag\\
& = I(Z_{i+1}^n;X_1^i,X_2^i,S_0|S_{[1:2]},X_{1,i+1}^n,X_{2,i+1}^n,Z^i) \notag\\
&\phantom{ww} +I(Z_{i+1}^n;M_{[1:2]}|S_{[0:2]},X_1^n,X_2^n,Z^i) \notag\\
& \stackrel{(b)} = 0+0 =0,
\end{align*}
where (a) follows from the Markov chain $Y^n \to (M_{[1:2]},S_{[1:2]},Z^n) \to (S_0,X_1^n,X_2^n)$, and in (b), the first term is zero since $S_{[0:2]}$ is independent of $(X_1^n,X_2^n,Z^n)$ and $(X_{1i},X_{2i},Z_i),i=1,\dots,n,$ are jointly i.i.d., while the second term is zero because of the Markov chain $M_{[1:2]} \to (X_1^n,X_2^n,S_{[0:2]}) \to Z^n$. Note that, for all $t\in[1:n]$, 
\begin{align}
\lVert p_{X_{1t},X_{2t},Z_t,Y_t|T=t}&-q_{X_1,X_2,Z,Y|T=t} \rVert_1\nonumber\\
&\stackrel{(a)}\leq\lVert p_{X_1^n,X_2^n,Z^n,Y^n}-q^{(n)}_{X_1,X_2,Z,Y} \rVert_1\nonumber\\
&\stackrel{(b)}<\epsilon,\nonumber
\end{align}
where $(a)$ follows from \cite[Lemma V.1]{cuff2013distributed} and the fact that $T$ is independent of all other random variables, while $(b)$ follows from \eqref{eqn:total1}. Therefore,
\begin{align*}
\lVert &p_{X_1,X_2,Z,Y|T=t}-q_{X_1,X_2,Z,Y|T=t} \rVert_1\\
&= \lVert p_{X_{1T},X_{2T},Z_T,Y_T|T=t}-q_{X_1,X_2,Z,Y|T=t} \rVert_1\\
&= \lVert p_{X_{1t},X_{2t},Z_t,Y_t|T=t}-q_{X_1,X_2,Z,Y|T=t} \rVert_1\\
&\leq \epsilon.
\end{align*}
\end{proof}

\subsubsection{Proof of Theorem \ref{thm:encsideOB}}\label{proof:thm2}
\begin{proof} 
 Note that by Definitions \ref{def:ach}, if $(R_1,R_2,R_{01},R_{02}) \in {\blue\mathcal{R}_{\textup{MAC-coord}}^{\text{NO-$S_0$}}}$, then we have $(R_1,R_2,0,R_{01},R_{02})$ $\in$ $ {\blue\mathcal{R}_{\textup{MAC-coord}}}$. In other words, an outer bound for the case without shared randomness between the encoders can be obtained by invoking Theorem~\ref{thm:encsideOBsr} with $R_{00}=0$. The proof of the implication would be complete by proving that the resulting outer bound is exactly identical to Theorem~\ref{thm:encsideOB}.

With $R_{00}=0$, Theorem~\ref{thm:encsideOBsr} reduces to the set of $(R_1,R_2,R_{01},R_{02})$ such that
\begin{align*}
R_1 &\geq \max\{I(U_1;X_1|U_0,Z,T),\\
&\hspace{22pt}I(U_1;X_1|U_0,U_2,X_2,Z,T)\}\!\\
R_2 &\geq \max\{I(U_2;X_2|U_0,Z,T),\\
&\hspace{22pt}I(U_2;X_2|U_0,U_1,X_1,Z,T)\}\!\\
R_1+R_2 &\geq I(U_1,U_2;X_1,X_2|U_0,Z,T)\\
R_1+R_{01} &\geq I(U_0,U_1;X_1,X_2,Y|Z,T)\\
R_2+R_{02} &\geq I(U_0,U_2;X_1,X_2,Y|Z,T)\\
R_1+R_2+R_{01}+R_{02} &\geq I(U_0,U_1,U_2;X_1,X_2,Y|Z,T),
\end{align*}
for some p.m.f. 
\begin{align*}
&p(x_1,x_2,z,t,u_0,u_1,u_2,y)=\\
&p(x_1,x_2,z)p(t)p(u_0|t)p(u_1,u_2|x_1,x_2,u_0,t)p(y|u_1,u_2,z,t)
\end{align*}
such that 
\begin{align*}
p(u_1|x_1,x_2,u_0,z,t)& = p(u_1|x_1,u_0,t)\\
p(u_2|x_1,x_2,u_0,z,t) &= p(u_2|x_2,u_0,t)\\
{\blue||p(x_1,x_2,y,z|t)-q(x_1,x_2,y,z)||_1} &{\blue\leq \epsilon \: \textup{for all} \: t.}
\end{align*}
Let us define $U_j^{\prime} \triangleq (U_0,U_j)$ for $j=1,2$. Using the Markov chains $U_0 \to T \to (X_1,X_2,Z)$, $Y \to (U_1,U_2,Z,T) \to (X_1,X_2,U_0)$, $(U_1,U_2) \to (X_1,X_2,U_0,T) \to Z$, $U_1 \to (X_1,U_0,T) \to (X_2,Z)$, and $U_2 \to (X_2,U_0,T) \to (X_1,Z)$, the region can be simplified as the set of $(R_1,R_2,R_{01},R_{02})$ such that
\begin{align*}
R_1 &\geq \max\{I(U_1^{\prime};X_1|Z,T),\\
&\hspace{22pt}I(U_1^{\prime};X_1|U_2^{\prime},X_2,Z,T)\}\!\\
R_2 &\geq \max\{I(U_2^{\prime};X_2|Z,T),\\
&\hspace{22pt}I(U_2^{\prime};X_2|U_1^{\prime},X_1,Z,T)\}\!\\
R_1+R_2 &\geq I(U_1^{\prime},U_2^{\prime};X_1,X_2|Z,T)\\
R_1+R_{01} &\geq I(U_1^{\prime};X_1,X_2,Y|Z,T)\\
R_2+R_{02} &\geq I(U_2^{\prime};X_1,X_2,Y|Z,T)\\
R_1+R_2+R_{01}+R_{02} &\geq I(U_1^{\prime},U_2^{\prime};X_1,X_2,Y|Z,T),
\end{align*}
for some p.m.f.
\begin{align*}
p(x_1,&x_2,z,t,u_1^{\prime},u_2^{\prime},y)=\\
&\hspace{12pt}p(x_1,x_2,z)p(t)p(u_1^{\prime},u_2^{\prime}|x_1,x_2,t)p(y|u_1^{\prime},u_2^{\prime},z,t)
\end{align*}
such that 
\begin{align*}
p(u_1^{\prime}|x_1,x_2,z,t)& = p(u_1^{\prime}|x_1,t)\\
p(u_2^{\prime}|x_1,x_2,z,t) &= p(u_2^{\prime}|x_2,t)\\
||p(x_1,x_2,y,z|t)-q(x_1,x_2,y,z)||_1 &\leq \epsilon \: \textup{for all} \: t.
\end{align*}
\end{proof}

\subsubsection{Converse Proof of Theorem~\ref{thm:indep}}\label{proof:conv-thm3}
\begin{proof}
{\blue We prove the following lemma later.
\begin{lemma} \label{lem:indep}
Consider a p.m.f. $q_{X_1X_2ZY}$ such that the random variables $X_1$ and $X_2$ are conditionally independent given $Z$. Then any rate tuple $(R_1,R_2,R_{00},R_{01},R_{02})$ in $\mathcal{R}_{\textup{MAC-coord}}$ satisfies, {\blue for every $\epsilon \in (0,\frac{1}{4}]$},
\begin{align}
R_1&\geq I(U_1;X_1|Z,T) \label{eq:lem4e1}\\
R_2&\geq I(U_2;X_2|Z,T) \label{eq:lem4e2}\\
R_1+R_{01}&\geq I(U_1;X_1,Y|X_2,Z,T)-g(\epsilon) \label{eq:lem4e3}\\
R_2+R_{02}&\geq I(U_2;X_2,Y|X_1,Z,T)-g(\epsilon) \label{eq:lem4e4}\\
R_1+R_2+R_{01}+R_{02} &\geq I(U_1,U_2;X_1,X_2,Y|Z,T)-g(\epsilon) \label{eq:lem4e5},
\end{align}
with $g(\epsilon)=2\sqrt{\epsilon}\left(H_q(X_1,X_2,Y,Z)+\log\frac{(|\mathcal{X}_1||\mathcal{X}_2||\mathcal{Y}||\mathcal{Z}|)}{\epsilon}\right)$ (which tends to 0 as $\epsilon \to 0$), for some p.m.f. 
\begin{align}\label{lempmfstructure1}
&p(x_1,x_2,z,t,u_1,u_2,y)=\nonumber\\
&p(z)p(x_1|z)p(x_2|z)p(t)p(u_1|x_1,t)p(u_2|x_2,t)p(y|u_1,u_2,z,t)
\end{align}
such that $||p(x_1,x_2,y,z|t)-q(x_1,x_2,y,z)||_1 \leq \epsilon \: \textup{for all} \: t$.
\end{lemma}
}
{\blue Notice that Lemma~\ref{lem:indep} gives an \emph{almost} matching outer bound to the inner bound in Theorem~\ref{thm:encsideIB} for $X_1$ and $X_2$ conditionally independent given $Z$ (see \eqref{eqn:thm3acvproof1}), in the sense that continuity at $\epsilon=0$ is the only difference between them. When the shared randomness rate $R_{02}$ is sufficiently large, we can in fact prove that such a continuity argument holds. We argue this below. 

First we prove the following lemma concerning the cardinality bounds on the auxiliary random variables $U_1$ and $U_2$ under sufficiently large shared randomness rate $R_{02}$. Let $S_{\epsilon}$ denote the set of rate tuples $(R_1,R_2,R_{01})$ that satisfy, for every $\epsilon>0$, 
\begin{align*}
R_1&\geq I(U_1;X_1|Z,T)\\
R_2&\geq I(U_2;X_2|Z,T)\\
R_1+R_{01}&\geq I(U_1;X_1,Y|X_2,Z,T)-g(\epsilon)\\
\end{align*}
with $g(\epsilon)=2\sqrt{\epsilon}\left(H_q(X_1,X_2,Y,Z)+\log\frac{(|\mathcal{X}_1||\mathcal{X}_2||\mathcal{Y}||\mathcal{Z}|)}{\epsilon}\right)$ (which tends to 0 as $\epsilon \to 0$), for some p.m.f. 
\begin{align}\label{pmfstructure1}
&p(x_1,x_2,z,t,u_1,u_2,y)=\nonumber\\
&p(z)p(x_1|z)p(x_2|z)p(t)p(u_1|x_1,t)p(u_2|x_2,t)p(y|u_1,u_2,z,t)
\end{align}
such that $||p(x_1,x_2,y,z|t)-q(x_1,x_2,y,z)||_1 \leq \epsilon \: \textup{for all} \: t$.
\begin{lemma} \label{lem:cardbnd}
The size of the auxiliary random variable alphabets in $S_{\epsilon}$ can be restricted to:
\begin{align*}
|\mathcal{U}_1| &\leq |\mathcal{X}_1||\mathcal{X}_2||\mathcal{Y}||\mathcal{Z}|,\\
|\mathcal{U}_2| &\leq |\mathcal{U}_1||\mathcal{X}_1||\mathcal{X}_2||\mathcal{Y}||\mathcal{Z}|,\\
|\mathcal{T}| &\leq 3.
\end{align*}
\end{lemma}
\begin{proof}
See Appendix~\ref{app:cardbnd}.
\end{proof}

From Lemmas \ref{lem:indep} and \ref{lem:cardbnd}, the converse for Theorem~\ref{thm:indep} follows using the continuity of total variation distance and mutual information in the probability simplex along the same lines as \cite[Lemma VI.5]{cuff2013distributed},\cite[Lemma 6]{YassaeeGA15}. It remains to prove Lemma~\ref{lem:indep}.
}
\begin{proof}[Proof of Lemma~\ref{lem:indep}]
Consider a code that induces a joint distribution on $(X_1^n,X_2^n,Z^n,Y^n)$ such that
\begin{align}
\lVert p_{X_1^n,X_2^n,Z^n,Y^n}-q^{(n)}_{X_1X_2ZY}\rVert_1< \epsilon\label{eqn:total1}.
\end{align}

Let us first prove the lower bound on $R_j$ for $j \in \{1,2\}$ in \eqref{eq:lem4e1}--\eqref{eq:lem4e2}.
\begin{align}
nR_j &\geq H(M_j)\notag\\
& \geq H(M_j|S_j,Z^n) \notag\\
&\geq I(M_j;X_j^n|S_j,Z^n)\notag\\
& \stackrel{(a)}= I(M_j,S_j;X_j^n|Z^n) \notag\\
&= \sum_{i=1}^n I(M_j,S_j;X_{ji}|X_{j,i+1}^n,Z^n) \notag\\
&\stackrel{(b)}= \sum_{i=1}^n I(M_j,S_j,Z_{\sim i},X_{j,i+1}^n;X_{ji}|Z_i) \notag\\
&\geq \sum_{i=1}^n I(M_j,S_j,X_{j,i+1}^n,Z^{i-1};X_{ji}|Z_i) \notag\\
&\stackrel{(c)}\geq \sum_{i=1}^n I(U_{ji};X_{ji}|Z_i) \notag\\
&\stackrel{(d)}= nI(U_{jT};X_{jT}|Z_T,T) \notag\\
&\stackrel{(e)}= nI(U_j;X_j|Z,T), \label{eq:refe1i}
\end{align}
where (a) follows since $S_j$ is independent of $(X_j^n,Z^n)$, (b) follows since $(X_{ji},Z_i),i=1,\dots,n,$ are jointly i.i.d., (c) follows by defining $U_{1i} = (M_1,S_1,X_{1,i+1}^n,Z^{i-1})$ and $U_{2i} = (M_2,S_2,X_{2,i+1}^n)$, (d) follows by introducing a uniform time-sharing random variable $T \in [1:n]$ that is independent of everything else, while (e) follows by defining $U_1:=U_{1T}$, $U_2:=U_{2T}$, $X_1:=X_{1T}$, $X_2:=X_{2T}$, $Y:=Y_T$ and $Z:=Z_T$.

We next derive the lower bound on $(R_1+R_{01})$ in \eqref{eq:lem4e3}.
\begin{align}
&n(R_1+R_{01}) \notag\\ 
&\geq H(M_1,S_1)\notag\\
& \geq H(M_1,S_1|X_2^n,Z^n) \notag\\
& \geq I(M_1,S_1;X_1^n,Y^n|X_2^n,Z^n) \notag\\
&= \sum_{i=1}^n I(M_1,S_1;X_{1i},Y_i|X_{1,i+1}^n,Y_{i+1}^n,X_2^n,Z^n) \notag\\
&= \sum_{i=1}^n I(M_1,S_1,X_{1,i+1}^n,Y_{i+1}^n,X_{2 \sim i},Z_{\sim i};X_{1i},Y_i|X_{2i},Z_i) \notag\\
&\phantom{www} - \sum_{i=1}^n I(X_{1,i+1}^n,Y_{i+1}^n,X_{2 \sim i},Z_{\sim i};X_{1i},Y_i|X_{2i},Z_i) \notag\\
&\stackrel{(a)}\geq \sum_{i=1}^n I(M_1,S_1,X_{1,i+1}^n,Z^{i-1};X_{1i},Y_i|X_{2i},Z_i)-ng(\epsilon) \notag\\
&= \sum_{i=1}^n I(U_{1i};X_{1i},Y_i|X_{2i},Z_i)-ng(\epsilon) \notag\\
&= nI(U_{1T};X_{1T},Y_T|X_{2T},Z_T,T) -ng(\epsilon) \notag\\
&= nI(U_1;X_1,Y|X_2,Z,T)-ng(\epsilon), \label{eq:refe4i}
\end{align}
where (a) follows since
\begin{align}
\sum_{i=1}^n &I(X_{1,i+1}^n,Y_{i+1}^n,X_{2 \sim i},Z_{\sim i};X_{1i},Y_i|X_{2i},Z_i) \notag\\
& \leq \sum_{i=1}^n I(X_{1 \sim i},X_{2 \sim i},Y_{\sim i},Z_{\sim i};X_{1i},X_{2i},Y_i,Z_i)\notag\\
& \leq ng(\epsilon)\notag
\end{align}
by \eqref{eqn:total1} and Lemma \ref{lem:c1}. The bound $n(R_2+R_{02}) \geq I(U_2;X_2,Y|X_1,Z,T)$ follows in a similar manner.

For the lower bound on $(R_1+R_2+R_{01}+R_{02})$ in \eqref{eq:lem4e5}, we proceed as follows.
\begin{align}
&n(R_1+R_2+R_{01}+R_{02}) \notag\\
& \geq H(M_{[1:2]},S_{[1:2]})\notag\\
&\geq H(M_{[1:2]},S_{[1:2]}|Z^n) \notag\\
& \geq I(M_{[1:2]},S_{[1:2]};X_1^n,X_2^n,Y^n|Z^n) \notag\\
&= \sum_{i=1}^n I(M_{[1:2]},S_{[1:2]};X_{1i},X_{2i},Y_i|X_{1,i+1}^n,X_{2,i+1}^n,Y_{i+1}^n,Z^n) \notag\\
&= \sum_{i=1}^n I(M_{[1:2]},\!S_{[1:2]},\!X_{1,i\!+\!1}^n,\!X_{2,i\!+\!1}^n,\!Y_{i\!+\!1}^n,\!Z_{\sim i};X_{1i},\!X_{2i},\!Y_i|Z_i) \notag\\
&\phantom{wwww} - \sum_{i=1}^n I(X_{1,i+1}^n,X_{2,i+1}^n,Y_{i+1}^n,Z_{\sim i};X_{1i},X_{2i},Y_i|Z_i) \notag\\
&\stackrel{(a)}\geq \sum_{i=1}^n I(M_{[1:2]},\!S_{[1:2]},\!X_{1,i\!+\!1}^n,\!X_{2,i\!+\!1}^n,\!Z^{i-1};X_{1i},\!X_{2i},\!Y_i|Z_i) \notag\\
&\phantom{wwww}-ng(\epsilon) \notag\\
&= \sum_{i=1}^n I(U_{1i},U_{2i};X_{1i},X_{2i},Y_i|Z_i) - ng(\epsilon) \notag\\
&= nI(U_{1T},U_{2T};X_{1T},X_{2T},Y_T|Z_T,T)-ng(\epsilon) \notag\\
& = nI(U_1,U_2;X_1,X_2,Y|Z,T)-ng(\epsilon), \label{eq:refe5i}
\end{align}
where (a) follows since
\begin{align}
\sum_{i=1}^n &I(X_{1,i+1}^n,X_{2,i+1}^n,Y_{i+1}^n,Z_{\sim i};X_{1i},X_{2i},Y_i|Z_i) \notag\\
 & \leq \sum_{i=1}^n I(X_{1 \sim i},X_{2 \sim i},Y_{\sim i},Z_{\sim i};X_{1i},X_{2i},Y_i,Z_i)\notag\\
 & \leq ng(\epsilon)\notag
\end{align}
by \eqref{eqn:total1} and Lemma \ref{lem:c1}.

We now prove the Markov chains $U_{1i} \to X_{1i} \to (X_{2i},Z_i)$, $U_{2i} \to X_{2i} \to (U_{1i},X_{1i},Z_i)$, $(U_{1i},U_{2i})-(X_{1i},X_{2i})-Z_i$, and $Y_i \to (U_{1i},U_{2i},Z_i) \to (X_{1i},X_{2i})$. Note that this implies that the joint p.m.f. satisfies \eqref{pmfstructure1}.
Recall the auxiliary random variable identifications
\begin{align}
U_{1i} &= (M_1,S_1,X_{1,i+1}^n,Z^{i-1}), \\
U_{2i} &= (M_2,S_2,X_{2,i+1}^n).
\end{align}
Let us first show that $(U_{1i},U_{2i}) \to (X_{1i},X_{2i}) \to Z_i$.
\begin{align}
&I(U_{1i},U_{2i};Z_i|X_{1i},X_{2i}) \notag\\
&\phantom{w} = I(M_{[1:2]},S_{[1:2]},X_{1,i+1}^n,X_{2,i+1}^n,Z^{i-1};Z_i|X_{1i},X_{2i}) \notag\\
&\phantom{w} \leq I(M_{[1:2]},S_{[1:2]},X_{1 \sim i},X_{2 \sim i},Z^{i-1};Z_i|X_{1i},X_{2i}) \notag\\
&\phantom{w} = I(S_{[1:2]},X_{1 \sim i},X_{2 \sim i},Z^{i-1};Z_i|X_{1i},X_{2i}) \notag\\
&\phantom{www}+I(M_{[1:2]};Z_i|X_1^n,X_2^n,S_{[1:2]},Z^{i-1}) \notag\\
&\phantom{w}\stackrel{(a)}= 0+0 = 0,
\end{align}
where in (a), the first term is zero since $S_{[1:2]}$ is independent of $(X_1^n,X_2^n,Z^n)$ and $(X_{1i},X_{2i},Z_i),i=1,\dots,n,$ are jointly i.i.d., while the second term is zero because of the Markov chain $M_{[1:2]} \to (X_1^n,X_2^n,S_{[1:2]}) \to Z^n$.
We now show that $U_{1i} \to X_{1i} \to (X_{2i},Z_i)$ is a Markov chain.
\begin{align}
&I(U_{1i};X_{2i},Z_i|X_{1i}) \notag\\
&\phantom{w} = I(M_1,S_1,X_{1,i+1}^n,Z^{i-1};X_{2i},Z_i|X_{1i}) \notag\\
&\phantom{w} \leq I(M_1,S_1,X_{1 \sim i},Z^{i-1};X_{2i},Z_i|X_{1i}) \notag\\
&\phantom{w} = I(S_1,X_{1 \sim i},Z^{i-1};X_{2i},Z_i|X_{1i}) \notag\\
&\phantom{www}+I(M_1;X_{2i},Z_i|X_1^n,S_1,Z^{i-1}) \notag\\
&\phantom{w}\stackrel{(a)}= 0+0 = 0,
\end{align}
where in (a), the first term is zero since $S_1$ is independent of $(X_1^n,X_2^n,Z^n)$ and $(X_{1i},X_{2i},Z_i),i=1,\dots,n,$ are jointly i.i.d., while the second term is zero because of the Markov chain $M_1 \to (X_1^n,S_1) \to (X_2^n,Z^n)$. Next, we show that $Y_i \to (U_{1i},U_{2i},Z_i) \to (X_{1i},X_{2i})$ is a Markov chain.
\begin{align*}
&I(Y_i;X_{1i},X_{2i}|U_{1i},U_{2i},Z_i) \notag\\
&\phantom{w} = I(Y_i;X_{1i},X_{2i}|M_{[1:2]},S_{[1:2]},X_{1,i+1}^n,X_{2,i+1}^n,Z^{i-1},Z_i) \notag\\
&\phantom{w} \leq I(Y_i,Z_{i+1}^n;X_{1i},X_{2i}|M_{[1:2]},S_{[1:2]},X_{1,i+1}^n,X_{2,i+1}^n,Z^i) \notag\\
&\phantom{w} = I(Z_{i+1}^n;X_{1i},X_{2i}|M_{[1:2]},S_{[1:2]},X_{1,i+1}^n,X_{2,i+1}^n,Z^i) \notag\\
&\phantom{www} +I(Y_i;X_{1i},X_{2i}|M_{[1:2]},S_{[1:2]},X_{1,i+1}^n,X_{2,i+1}^n,Z^n) \notag\\
&\phantom{w} \stackrel{(a)}= I(Z_{i+1}^n;X_{1i},X_{2i}|M_{[1:2]},S_{[1:2]},X_{1,i+1}^n,X_{2,i+1}^n,Z^i)\\
&\phantom{w} \leq I(Z_{i+1}^n;X_1^i,X_2^i|M_{[1:2]},S_{[1:2]},X_{1,i+1}^n,X_{2,i+1}^n,Z^i) \notag\\
&\phantom{w} \leq I(Z_{i+1}^n;X_1^i,X_2^i,M_{[1:2]}|S_{[1:2]},X_{1,i+1}^n,X_{2,i+1}^n,Z^i) \notag\\
&\phantom{w} = I(Z_{i+1}^n;X_1^i,X_2^i|S_{[1:2]},X_{1,i+1}^n,X_{2,i+1}^n,Z^i) \notag\\
&\phantom{www} +I(Z_{i+1}^n;M_{[1:2]}|S_{[1:2]},X_1^n,X_2^n,Z^i) \notag\\
&\phantom{w} \stackrel{(a)} = 0+0 =0,
\end{align*}
where (a) follows from the Markov chain $Y^n \to (M_{[1:2]},S_{[1:2]},Z^n) \to (X_1^n,X_2^n)$ and in (b), the first term is zero since $S_{[1:2]}$ is independent of $(X_1^n,X_2^n,Z^n)$ and $(X_{1i},X_{2i},Z_i),i=1,\dots,n,$ are jointly i.i.d., while the second term is zero because of the Markov chain $M_{[1:2]} \to (X_1^n,X_2^n,S_{[1:2]}) \to Z^n$.

It remains to prove the Markov chain $U_{2i} \to X_{2i} \to (U_{1i},X_{1i},Z_i)$. Consider the following chain of inequalities.
\begin{align}
&I(U_{2i};U_{1i},X_{1i},Z_i|X_{2i}) \notag\\
&\phantom{w}= I(M_2,S_2,X_{2,i+1}^n;M_1,S_1,X_{1,i}^n,Z^{i}|X_{2i}) \notag\\
&\phantom{w}= I(S_2,X_{2,i+1}^n;M_1,S_1,X_{1,i}^n,Z^{i}|X_{2i}) \notag\\
&\phantom{www}+I(M_2;M_1,S_1,X_{1,i}^n,Z^{i}|S_2,X_{2,i}^n) \notag\\
&\phantom{w}= I(S_2,X_{2,i+1}^n;S_1,X_{1,i}^n,Z^{i}|X_{2i}) \notag\\
&\phantom{www}+I(S_2,X_{2,i+1}^n;M_1|S_1,X_{1,i}^n,Z^i,X_{2i}) \notag\\
&\phantom{www}+I(M_2;S_1,X_{1,i}^n,Z^{i}|S_2,X_{2,i}^n) \notag\\
&\phantom{www}+I(M_2;M_1|S_1,S_2,X_{1,i}^n,X_{2,i}^n,Z^i) \notag\\
&\phantom{w}\stackrel{(a)} \leq 0+I(S_2,X_{2,i+1}^n;M_1,X_{1}^{i-1},Z^{i-1}|S_1,X_{1,i}^n,Z^i,X_{2i}) \notag\\
&\phantom{www}+I(M_2,X_2^{i-1};S_1,X_{1,i}^n,Z^{i}|S_2,X_{2,i}^n) \notag\\
&\phantom{www}+I(M_2,X_2^{i-1};M_1|S_1,S_2,X_{1,i}^n,X_{2,i}^n,Z^i) \notag\\
&\phantom{w}= I(S_2,X_{2,i+1}^n;X_{1}^{i-1},Z^{i-1}|S_1,X_{1,i}^n,Z^i,X_{2i}) \notag\\
&\phantom{www}+I(S_2,X_{2,i+1}^n;M_1|S_1,X_1^n,Z^n,X_{2i}) \notag\\
&\phantom{www}+I(X_2^{i-1};S_1,X_{1,i}^n,Z^{i}|S_2,X_{2,i}^n) \notag\\
&\phantom{www}+I(M_2;S_1,X_{1,i}^n,Z^{i}|S_2,X_2^n) \notag\\
&\phantom{www}+I(X_2^{i-1};M_1|S_1,S_2,X_{1,i}^n,X_{2,i}^n,Z^i) \notag\\
&\phantom{www}+I(M_2;M_1|S_1,S_2,X_{1,i}^n,X_2^n,Z^i) \notag\\
&\phantom{w} \stackrel{(b)}= 0\!+0\!+0\!+0\!+I(X_2^{i-1};M_1|S_1,S_2,X_{1,i}^n,X_{2,i}^n,Z^i)\!+0 \notag\\
&\phantom{w} \leq I(X_2^{i-1};M_1,X_1^{i-1}|S_1,S_2,X_{1,i}^n,X_{2,i}^n,Z^i) \notag\\
&\phantom{w} = I(X_2^{i-1};X_1^{i-1}|S_1,S_2,X_{1,i}^n,X_{2,i}^n,Z^i) \notag\\
&\phantom{www}+I(X_2^{i-1};M_1|S_1,S_2,X_1^n,X_{2,i}^n,Z^i) \notag\\
&\phantom{w} \stackrel{(c)}= I(X_2^{i-1};X_1^{i-1}|Z^{i-1}) \notag\\
&\phantom{www}+I(X_2^{i-1};M_1|S_1,S_2,X_1^n,X_{2,i}^n,Z^i) \notag\\
&\phantom{w} \stackrel{(d)}= 0+0 =0,
\end{align}
where the fact that the first term in (a) and the first and third terms in (b) are zeros, as well as (c), follow since $S_{[1:2]}$ is independent of $(X_1^n,X_2^n,Z^n)$ and $(X_{1i},X_{2i},Z_i),i=1,\dots,n,$ are jointly i.i.d. The second, fourth and sixth terms in (b) and the second term in (d) are zeros because of the Markov chains $M_1 \to (X_1^n,S_1) \to (S_2,X_2^n,Z^n)$ and $M_2 \to (X_2^n,S_2) \to (S_1,M_1,X_1^n,Z^n)$. The first term in (d) is zero because $(X_{1i},X_{2i},Z_i),i=1,\dots,n,$ are jointly i.i.d. with $q_{X_1X_2,Z}$ and $I(X_1;X_2|Z)=0$. Notice that we also have $\lVert p_{X_1,X_2,Z,Y|T=t}-q_{X_1,X_2,Z,Y|T=t} \rVert_1\leq \epsilon$ along the same lines as in the proof of Theorem~\ref{thm:encsideOBsr}. 
\end{proof}
\end{proof}

\subsubsection{Proof of Theorem~\ref{Theorem:outbndindp}}\label{proof:thm6}
\begin{proof}
{\blue We prove the following lemma later. 
\begin{lemma}\label{lem:outbndindp}
Consider a p.m.f. $q_{X_1X_2ZY}$ such that the random variables $X_1$ and $X_2$ are conditionally independent given $Z$. Then any rate tuple $(R_1,R_2,R_{00},R_{01},R_{02})$ in $\mathcal{R}_{\textup{MAC-coord}}$ satisfies, {\blue for every $\epsilon \in (0,\frac{1}{4}]$},
\begin{align*}
R_1&\geq I(U_0,U_1;X_1|Z,T)\nonumber\\
R_2&\geq I(U_0,U_2;X_2|Z,T)\nonumber\\
R_1+R_2&\geq I(U_0,U_1,U_2;X_1,X_2|Z,T)\nonumber\\
R_1+R_{01}&\geq I(U_1;X_1,Y|X_2,Z,T){\blue-g(\epsilon)}\nonumber\\
R_2+R_{02}&\geq I(U_2;X_2,Y|X_1,Z,T){\blue-g(\epsilon)}\nonumber\\
R_{00}+R_1+R_{01}&\geq I(U_0,U_1;X_1,Y|X_2,Z,T){\blue-g(\epsilon)}\nonumber\\
R_{00}+R_2+R_{02}&\geq I(U_0,U_2;X_2,Y|X_1,Z,T){\blue-g(\epsilon)}\nonumber\\
R_1+R_2+R_{01}+R_{02} &\geq I(U_1,U_2;X_1,X_2,Y|Z,T) {\blue-g(\epsilon)}\nonumber\\
R_{00}+R_1+R_2+&R_{01}+R_{02}\nonumber\\
&\hspace{-0.25cm}\geq I(U_0,U_1,U_2;X_1,X_2,Y|Z,T){\blue-g(\epsilon)}\nonumber,
\end{align*}
with $g(\epsilon)=2\sqrt{\epsilon}\left(H_q(X_1,X_2,Y,Z)+\log\frac{(|\mathcal{X}_1||\mathcal{X}_2||\mathcal{Y}||\mathcal{Z}|)}{\epsilon}\right)$ (which tends to 0 as $\epsilon \to 0$), for some p.m.f.
\begin{align*}
p(&x_1,x_2,z,t,u_0,u_1,u_2,y)=p(z)p(x_1|z)p(x_2|z)p(t)\\ 
&\hspace{12pt}\times p(u_0|t)p(u_1|x_1,u_0,t)p(u_2|x_2,u_0,t)p(y|u_1,u_2,z,t)
\end{align*}
such that {\blue$||p(x_1,x_2,y,z|t)-q(x_1,x_2,y,z)||_1 \leq \epsilon \: \textup{for all} \: t.$} 
\end{lemma}
}

{\blue Notice that Lemma~\ref{lem:outbndindp} gives an epsilon rate region, whose continuity at $\epsilon=0$ is unknown. When the shared randomness rates $(R_{00},R_{01},R_{02})$ are sufficiently large, we can in fact prove that such a continuity argument holds. We argue this below. 

First we prove the following lemma concerning the cardinality bounds on the auxiliary random variables $U_0$, $U_1$ and $U_2$ under sufficiently large shared randomness rates $(R_{00},R_{01},R_{02})$. Let $S'_{\epsilon}$ denote the set of rate pairs $(R_1,R_2)$ that satisfy, for every $\epsilon>0$, 
\begin{align*}
R_1&\geq I(U_0,U_1;X_1|Z,T)\\
R_2&\geq I(U_0,U_2;X_2|Z,T)\\
R_1+R_2&\geq I(U_0,U_1,U_2;X_1,X_2|Z,T)\\
\end{align*}
for some p.m.f. 
\begin{align*}
p(&x_1,x_2,z,t,u_0,u_1,u_2,y)=p(z)p(x_1|z)p(x_2|z)p(t)\\ 
&\hspace{12pt}\times p(u_0|t)p(u_1|x_1,u_0,t)p(u_2|x_2,u_0,t)p(y|u_1,u_2,z,t)
\end{align*}
such that $||p(x_1,x_2,y,z|t)-q(x_1,x_2,y,z)||_1 \leq \epsilon \: \textup{for all} \: t$.
\begin{lemma} \label{lem:cardbndU0}
The size of the auxiliary random variable alphabets in $S'_{\epsilon}$ can be restricted to:
\begin{align*}
|\mathcal{U}_0| &\leq |\mathcal{X}_1||\mathcal{X}_2||\mathcal{Y}||\mathcal{Z}|,\\
|\mathcal{U}_1| &\leq |\mathcal{U}_0||\mathcal{X}_1||\mathcal{X}_2||\mathcal{Y}||\mathcal{Z}|,\\
|\mathcal{U}_2| &\leq |\mathcal{U}_0||\mathcal{X}_1||\mathcal{X}_2||\mathcal{Y}||\mathcal{Z}|,\\
|\mathcal{T}| &\leq 3.
\end{align*}
\end{lemma}
\begin{proof}
See Appendix~\ref{app:cardbndU0}.
\end{proof}

From Lemmas \ref{lem:outbndindp} and \ref{lem:cardbndU0}, the outer bound in Theorem~\ref{Theorem:outbndindp} follows using the continuity of total variation distance and mutual information in the probability simplex along the same lines as \cite[Lemma VI.5]{cuff2013distributed},\cite[Lemma 6]{YassaeeGA15}. It remains to prove Lemma~\ref{lem:outbndindp}.
}
\begin{proof}[Proof of Lemma~\ref{lem:outbndindp}]
With the same choice of auxiliary random variables as in the proof of Theorem~\ref{thm:encsideOBsr}, {\blue i.e., $U_{0i} = S_0$, $U_{1i} = (M_1,S_1,X_{1,i+1}^n,Z^{i-1})$, and $U_{2i} = (M_2,S_2,X_{2,i+1}^n)$}, we will show that the Markov chain $U_2 \to (X_2,U_0,T) \to (U_1,X_1,Z)$ holds when $X_1 \to Z \to X_2$. Then in addition to the Markov chains $U_0 \to T \to (X_1,X_2,Z)$, $U_1 \to (X_1,U_0,T) \to (X_2,Z)$, $(U_1,U_2) \to (X_1,X_2,U_0,T) \to Z$ and $Y \to (U_1,U_2,Z,T) \to (X_1,X_2,U_0)$, we note the following simplifications to the rate constraints in Theorem \ref{thm:encsideOBsr}:
\begin{align*}
R_1 &\geq \max\{I(U_1;X_1|U_0,Z,T),I(U_1;X_1|U_0,U_2,X_2,Z,T)\} \notag\\
&= \max\{I(U_1;X_1|U_0,Z,T),I(U_1;X_1|U_0,Z,T)\} \notag\\
&= I(U_0,U_1;X_1|Z,T), \\
R_2 &\geq \max\{I(U_2;X_2|U_0,Z,T),I(U_2;X_2|U_0,U_1,X_1,Z,T)\} \notag\\
&= \max\{I(U_2;X_2|U_0,Z,T),I(U_2;X_2|U_0,Z,T)\} \notag\\
& = I(U_0,U_2;X_2|Z,T),
\end{align*}
\begin{align*}
R_1+R_2 &\geq I(U_1,U_2;X_1,X_2|U_0,Z,T) \notag\\
&= I(U_0,U_1,U_2;X_1,X_2|Z,T),
\end{align*}
\begin{align*}
R_1+R_{01} &\geq I(U_1;X_1,X_2,Y|Z,T) \notag\\
&= I(U_1;X_1,Y|X_2,Z,T),\\
R_2+R_{02} &\geq I(U_2;X_1,X_2,Y|Z,T) \notag\\
&= I(U_2;X_2,Y|X_1,Z,T),
\end{align*}
\begin{align*}
R_{00}+R_1+R_{01} &\geq I(U_0,U_1;X_1,X_2,Y|Z,T) \notag\\
&= I(U_0,U_1;X_1,Y|X_2,Z,T)\\
R_{00}+R_2+R_{02} &\geq I(U_0,U_2;X_1,X_2,Y|Z,T) \notag\\
&= I(U_0,U_2;X_2,Y|X_1,Z,T),
\end{align*}
\begin{align*}
R_1+R_2+R_{01}+R_{02} &\geq I(U_1,U_2;X_1,X_2,Y|Z,T),\\
R_{00}+R_1+R_2+R_{01}+&R_{02}\\
 &\geq I(U_0,U_1,U_2;X_1,X_2,Y|Z,T)\!.
\end{align*}
It remains to prove the Markov chain $U_{2i} \to (X_{2i},U_{0i}) \to (U_{1i},X_{1i},Z_i)$. Recall the choice of auxiliary random variables.
\begin{align*}
U_{0i} &= S_0, \\
U_{1i} &= (M_1,S_1,X_{1,i+1}^n,Z^{i-1}), \\
U_{2i} &= (M_2,S_2,X_{2,i+1}^n).
\end{align*}
Consider the following chain of inequalities.
\begin{align}
&I(U_{2i};U_{1i},X_{1i},Z_i|X_{2i},U_{0i}) \notag\\
&= I(M_2,S_2,X_{2,i+1}^n;M_1,S_1,X_{1,i}^n,Z^{i}|X_{2i},S_0) \notag\\
&= I(S_2,X_{2,i+1}^n;M_1,S_1,X_{1,i}^n,Z^{i}|X_{2i},S_0) \notag\\
&\phantom{www}+I(M_2;M_1,S_1,X_{1,i}^n,Z^{i}|S_0,S_2,X_{2,i}^n) \notag\\
&= I(S_2,X_{2,i+1}^n;S_1,X_{1,i}^n,Z^{i}|X_{2i},S_0) \notag\\
&\phantom{www}+I(S_2,X_{2,i+1}^n;M_1|S_1,X_{1,i}^n,Z^i,X_{2i},S_0) \notag\\
&\phantom{www}+I(M_2;S_1,X_{1,i}^n,Z^{i}|S_0,S_2,X_{2,i}^n) \notag\\
&\phantom{www}+I(M_2;M_1|S_0,S_1,S_2,X_{1,i}^n,X_{2,i}^n,Z^i) \notag\\
&\stackrel{(a)} \leq 0+I(S_2,X_{2,i+1}^n;M_1,X_{1}^{i-1},Z^{i-1}|S_0,S_1,X_{1,i}^n,Z^i,X_{2i}) \notag\\
&\phantom{www}+I(M_2,X_2^{i-1};S_1,X_{1,i}^n,Z^{i}|S_0,S_2,X_{2,i}^n) \notag\\
&\phantom{www}+I(M_2,X_2^{i-1};M_1|S_0,S_1,S_2,X_{1,i}^n,X_{2,i}^n,Z^i) \notag\\
&= I(S_2,X_{2,i+1}^n;X_{1}^{i-1},Z^{i-1}|S_0,S_1,X_{1,i}^n,Z^i,X_{2i}) \notag\\
&\phantom{www}+I(S_2,X_{2,i+1}^n;M_1|S_0,S_1,X_1^n,Z^n,X_{2i}) \notag\\
&\phantom{www}+I(X_2^{i-1};S_1,X_{1,i}^n,Z^{i}|S_0,S_2,X_{2,i}^n) \notag\\
&\phantom{www}+I(M_2;S_1,X_{1,i}^n,Z^{i}|S_0,S_2,X_2^n) \notag\\
&\phantom{www}+I(X_2^{i-1};M_1|S_0,S_1,S_2,X_{1,i}^n,X_{2,i}^n,Z^i) \notag\\
&\phantom{www}+I(M_2;M_1|S_0,S_1,S_2,X_{1,i}^n,X_2^n,Z^i) \notag\\
& \stackrel{(b)}= 0\!+0\!+0\!+0\!\notag\\
&\phantom{www}+\!I(X_2^{i-1};M_1|S_0,S_1,S_2,X_{1,i}^n,X_{2,i}^n,Z^i)\!+\!0 \notag\\
& \leq I(X_2^{i-1};M_1,X_1^{i-1}|S_0,S_1,S_2,X_{1,i}^n,X_{2,i}^n,Z^i) \notag\\
& = I(X_2^{i-1};X_1^{i-1}|S_0,S_1,S_2,X_{1,i}^n,X_{2,i}^n,Z^i) \notag\\
&\phantom{www}+I(X_2^{i-1};M_1|S_0,S_1,S_2,X_1^n,X_{2,i}^n,Z^i) \notag\\
& \stackrel{(c)}= I(X_2^{i-1};X_1^{i-1}|Z^{i-1}) \notag\\
&\phantom{www}+I(X_2^{i-1};M_1|S_0,S_1,S_2,X_1^n,X_{2,i}^n,Z^i) \notag\\
& \stackrel{(d)}= 0+0 =0,
\end{align}
where the first term in (a), the first and third terms in (b) are zeros and (c) follows since $S_{[0:2]}$ is independent of $(X_1^n,X_2^n,Z^n)$ and $(X_{1i},X_{2i},Z_i),i=1,\dots,n,$ are jointly i.i.d. The second, fourth and sixth terms in (b) and the second term in (d) are zeros because of the Markov chains $M_1 \to (X_1^n,S_0,S_1) \to (S_2,X_2^n,Z^n)$ and $M_2 \to (X_2^n,S_0,S_2) \to (S_1,M_1,X_1^n,Z^n)$. The first term in (d) is zero because $(X_{1i},X_{2i},Z_i),i=1,\dots,n,$ are jointly i.i.d. with $q_{X_1X_2,Z}$ and $I(X_1;X_2|Z)=0$.
\end{proof}
\end{proof}

{\blue \section{Acknowledgements}
We thank the Associate Editor and the anonymous reviewers for their insightful comments on the paper. In particular, we would like to thank them for making us aware of \cite[Section IV]{HaddadpourYAG13}. In addition, we thank Sandeep Pradhan for alerting us of an error in an earlier version of the paper.}

\appendices
\section{Proof of Theorem~\ref{thm:encsideIB}} \label{app:pfThm1}
The proof employs the Output Statistics of Random Binning (OSRB) framework developed by Yassaee et al.~\cite{yassaee2014achievability}. In the sequel, we use capital letters (like $P_{X}$) to denote random p.m.f.'s (see, e.g., \cite{cuff2013distributed,yassaee2014achievability}) and lower-case letters (like $p_X$) to denote non-random p.m.f.'s. We use $p_{\mathcal{A}}^{\text{U}}$ to denote the uniform distribution over the set $\mathcal{A}$. The notation $\approx$ for pmf approximations is adopted from \cite{yassaee2014achievability} -- for two random pmfs $P_X$ and $Q_X$ on the same alphabet $\mathcal{X}$, we say that $P_X \stackrel{\epsilon}\approx Q_X$ provided $\eE{[||P_X-Q_X||_1]} \leq \epsilon$. {\blue For any two sequences of random p.m.f.'s $P_{X^{(n)}}$ and $Q_{X^{(n)}}$ on $\mathcal{X}^{(n)}$, we write $P_{X^{(n)}}\approx Q_{X^{(n)}}$ if $\lim_{n\rightarrow \infty}\mathbbm{E}\lVert P_{X^{(n)}}-Q_{X^{(n)}}\rVert_1=0$. Likewise, we use $p_X \stackrel{\epsilon}\approx q_X$ for two (non-random) p.m.f.'s provided $||p_X-q_X||_1 \leq \epsilon$. For any two sequences of random pmfs $P_{X^n}$ and $Q_{X^n}$ on $\mathcal{X}^n$, we write $P_{X^n} \approx Q_{X^n}$ if $\lim_{n \to \infty} \eE{[||P_{X^n}-Q_{X^n}||_1]} =0$. Similarly, we write $p_{X^n} \approx q_{X^n}$ for two sequences of (non-random) p.m.f.'s.} We also quote the following results that will prove useful in the proof of Theorem~\ref{thm:encsideIB}, where the first one is a restatement of the OSRB result~\cite[Theorem 1]{yassaee2014achievability}.
\begin{theorem} \cite[Theorem 1]{yassaee2014achievability} \label{lem:a1}
Given a discrete memoryless source $(A^n,B_1^n,\cdots,B_L^n) \sim \text{i.i.d.}$ with $p_{A,B_1,\cdots,B_L}$ on $\mathcal{A}\times\prod_{j=1}^L\mathcal{B}_i$. For $j\in[1:L]$, let $\phi_j:\mathcal{B}_j^n\rightarrow [1:2^{nR_j}]$ be a random binning in which $\phi_j$ maps each sequence of $\mathcal{B}_j^n$ uniformly and independently to the set $[1:2^{nR_j}]$.  Let $K_j=\phi_j(B_j^n), j \in \{1,\cdots,L\}$. If for each $\mathcal{S} \subseteq [1:L]$, the following constraint
\begin{align}
\sum_{t \in \mathcal{S}} R_t \leq H(B_{\mathcal{S}}|A)
\end{align}
holds, then we have
\begin{align}
\lim_{n \to \infty} \eE_{\phi_1,\cdots,\phi_L}\left[ ||P_{A^n, K_1,\cdots,K_L}-p_{A^n} \prod_{j=1}^{L} p_{K_j}^{\emph{U}}||_1 \right] = 0,
\end{align}
where $\eE_{\phi_1,\cdots,\phi_L}$ denotes the expectation over the random binnings, $P_{A^n, K_1,\cdots,K_L}$ is a random p.m.f., and $p_{K_j}^{\emph{U}}$ is the uniform distribution over $[1:2^{nR_j}]$.
\end{theorem}

\begin{lemma} \cite[Lemma 4]{yassaee2014achievability} \label{lem:a2}
\begin{enumerate}
\item If $P_{X^n} \approx Q_{X^n}$, then $P_{X^n} P_{Y^n|X^n} \approx Q_{X^n} P_{Y^n|X^n}$. Also if $P_{X^n} P_{Y^n|X^n} \approx Q_{X^n} Q_{Y^n|X^n}$, then $P_{X^n} \approx Q_{X^n}$.
\item If $p_{X^n} p_{Y^n|X^n} \approx q_{X^n} q_{Y^n|X^n}$, then there exists a sequence $x^n \in \mathcal{X}^n$ such that $p_{Y^n|X^n=x^n} \approx q_{Y^n|X^n=x^n}$.
\item If $P_{X^n} \approx Q_{X^n}$ and $P_{X^n} P_{Y^n|X^n} \approx P_{X^n} Q_{Y^n|X^n}$, then $P_{X^n} P_{Y^n|X^n} \approx Q_{X^n} Q_{Y^n|X^n}$.
\end{enumerate}
\end{lemma}
{\blue For simplicity, we prove the achievability for $|\mathcal{T}|=1$. The proof with the general time-sharing random variable $T$ then follows using standard time sharing argument outlined below.

Suppose the rate tuples $\boldsymbol{R}':=(R_1',R_2',R_{01}',R_{02}')$ and $\boldsymbol{R}'':=(R_1'',R_2'',R_{01}'',R_{02}'')$ are achievable for $q_{X_1X_2ZY}$, i.e., there exist p.m.f.'s $p^\prime(x_1,x_2,z,u_1,u_2,y)$ and $p^{\prime\prime}(x_1,x_2,z,u_1,u_2,y)$ which together with the rate tuples $\boldsymbol{R}'$ and $\boldsymbol{R}''$, respectively, satisfy the constraints in Theorem~1 (with $T=\emptyset$). For each blocklength $n$, we use the code corresponding to $\boldsymbol{R}'$ for the first $\alpha n$ (with $\alpha\in[0,1]$, where $\alpha n$ is an integer) transmissions and use the code corresponding to $\boldsymbol{R}''$ for the rest of the transmissions. Let $p_{X_1^n,X_2^n,Z^n,Y^n}$ be the induced distribution with this new code. For ease of notation, let $D\triangleq (X_1,X_2,Z,Y)$. Note that the induced distribution decomposes as $p_{D^n}=p_{D^{\alpha n}}\cdot p_{D_{\alpha n+1}^n}$, where $D_{\alpha n+1}^n=(D_{\alpha n+1},\dots,D_{n})$, because of the independence of the two segments. Then, the overall total variation distance of interest with this new code is given by
\begin{align}
    &\lVert p_{D^n}-\prod q_{D} \rVert\nonumber\\
    &=\lVert p_{D^{\alpha n}}\cdot p_{D_{\alpha n+1}^{n}}-q^{(n)}_{D} \rVert\nonumber\\
    &=\sum_{d^n} \bigg|p_{D^{\alpha n}}(d^{\alpha n})p_{D_{\alpha n+1}^n}(d_{\alpha n+1}^n)\nonumber\\
    &\hspace{1cm}-q^{(\alpha n)}_D(d^{\alpha n})q^{(\alpha n+1:n)}_D(d_{\alpha n+1}^n)\bigg|\label{eqn:timesharing1}\\
    &\leq\sum_{d^n}\bigg|p_{D^{\alpha n}}(d^{\alpha n})p_{D_{\alpha n+1}^n}(d_{\alpha n+1}^n)\nonumber\\
    &\hspace{1cm}-p_{D^{\alpha n}}(d^{\alpha n})q^{(\alpha n+1:n)}_D(d_{\alpha n+1}^n)\bigg|\nonumber\\
    &\hspace{12pt}+\sum_{d^n}\bigg|p_{D^{\alpha n}}(d^{\alpha n})q^{(\alpha n+1:n)}_D(d_{\alpha n+1}^n)\nonumber\\
    &\hspace{1cm}-q^{(\alpha n)}_D(d^{\alpha n})q^{(\alpha n+1:n)}_D(d_{\alpha n+1}^n)\bigg|\label{eqn:timesharing2}\\
    &=\sum_{d_{\alpha n+1}^n}|p_{D_{\alpha n+1}^n}(d_{\alpha n+1}^n)-q^{(\alpha n+1:n)}_D(d_{\alpha n+1}^n)|\nonumber\\
    &\hspace{1cm}+\sum_{d^{n\alpha }}|p_{D^{\alpha n}}(d^{\alpha n})-q^{(\alpha n)}_D(d^{\alpha n})|\nonumber\\
    &\rightarrow 0 \ \text{as} \ n\rightarrow \infty\label{eqn:timesharing3},
\end{align}
where \eqref{eqn:timesharing1} follows by defining $q^{(\alpha n+1:n)}_D(d_{\alpha n+1}^n)=\prod_{i=\alpha n+1}^nq_D(d_i)$, \eqref{eqn:timesharing2} follows by adding and subtracting the term $p_{D^{\alpha n}}(d^{\alpha n})q^{(\alpha n+1:n)}_D(d_{\alpha n+1}^n)$ for each $d^n$ inside the absolute value in \eqref{eqn:timesharing1} and then applying the triangle inequality, and \eqref{eqn:timesharing3} follows because the codes in each of the segments satisfy \eqref{eqn:correctness} with the respective blocklengths. So, the rate tuple $\alpha \boldsymbol{R}'+(1-\alpha)\boldsymbol{R}''$ is achievable. This handles the case of $|\mathcal{T}|=2$ with $\mathcal{T}=\{1,2\}$ and $p_T(1)=1-p_T(2)=\alpha$. The generalization to larger alphabets for $T$ follows along similar lines.
} 

We now follow the standard structure of an achievability proof via OSRB. This involves defining two protocols, one each based on random coding and random binning, that induce a joint distribution on the random variables defined during the protocols. \\
\textbf{Random Binning Scheme:} Let $(U_1^n,U_2^n,X_1^n,X_2^n,Z^n,Y^n)$ be drawn i.i.d. with the joint distribution $$p(x_1,x_2,z)p(u_1|x_1)p(u_2|x_2)p(y|u_1,u_2,z)$$ such that $p(x_1,x_2,z,y)=q(x_1,x_2,z,y)$. 
Now we employ the following random binning:
\begin{itemize}
\item Generate $(S_1,M_1,F_1)$ as three uniform binnings of $U_1^n$ independently, i.e. $S_1 = \phi_{11}(U_1^n) \in [1:2^{n R_{01}}]$, $M_1 = \phi_{12}(U_1^n) \in [1:2^{n R_1}]$ and $F_1 = \phi_{13}(U_1^n) \in [1:2^{n \tilde{R}_1}]$. {\blue Here, $S_1$ stands for the pairwise shared randomness between the first encoder and the decoder, $M_1$ stands for a message communicated over the noiseless link, while $F_1$ is additional shared randomness assumed in OSRB to be eliminated later without disturbing the i.i.d. distribution.}
\item Likewise, generate $(S_2,M_2,F_2)$ as three uniform binnings of $U_2^n$ independently, i.e. $S_2 = \phi_{21}(U_2^n) \in [1:2^{n R_{02}}]$, $M_2 = \phi_{22}(U_2^n) \in [1:2^{n R_2}]$ and $F_2 = \phi_{23}(U_2^n) \in [1:2^{n \tilde{R}_2}]$. 
\end{itemize}
The receiver uses a Slepian-Wolf decoder to estimate $(\hat{u}_1^n,\hat{u}_2^n)$ from $(s_1,s_2,f_1,f_2,m_1,m_2,z^n)$.
The corresponding random p.m.f. induced is (the randomness is due to the binning)
\begin{align}
&P(x_1^n,x_2^n,z^n,y^n,u_1^n,u_2^n,s_1,f_1,m_1,s_2,f_2,m_2,\hat{u}_1^n,\hat{u}_2^n) \notag\\
&= p(x_1^n,x_2^n,z^n)p(u_1^n|x_1^n)p(u_2^n|x_2^n)p(y^n|u_1^n,u_2^n,z^n) \notag\\
&\phantom{w} \times P(s_1,f_1,m_1|u_1^n) P(s_2,f_2,m_2|u_2^n) \notag\\
&\phantom{w} \times P^{SW}(\hat{u}_1^n,\hat{u}_2^n|s_1,f_1,m_1,s_2,f_2,m_2,z^n) \label{eq:osrb2} \\
&= p(x_1^n,x_2^n,z^n) P(s_1,f_1,m_1,u_1^n|x_1^n) P(s_2,f_2,m_2,u_2^n|x_2^n) \notag\\
&\phantom{w} \times P^{SW}(\hat{u}_1^n,\hat{u}_2^n|s_1,f_1,m_1,s_2,f_2,m_2,z^n) p(y^n|u_1^n,u_2^n,z^n) \notag\\
&= p(x_1^n,x_2^n,z^n) P(s_1,f_1|x_1^n) P(u_1^n|s_1,f_1,x_1^n) P(m_1|u_1^n) \notag\\
&\phantom{w} \times P(s_2,f_2|x_2^n) P(u_2^n|s_2,f_2,x_2^n) P(m_2|u_2^n) \notag\\
&\phantom{w} \times P^{SW}(\hat{u}_1^n,\hat{u}_2^n|s_1,f_1,m_1,s_2,f_2,m_2,z^n) p(y^n|u_1^n,u_2^n,z^n) \notag\\
&= P(x_1^n,x_2^n,z^n,s_1,s_2,f_1,f_2) P(u_1^n|s_1,f_1,x_1^n) P(m_1|u_1^n) \notag\\
&\phantom{w} \times P(u_2^n|s_2,f_2,x_2^n) P(m_2|u_2^n) \notag\\
&\phantom{w} \times P^{SW}(\hat{u}_1^n,\hat{u}_2^n|s_1,f_1,m_1,s_2,f_2,m_2,z^n) p(y^n|u_1^n,u_2^n,z^n). \label{eq:osrb3}
\end{align}
where \eqref{eq:osrb2} uses the Markov chains $(U_1,U_2) \to (X_1,X_2) \to Z$, $U_1 \to X_1 \to X_2 \to U_2$, $Y \to (U_1,U_2,Z) \to (X_1,X_2)$, and the binning construction. 

\noindent \textbf{Random Coding Scheme:} We assume that additional shared randomness $F_j$ of rate $\tilde{R}_j, j \in \{1,2\}$ are available between the respective encoders and the decoder in the main problem. 
Encoder $j \in \{1,2\}$, knowing $(s_j,f_j,x_j^n)$, generates $u_j^n$ according to the p.m.f. $P(u_j^n|s_j,f_j,x_j^n)$ (from the previous protocol) and sends 
the bin index of $u_j^n$ corresponding to the binning $\phi_{j2}$ in the previous protocol over the noiseless link to the decoder. 
The decoder obtains $(s_1,s_2,f_1,f_2,m_1,m_2,z^n)$, and employs the Slepian-Wolf decoder from the previous protocol, i.e. $P^{SW}(\hat{u}_1^n,\hat{u}_2^n|s_1,f_1,m_1,s_2,f_2,m_2,z^n)$, to estimate $(u_1^n,u_2^n)$. 
Then it constructs $y^n$ according to the distribution $p_{Y^n|U_1^n,U_2^n,Z^n}(y^n|\hat{u}_1^n,\hat{u}_2^n,z^n)$. 
The induced random p.m.f. from this protocol is given by
\begin{align}
&\hat{P}(x_1^n,x_2^n,z^n,y^n,u_1^n,u_2^n,s_1,f_1,m_1,s_2,f_2,m_2,\hat{u}_1^n,\hat{u}_2^n) \notag\\
&= p^{\text{U}}(s_1)p^{\text{U}}(f_1)p^{\text{U}}(s_2)p^{\text{U}}(f_2)p(x_1^n,x_2^n,z^n)  \notag\\
&\phantom{ww} \times P(u_1^n|s_1,f_1,x_1^n) P(m_1|u_1^n) \notag\\
&\phantom{ww} \times P(u_2^n|s_2,f_2,x_2^n) P(m_2|u_2^n) \notag\\
&\phantom{ww} \times P^{SW}(\hat{u}_1^n,\hat{u}_2^n|s_1,f_1,m_1,s_2,f_2,m_2,z^n) p(y^n|\hat{u}_1^n,\hat{u}_2^n,z^n). \label{eq:rcpmf}
\end{align}

Next we find constraints that imply that the induced p.m.f.'s from the two protocols are almost identical. Then one can restrict attention to the source coding side of the problem (related to the random binning protocol) and investigate the desired properties like vanishing total variation distance.

\noindent \textbf{Analysis of Rate Constraints:}\\
We now derive sufficient conditions for the joint statistics of the random variables from the two protocols to be identical. 
Since $(s_j,f_j)$ are the bin indices of $u_j^n$ for $j \in \{1,2\}$, if we ensure that
\begin{align}
R_{01}+\tilde{R}_1 &\leq H(U_1|X_1,X_2,Z) = H(U_1|X_1), \label{eq:cond11} \\
R_{02}+\tilde{R}_2 &\leq H(U_2|X_1,X_2,Z) = H(U_2|X_2), \label{eq:cond12} \\
R_{01}+\tilde{R}_1+R_{02}+\tilde{R}_2 &\leq H(U_1,U_2|X_1,X_2,Z)\nonumber\\
&=H(U_1,U_2|X_1,X_2) \label{eq:cond13}
\end{align}
{\blue(where the equalities in \eqref{eq:cond11}--\eqref{eq:cond13} follow from the Markov chains $U_1 \to X_1 \to (X_2,Z)$, $U_2 \to X_2 \to (X_1,Z)$, and $(U_1,U_2) \to (X_1,X_2) \to Z$, respectively)}, then by Theorem~\ref{lem:a1}, we obtain 
\begin{align}
&P(x_1^n,x_2^n,z^n,s_1,s_2,f_1,f_2) \notag\\
&\phantom{ww} \approx p^{\text{U}}\!(s_1)p^{\text{U}}\!(f_1)p^{\text{U}}\!(s_2)p^{\text{U}}\!(f_2)p(x_1^n,x_2^n,z^n) \notag\\
&\phantom{ww} = \hat{P}(x_1^n,x_2^n,z^n,s_1,s_2,f_1,f_2).
\end{align}
This in turn results in
\begin{align}
&P(x_1^n,x_2^n,z^n,u_1^n,u_2^n,s_1,f_1,m_1,s_2,f_2,m_2,\hat{u}_1^n,\hat{u}_2^n) \notag\\
&\phantom{ww} \approx \hat{P}(x_1^n,x_2^n,z^n,u_1^n,u_2^n,s_1,f_1,m_1,s_2,f_2,m_2,\hat{u}_1^n,\hat{u}_2^n). \label{eq:condit1}
\end{align}
Note that the condition \eqref{eq:cond13} above is redundant because $H(U_1,U_2|X_1,X_2)=H(U_1|X_1)+H(U_2|X_2)$ using the Markov chain $U_1 \to X_1 \to X_2 \to U_2$.

For the Slepian-Wolf decoder to succeed, we require (by Slepian-Wolf theorem~\cite{slepian1973noiseless}, see also \cite[Lemma 1]{yassaee2014achievability})
\begin{align}
R_{01}+\tilde{R}_1+R_1 &\geq H(U_1|U_2,Z), \label{eq:cond21} \\
R_{02}+\tilde{R}_2+R_2 &\geq H(U_2|U_1,Z), \label{eq:cond22} \\
R_{01}+\tilde{R}_1+R_1+R_{02}+\tilde{R}_2+R_2 &\geq H(U_1,U_2|Z). \label{eq:cond23}
\end{align}
This ensures that
\begin{align}
&P(x_1^n,x_2^n,z^n,u_1^n,u_2^n,s_1,f_1,m_1,s_2,f_2,m_2,\hat{u}_1^n,\hat{u}_2^n) \notag\\
&\phantom{w} \approx P(x_1^n,x_2^n,z^n,u_1^n,u_2^n,s_1,f_1,m_1,s_2,f_2,m_2) \notag\\
&\phantom{wwwww} \times \mathbbm{1}\{\hat{u}_1^n=u_1^n,\hat{u}_2^n=u_2^n\}. \label{eq:condit2}
\end{align}
Using \eqref{eq:condit2}, \eqref{eq:condit1} and the first and third parts of Lemma \ref{lem:a2}, we can write the following for the joint probability distribution involving $y^n$
\begin{align}
&\hat{P}(x_1^n,x_2^n,z^n,u_1^n,u_2^n,s_1,f_1,m_1,s_2,f_2,m_2,\hat{u}_1^n,\hat{u}_2^n,y^n) \notag\\
&= \hat{P}(x_1^n,x_2^n,z^n,u_1^n,u_2^n,s_1,f_1,m_1,s_2,f_2,m_2,\hat{u}_1^n,\hat{u}_2^n) \notag\\
&\phantom{wwwww} \times p(y^n|\hat{u}_1^n,\hat{u}_2^n,z^n) \notag\\
&\approx P(x_1^n,x_2^n,z^n,u_1^n,u_2^n,s_1,f_1,m_1,s_2,f_2,m_2) \notag\\
&\phantom{wwwww} \times \mathbbm{1}\{\hat{u}_1^n=u_1^n,\hat{u}_2^n=u_2^n\} p(y^n|\hat{u}_1^n,\hat{u}_2^n,z^n) \notag\\
&= P(x_1^n,x_2^n,z^n,u_1^n,u_2^n,s_1,f_1,m_1,s_2,f_2,m_2) \notag\\
&\phantom{wwwww} \times \mathbbm{1}\{\hat{u}_1^n=u_1^n,\hat{u}_2^n=u_2^n\} p(y^n|u_1^n,u_2^n,z^n) \notag\\
&= P(x_1^n,x_2^n,z^n,u_1^n,u_2^n,s_1,f_1,m_1,s_2,f_2,m_2,y^n) \notag\\
&\phantom{wwwww} \times \mathbbm{1}\{\hat{u}_1^n=u_1^n,\hat{u}_2^n=u_2^n\}.
\end{align}
Thus, using the first part of Lemma \ref{lem:a2}, we can conclude that
\begin{align}
\hat{P}(x_1^n,x_2^n,z^n,y^n,f_1,f_2) \approx P(x_1^n,x_2^n,z^n,y^n,f_1,f_2). \label{eq:osrbcond1}
\end{align}

We require $(X_1^n,X_2^n,Z^n,Y^n)$ to be independent of the extra shared randomness $(F_1,F_2)$ to eliminate them without disturbing the desired i.i.d. distribution. This can be accomplished by imposing the following conditions according to Theorem~\ref{lem:a1}.
\begin{align}
\tilde{R}_1 &\leq H(U_1|X_1,X_2,Y,Z), \label{eq:cond31} \\
\tilde{R}_2 &\leq H(U_2|X_1,X_2,Y,Z), \label{eq:cond32} \\
\tilde{R}_1+\tilde{R}_2 &\leq H(U_1,U_2|X_1,X_2,Y,Z). \label{eq:cond33}
\end{align}
This ensures that
\begin{align}
P(x_1^n,x_2^n,z^n,y^n,f_1,f_2) \approx p^{\text{U}}(f_1) p^{\text{U}}(f_2) p(x_1^n,x_2^n,z^n,y^n), 
\end{align}
which along with \eqref{eq:osrbcond1} and the triangle inequality, implies that
\begin{align}
\hat{P}(x_1^n,x_2^n,z^n,y^n,f_1,f_2) \approx p^{\text{U}}(f_1) p^{\text{U}}(f_2) p(x_1^n,x_2^n,z^n,y^n).
\end{align}
Hence there exists a fixed binning with corresponding pmf $\tilde{p}$ such that if we replace $P$ by $\tilde{p}$ in \eqref{eq:rcpmf} and denote the resulting pmf by $\hat{p}$, then 
\begin{align}
\hat{p}(x_1^n,x_2^n,z^n,y^n,f_1,f_2) \approx p^{\text{U}}(f_1) p^{\text{U}}(f_2) p(x_1^n,x_2^n,y^n,z^n).
\end{align}
Now the second part of Lemma \ref{lem:a2} allows us to conclude that there exist instances $F_1 = f_1^{*}, F_2=f_2^{*}$ such that 
\begin{align}
\hat{p}(x_1^n,x_2^n,z^n,y^n|f_1^{*},f_2^{*}) \approx p(x_1^n,x_2^n,z^n,y^n).
\end{align}
Now along with the rate constraints imposed in equations \eqref{eq:cond11} -- \eqref{eq:cond12}, \eqref{eq:cond21} -- \eqref{eq:cond23} and \eqref{eq:cond31} -- \eqref{eq:cond33}, we also need to impose the non-negativity constraints on all the rates. But it turns out that the constraints $\tilde{R}_1 \geq 0$ and $\tilde{R}_2 \geq 0$ are redundant, which can be shown along the lines of \cite[Remark 4]{yassaee2014achievability}. We prove that if $(\tilde{R}_1,\tilde{R}_2)$ are not necessarily all positive and satisfy \eqref{eq:cond11} -- \eqref{eq:cond12}, \eqref{eq:cond21} -- \eqref{eq:cond23} and \eqref{eq:cond31} -- \eqref{eq:cond33} along with $(R_1,R_2,R_{01},R_{02})$ for some $(U_1,U_2)$ such that $(U_1,U_2) \to (X_1,X_2) \to Z$, $U_1 \to X_1 \to X_2 \to U_2$ and $Y \to (U_1,U_2,Z) \to (X_1,X_2)$, then there exists $(\bar{U}_1,\bar{U}_2)$ with $(\bar{U}_1,\bar{U}_2) \to (X_1,X_2) \to Z$, $\bar{U}_1 \to X_1 \to X_2 \to \bar{U}_2$ and $Y \to (\bar{U}_1,\bar{U}_2,Z) \to (X_1,X_2)$ and $\bar{R}_1 \geq 0, \bar{R}_2 \geq 0$ such that $(\bar{R}_1,\bar{R}_2)$ along with $(R_1,R_2,R_{01},R_{02})$ satisfy \eqref{eq:cond11} -- \eqref{eq:cond12}, \eqref{eq:cond21} -- \eqref{eq:cond23} and \eqref{eq:cond31} -- \eqref{eq:cond33} for $(\bar{U}_1,\bar{U}_2)$ instead of $(U_1,U_2)$.

{\blue Towards this end}, suppose $\tilde{R}_1 < 0$ and $\tilde{R}_2 < 0$. Let $(W_1,W_2)$ be random variables such that $H(W_1) > |\tilde{R}_1|$ and $H(W_2) > |\tilde{R}_2|$. We also assume that $W_1$ as well as $W_2$ are independent of all the other random variables. Define $\bar{R}_1=\tilde{R}_1+H(W_1)$, $\bar{R}_2=\tilde{R}_2+H(W_2)$ and $\bar{U}_1=(U_1,W_1)$, $\bar{U}_2=(U_2,W_2)$. Clearly, we have $\bar{R}_1, \bar{R}_2 \geq 0$ and it is easy to see that $(\bar{R}_1,\bar{R}_2)$ along with $(R_1,R_2,R_{01},R_{02})$ satisfy \eqref{eq:cond11} -- \eqref{eq:cond12}, \eqref{eq:cond21} -- \eqref{eq:cond23} and \eqref{eq:cond31} -- \eqref{eq:cond33} for $(\bar{U}_1,\bar{U}_2)$ using the independence of $W_1$ and $W_2$ from all other random variables and the fact that $(\tilde{R}_1,\tilde{R}_2)$ satisfy \eqref{eq:cond11} -- \eqref{eq:cond12}, \eqref{eq:cond21} -- \eqref{eq:cond23} and \eqref{eq:cond31} -- \eqref{eq:cond33} along with $(R_1,R_2,R_{01},R_{02})$. {\blue Next suppose $\tilde{R}_1 < 0, \tilde{R}_2 \geq 0$ -- here the proof follows by defining $W_1,\bar{R}_1$ and $\bar{U}_1$ as above and noting that $(\bar{R}_1,\tilde{R}_2)$ along with $(R_1,R_2,R_{01},R_{02})$ satisfy \eqref{eq:cond11} -- \eqref{eq:cond12}, \eqref{eq:cond21} -- \eqref{eq:cond23} and \eqref{eq:cond31} -- \eqref{eq:cond33} for $(\bar{U}_1,U_2)$. The remaining configuration $\tilde{R}_1 \geq 0, \tilde{R}_2 < 0$ can be dealt similarly.} Finally on eliminating $(\tilde{R}_1,\tilde{R}_2)$ from equations \eqref{eq:cond11} -- \eqref{eq:cond12}, \eqref{eq:cond21} -- \eqref{eq:cond23} and \eqref{eq:cond31} -- \eqref{eq:cond33} by Fourier-Motzkin elimination (FME), we obtain the rate constraints
\begin{align}
R_1 &\geq H(U_1|U_2,Z)-H(U_1|X_1,U_2,Z) = I(U_1;X_1|U_2,Z), \label{eq:reff1}\\
R_2 &\geq H(U_2|U_1,Z)-H(U_2|X_2,U_1,Z) = I(U_2;X_2|U_1,Z),
\end{align}
\begin{align}
R_1+R_2 &\geq H(U_1,U_2|Z)-H(U_1,U_2|X_1,X_2,Z) \notag\\
& = I(U_1,U_2;X_1,X_2|Z),
\end{align}
\begin{align}
R_1+R_{01} &\geq H(U_1|U_2,Z)-H(U_1|X_1,X_2,Y,Z) \notag\\
& = I(U_1;X_1,X_2,Y|Z)-I(U_1;U_2|Z), \\
R_2+R_{02} &\geq H(U_2|U_1,Z)-H(U_2|X_1,X_2,Y,Z) \notag\\
& = I(U_2;X_1,X_2,Y|Z)-I(U_1;U_2|Z),
\end{align}
\begin{align}
&R_1+R_2+R_{01} \notag\\
&\phantom{w}\geq H(U_1,U_2|Z)\!-\!H(U_2|X_2,Z)\!-\!H(U_1|X_1,X_2,Y,Z) \notag\\
&\phantom{w} = I(U_1;X_1,X_2,Y|Z)+I(U_2;X_2|U_1,Z), \\
&R_1+R_2+R_{02} \notag\\
&\phantom{w}\geq H(U_1,U_2|Z)\!-\!H(U_1|X_1,Z)\!-\!H(U_2|X_1,X_2,Y,Z) \notag\\
&\phantom{w} = I(U_2;X_1,X_2,Y|Z)+I(U_1;X_1|U_2,Z),
\end{align}
\begin{align}
&R_1+R_2+R_{01}+R_{02} \notag\\
&\phantom{w} \geq H(U_1,U_2|Z)-H(U_1,U_2|X_1,X_2,Y,Z) \notag\\
&\phantom{w} = I(U_1,U_2;X_1,X_2,Y|Z). \label{eq:reff8}
\end{align}
Thus when the rate constraints in \eqref{eq:reff1} -- \eqref{eq:reff8} are met, there exists a sequence of $(2^{nR_{01}},2^{nR_{02}},2^{nR_1},2^{nR_2},n)$ codes with encoders and decoders as described in the second protocol with the particular realization of binning along with the fixed instances $f_1^{*},f_2^{*}$ resulting in vanishing total variation distance.

\section{Specialization of Theorem \ref{thm:encsideIB} to Deterministic Function Computation} \label{app:rem2}
Let $|\mathcal{T}|=1$ for simplicity. When $Y$ is a deterministic function of $(X_1,X_2,Z)$, we note the following simplifications to the rate constraints of Theorem \ref{thm:encsideIB}.
\begin{align}
R_1+R_{01} &\geq I(U_1;X_1,X_2,Y|Z)-I(U_1;U_2|Z) \notag\\
&\stackrel{(a)}= I(U_1;X_1,X_2|Z)-I(U_1;U_2|Z) \notag\\
&= I(U_1;X_1,X_2|U_2,Z)-I(U_1;U_2|X_1,X_2,Z) \notag\\
&\stackrel{(b)}= I(U_1;X_1|U_2,Z),
\end{align}
where (a) follows since $Y$ is determined by $(X_1,X_2,Z)$ while (b) follows from the Markov chains $U_1 \to X_1 \to (U_2,X_2,Z)$ and $U_2 \to X_2 \to (U_1,X_1,Z)$. Similarly, it follows that
\begin{align}
R_2+R_{02} &\geq I(U_2;X_2|U_1,Z).
\end{align}
Furthermore, we note that
\begin{align}
R_1+R_2+R_{01} &\geq I(U_1;X_1,X_2,Y|Z)+I(U_2;X_2|U_1,Z) \notag\\
&\stackrel{(a)}=  I(U_1;X_1,X_2|Z)+I(U_2;X_2|U_1,Z) \notag\\
&\stackrel{(b)}= I(U_1;X_1,X_2|Z)+I(U_2;X_1,X_2|U_1,Z) \notag\\
&= I(U_1,U_2;X_1,X_2|Z),
\end{align}
where (a) follows since $Y$ is determined by $(X_1,X_2,Z)$ while (b) follows from the Markov chain $U_2 \to X_2 \to (U_1,X_1,Z)$. Similarly, it follows that
\begin{align}
R_1+R_2+R_{02} &\geq I(U_1,U_2;X_1,X_2|Z).
\end{align}
Finally, we note that
\begin{align}
R_1+R_2+R_{01}+R_{02} &\geq I(U_1,U_2;X_1,X_2,Y|Z) \notag\\
&= I(U_1,U_2;X_1,X_2|Z).
\end{align}
Thus when $Y$ is a deterministic function of $(X_1,X_2,Z)$, the rate constraints involving the shared randomness rates $R_{01}$ and $R_{02}$ become redundant in Theorem \ref{thm:encsideIB}. Hence, the region  simplifies to the set of rate pairs $(R_1,R_2)$ satisfying
\begin{align*}
R_1 &\geq I(U_1;X_1|U_2,Z), \\
R_2 &\geq I(U_2;X_2|U_1,Z), \\
R_1+R_2 &\geq I(U_1,U_2;X_1,X_2|Z), 
\end{align*}
for some p.m.f. 
\begin{align*}
p(x_1,x_2,z,u_1,u_2)=p(x_1,x_2,z)p(u_1|x_1)p(u_2|x_2)
\end{align*}
such that $H(Y|U_1,U_2,Z)=0$. This is precisely the inner bound of \cite[Theorem 2]{SefidgaranT16} specialized to the multiple-access network where the condition $H(Y|U_1,U_2,Z)=0$, in conjunction with the structure of the p.m.f., is expressed in an alternate form involving graph entropy using \cite[Lemma 3]{SefidgaranA11}. 

\section{Specialization of Theorem \ref{thm:indep} to Deterministic Function Computation} \label{app:rem5}
Let $|\mathcal{T}|=1$ again for simplicity. The joint distribution in \eqref{pmfstructure1} is such that $(U_1,X_1)\to Z\to (U_2,X_2)$
 is a Markov chain. This leads to the following simplification to the rate constraints involving shared randomness when $Y$ is a deterministic function of $(X_1,X_2,Z)$, say, $Y=f(X_1,X_2,Z)$.
 \begin{align*}
 I(U_1;X_1,Y|X_2,Z)&=I(U_1;X_1,X_2,Y|Z)\\
 &=I(U_1;X_1,X_2|Z)\\
 &=I(U_1;X_1|Z).
 \end{align*}
 This simplification renders the bound on $R_1+R_{01}$ in Theorem~\ref{thm:indep} redundant. Hence, the region simplifies to the set of rate pairs $(R_1,R_2)$ satisfying 
 \begin{align*}
 R_1&\geq I(U_1;X_1|Z),\\
 R_2&\geq I(U_2;X_2|Z),
 \end{align*}  
 for some p.m.f. 
 \begin{align*}
 p(x_1,&x_2,z,u_1,u_2,y)=\\
 &\hspace{12pt}p(z)p(x_1|z)p(x_2|z)p(u_1|x_1)p(u_2|x_2)p(y|u_1,u_2,z)
 \end{align*}
  such that $H(Y|U_1,U_2,Z)=0$. This is the rate region in  \cite[Theorem 3]{SefidgaranT16} specialized to the multiple-access network where it is expressed in an alternate form involving graph entropies, which follows by using \cite[Proof of Theorem 3]{SefidgaranA11}.


\section{Proof of Claim \ref{claim:ex}}\label{apendix:example}
It suffices to restrict attention to realizations $u_1,u_2$ for which $H(X_1|U_1=u_1)>0$ and $H(X_2|U_2=u_2)>0$. We prove the claim in three steps.

\noindent{Step $1$}: We prove that under Case $3$, for any $u_1$ with $P(U_1=u_1)>0$, there exists a $k\in\{1,2\}$ such that $H(X_{1k}|U_1=u_1)=0$.

\noindent{Step $2$}: Then for any $u_2$ with $P(U_2=u_2)>0$, for the $k$ from Step $1$, we show that $H(X_{2k}|U_2=u_2)=0$.

 \noindent{Step $3$}: Finally, for any $u_1^\prime$ with $u_1^\prime\neq u_1$ and $P(U_1=u_1^\prime)>0$, for the $k$ from Step $1$, we prove that $H(X_{1k}|U_1=u_1^\prime)=0$.

\vspace{2mm}
\noindent\underline{Step $1$:} We prove this by contradiction. Suppose $H(X_{1i}|U_1=u_1)>0$, for $i=1,2$. This implies that the support of $p_{X_1|U_1=u_1}$ cannot be of the form $\{b_1,b_2\},\{(b_1,b_2),(b_1,1-b_2)\},\{(b_1,b_2),(1-b_1,b_2)\}$, for $b_1,b_2\in\{0,1\}$. The remaining possibility is that the support has to be a superset of either $\{(0,1),(1,0)\}$ or $\{(0,0),(1,1)\}$.  In the sequel, we use the independence of $(U_1,X_1)$ and $(U_2,X_2)$, and the Markov chain $Y\rightarrow(U_1,U_2)\rightarrow(X_1,X_2)$ repeatedly. Consider a $u_2$ such that $P(U_2=u_2)>0$. It turns out that the support cannot be a superset of $\{(0,0),(1,1)\}$. To see this, first notice that $P(X_1=(0,0)|U_1=u_1,U_2=u_2),P(X_1=(1,1)|U_1=u_1,U_2=u_2)>0$. Now since $X_{1J}-(U_1,U_2)-X_1$, we have $P(X_{1J}=0|U_1=u_1,U_2=u_2)=P(X_{1J}=0|U_1=u_1,U_2=u_2,X_1=(1,1))=0$, where the last equality follows from the correctness of output $Y=(X_{1J},X_{2J})$. Similarly, $P(X_{1J}=1|U_1=u_1,U_2=u_2)=P(X_{1J}=1|U_1=u_1,U_2=u_2,X_1=(0,0))=0$. This is a contradiction since $p_{X_{1J}|U_1=u_1,U_2=u_2}$ has to be a probability distribution. The only other possibility is that the support of $p_{X_1|U_1=u_1}$ is a superset of $\{(0,1),(1,0)\}$. Since $H(X_2|U_2)>0$, there exists a $u_2$ with $P(U_2=u_2)>0$ such that $H(X_2|U_2=u_2)>0$. So, either exactly one or both of $H(X_{21}|U_2=u_2)$ and $H(X_{22}|U_2=u_2)$ will be strictly positive. Suppose exactly one of them is strictly positive. Without loss of generality, suppose $H(X_{21}|U_2=u_2)>0$, i.e., $P(X_{21}=0|U_2=u_2),P(X_{21}=1|U_2=u_2)>0$ and $H(X_{22}|U_2=u_2)=0$. Also, assume that $P(X_{22}=1|U_2=u_2)=1$. Consider the probability distribution $p_{Y|U_1=u_1,U_2=u_2}$. This is well defined because $P(U_1=u_1,U_2=u_2)>0$ as $U_1$ is independent of $U_2$ and $P(U_1=u_1),P(U_2=u_2)>0$.  From the above we have the following.

$P(X_1=(0,1),X_2=(0,1)|U_1=u_1,U_2=u_2)>0$ which implies that $P(Y=(0,1)|U_1=u_1,U_2=u_2)=0$ and $P(Y=(1,0)|U_1=u_1,U_2=u_2)=0$. 


$P(X_1=(1,0),X_2=(0,1)|U_1=u_1,U_2=u_2)>0$ which implies that $P(Y=(0,0)|U_1=u_1,U_2=u_2)=0$ and $P(Y=(1,1)|U_1=u_1,U_2=u_2)=0$. 

This is a contradiction since $p_{Y|U_1=u_1,U_2=u_2}$ has to be a probability distribution. Now suppose both $H(X_{21}|U_2=u_2)$ and $H(X_{22}|U_2=u_2)$ are strictly positive. Since $Y_2 := X_{2J}\in\{X_{21},X_{22}\}$, the only possibility is that $p_{X_2|U_2=u_2}$ has a support that is a superset of $\{(0,1),(1,0)\}$. Then we have the following.

$P(X_1=(0,1),X_2=(0,1)|U_1=u_1,U_2=u_2)>0$ which implies that $P(Y=(0,1)|U_1=u_1,U_2=u_2)=0$ and $P(Y=(1,0)|U_1=u_1,U_2=u_2)=0$. 

$P(X_1=(0,1),X_2=(1,0)|U_1=u_1,U_2=u_2)>0$ which implies that $P(Y=(0,0)|U_1=u_1,U_2=u_2)=0$ and $P(Y=(1,1)|U_1=u_1,U_2=u_2)=0$. 

This is a contradiction since $p_{Y|U_1=u_1,U_2=u_2}$ has to be a probability distribution. Thus, we have, under Case $3$, if $H(X_1|U_1=u_1)>0$, then there exists a $k\in\{1,2\}$ such that $H(X_{1k}|U_1=u_1)=0$.

\par\noindent\underline{Step $2$:} Note that there exists $u_1$ with $P(U_1=u_1)>0$ such that $H(X_1|U_1=u_1)>0$ since $H(X_1|U_1)>0$.  So, by the discussion in Step~$1$, there exists a $k$ such that $H(X_{1k}|U_1=u_1)=0$. Note that $H(X_{1k^\prime}|U_1=u_1)>0$, where $k^\prime=3-k$, since $H(X_1|U_1=u_1)>0$.  Similarly, there exists $u_2$ with $P(U_2=u_2)>0$ such that $H(X_2|U_2=u_2)>0$. Now we show that $H(X_{2k}|U_2=u_2)=0$. We prove this by contradiction. Suppose $H(X_{2k}|U_2=u_2)>0$. Then in view of the above discussion, we have $H(X_{2k^\prime}|U_2=u_2)=0$. Without loss of generality, let $k=1$, i.e., $k^\prime=2$. Note that $p_{X_{12},X_{21}|U_1=u_1,U_2=u_2}=p_{X_{12}|U_1=u_1}\cdot p_{X_{21}|U_2=u_2}$ has full support since $H(X_{12}|U_1=u_1),H(X_{21}|U_2=u_2)>0$. Also, assume that $P(X_{11}=0|U_1=u_1)=1$ and $P(X_{22}=0|U_2=u_2)=1$ (other choices can be dealt similarly). Consider the probability distribution $p_{Y|U_1=u_1,U_2=u_2}$. Then we have the following.

$P(X_1=(0,1),X_2=(1,0)|U_1=u_1,U_2=u_2)=P(X_{12}=1,X_{21}=1|U_1=u_1,U_2=u_2)>0$, which implies that $P(Y=(0,0)|U_1=u_1,U_2=u_2)=0$ and $P(Y=(1,1)|U_1=u_1,U_2=u_2)=0$ since we have the Markov chain $Y-(U_1,U_2)-(X_1,X_2)$.

$P(X_1=(0,0),X_2=(0,0)|U_1=u_1,U_2=u_2)=P(X_{12}=0,X_{21}=0|U_1=u_1,U_2=u_2)>0$, which implies that $P(Y=(0,1)|U_1=u_1,U_2=u_2)=0$ and $P(Y=(1,0)|U_1=u_1,U_2=u_2)=0$ since we have the Markov chain $Y-(U_1,U_2)-(X_1,X_2)$.

This is a contradiction since $p_{Y|U_1=u_1,U_2=u_2}$ has to be a probability distribution. 

\par\noindent\underline{Step $3$:} Now suppose that there exists a $u_1^\prime$ with $u_1^\prime\neq u_1$ and $P(U_1=u_1^\prime)>0$ such that $H(X_1|U_1=u_1^\prime)>0$. Since there is a $u_2$ such that $H(X_{2k}|U_2=u_2)=0$, by the same argument as above (reversing the roles of $1$ and $2$), we have $H(X_{1k}|U_1=u_1^\prime)=0$. 
This completes the proof of Claim~\ref{claim:ex}.

{\blue\section{Proof of Lemma~\ref{lem:cardbnd}}\label{app:cardbnd} 
The cardinality bound $|\mathcal{T}| \leq 3$ on the time-sharing random variable $T$ can be obtained using standard arguments based on the support lemma~\cite[Appendix C]{el2011network}\footnote{It suffices for the alphabet $\mathcal{T}$ to have 3 elements to preserve $I(U_1;X_1|Z,T)$, $I(U_2;X_2;|Z,T)$, and $I(U_1;X_1,Y|X_2,Z,T)$, thereby preserving the rate region.}. We now focus on cardinality bounds for $U_1$ and $U_2$.
Consider the region ${S}_\epsilon$ with $|\mathcal{T}|=1$:
\begin{align}
R_1&\geq I(U_1;X_1|Z) \label{eq:T3r1}\\
R_2&\geq I(U_2;X_2|Z) \label{eq:T3r2},\\
R_1+R_{01}&\geq I(U_1;X_1,Y|X_2,Z)-g(\epsilon), \label{eq:T3r3}
\end{align}
for some p.m.f. 
\begin{align}\label{p.m.f.structure1}
p(x_1,&x_2,z,u_1,u_2,y)\nonumber\\
&=p(z)p(x_1|z)p(x_2|z)p(u_1|x_1)p(u_2|x_2)p(y|u_1,u_2,z)
\end{align}
such that 
\begin{align}
\lVert \sum_{u_1,u_2} p(x_1,x_2,z,u_1,u_2,y)- q(x_1,x_2,z,y)\rVert\leq \epsilon. \label{eq:pmfmarginalize}
\end{align}
Call this region $\mathcal{C}$. We now show that the auxiliary cardinalities can be restricted to $|\mathcal{U}_1| \leq |\mathcal{X}_1||\mathcal{X}_2||\mathcal{Y}||\mathcal{Z}|$ and $|\mathcal{U}_2| \leq |\mathcal{U}_1||\mathcal{X}_1||\mathcal{X}_2||\mathcal{Y}||\mathcal{Z}|$ via the perturbation argument of \cite{gohari2012evaluation}. This is done in two steps.
\begin{itemize}
\item \textbf{Step 1:} We prove that 
\begin{align}
\mathcal{C} = \textup{Closure}\left(\underset{K_{1},K_{2} \geq 0}{\bigcup} \mathcal{C}^{K_{1},K_{2}}\right),
\end{align}
where the region $\mathcal{C}^{K_{1},K_{2}}$ for positive integers $K_1,K_2$ is given by the union of rate triples $(R_1,R_2,R_{01})$ satisfying \eqref{eq:T3r1}--\eqref{eq:T3r3} over random variables $(U_1,U_2,X_1,X_2,Y,Z)$ with cardinality bounds $|\mathcal{U}_1| \leq K_{1}$ $\&$ $|\mathcal{U}_2| \leq K_{2}$ and having a joint p.m.f. $p(z)p(x_1|z)p(x_2|z)p(u_1|x_1)p(u_2|x_2)p(y|u_1,u_2,z)$ such that $\lVert p(x_1,x_2,z,y)- q(x_1,x_2,z,y)\rVert\leq \epsilon$.
\item \textbf{Step 2:} If $|\mathcal{U}_1| \leq K_{1}$ and $|\mathcal{U}_2| \leq K_{2}$ for some constants $K_{1},K_{2}$, we show that the auxiliary cardinalities can be brought down to $|\mathcal{U}_1| \leq |\mathcal{X}_1||\mathcal{X}_2||\mathcal{Y}||\mathcal{Z}|$ and $|\mathcal{U}_2| \leq |\mathcal{U}_1||\mathcal{X}_1||\mathcal{X}_2||\mathcal{Y}||\mathcal{Z}|$.
\end{itemize}

\vspace{2mm}

\noindent \textbf{Proof Step 1:}\\
It suffices to show that any rate triple $(R_1,R_2,R_{01}) \in \mathcal{C}$ is a limit point of the set $\underset{K_{1},K_{2} \geq 0}{\bigcup} \mathcal{C}^{K_{1},K_{2}}$. Firstly, if $(R_1,R_2,R_{01}) \in \mathcal{C}$, then random variables $(U_1,U_2,X_1,X_2,Y,Z)$ satisfying \eqref{eq:T3r1}--\eqref{eq:T3r3} exist such that their joint p.m.f. is of the form \eqref{p.m.f.structure1} and \eqref{eq:pmfmarginalize} holds. Assume that $\mathcal{U}_j = \{1,2,\cdots\}$ for $j=1,2$. Define modified versions $(U_1',U_2')$ of random variables $(U_1,U_2)$ taking values in $\{1,2,\cdots,m\} \cup \mathcal{X}_1$ and $\{1,2,\cdots,m\} \cup \mathcal{X}_2$ respectively, where $m$ is an integer. The alphabet cardinalities of $(U_1',U_2')$ are $m+|\mathcal{X}_1|$ and $m+|\mathcal{X}_2|$ respectively. Let the conditional p.m.f. of $U_j'$ given $X_j$ for $j \in \{1,2\}$, $x_j \in \mathcal{X}_j$ be specified as follows:
\begin{align} 
p_{U_j'|X_j}(i|x_j)& = p_{U_j|X_j}(i|x_j), i=\{1,2,\cdots,m\}, \label{eq:truncpmf1}\\
p_{U_j'|X_j}(x_j'|x_j)& = 0 \:\: \forall \:\: x_j' \neq x_j, \label{eq:truncpmf2}\\
p_{U_j'|X_j}(x_j|x_j) &= \textup{Pr}(U_j>m|X_j=x_j)= \!\!\sum_{i=m+1}^{\infty}p_{U_j|X_j}(i|x_j). \label{eq:truncpmf3}
\end{align}
From the definitions in \eqref{eq:truncpmf1}--\eqref{eq:truncpmf3} and the fact that the original random variables $(U_1,U_2,X_1,X_2,Y,Z)$ satisfy $U_1 \to X_1 \to (U_2,X_2,Z)$ and $U_2 \to X_2 \to (U_1,X_1,Z)$, it follows that the Markov chains $U_1' \to X_1 \to (U_2',X_2,Z)$ and $U_2' \to X_2 \to (U_1',X_1,Z)$ hold as well.

To define $p_{Y|U_1',U_2',Z}$, consider new auxiliary random variables $(U_1'',U_2'')$ whose alphabets are the same as that of $(U_1,U_2)$, and their conditional p.m.f. given $U_j'$ for $j \in \{1,2\}$ is specified as follows:
\begin{itemize}
\item If $U_j' \in \{1,2,\cdots,m\}$, then $U_j''=U_j'$, i.e.,
\begin{align}
p_{U_j''|U_j'}(i|i) = 1, i=1,2,\cdots,m.
\end{align}
\item Otherwise, let
\begin{align}
&p_{U_j''|U_j'}(i|x_j) = \nonumber\\
&\begin{cases}
\textup{Pr}(U_j=i|X_j=x_j,U_j>m), &{\!\!i=m+1,m+2,\cdots}\\
0, &{\!\!i=1,2,\cdots,m}.
\end{cases} \label{eq:truncpmf4}
\end{align}
\end{itemize}
Note that we have
\begin{align}
\sum_{u_j'} p_{U_j'|X_j}&(u_j'|x_j)p_{U_j''|U_j'}(u_j''|u_j') \nonumber\\
&= p_{U_j|X_j}(u_j''|x_j), u_j'' \in \mathcal{U}_j, x_j \in \mathcal{X}_j, j=1,2. \label{eq:pmfcondition}
\end{align}
Further, let
\begin{align}
p_{Y|U_1'',U_2'',Z}(y|u_1,u_2,z) = p_{Y|U_1,U_2,Z}(y|u_1,u_2,z). \label{eq:truncpmf5}
\end{align}
The conditional distribution $p_{Y|U_1',U_2',Z}$ is defined as
\begin{align}
&p_{Y|U_1',U_2',Z}(y|u_1',u_2',z) \nonumber\\
&=\!\! \sum_{u_1'',u_2''} \! p_{U_1''|U_1'}(u_1''|u_1') p_{U_2''|U_2'}(u_2''|u_2') p_{Y|U_1'',U_2'',Z}(y|u_1'',u_2'',z) \label{eq:truncpmf6}\\
&=\!\! \sum_{u_1'',u_2''} \! p_{U_1''|U_1'}(u_1''|u_1') p_{U_2''|U_2'}(u_2''|u_2') p_{Y|U_1,U_2,Z}(y|u_1'',u_2'',z). 
\end{align}
From the definitions in \eqref{eq:truncpmf4}--\eqref{eq:truncpmf6} and the fact that the original random variables $(U_1,U_2,X_1,X_2,Y,Z)$ satisfy $Y \to (U_1,U_2,Z) \to (X_1,X_2)$, it follows that the Markov chain $Y \to (U_1',U_2',Z) \to (X_1,X_2)$ holds as well. Due to these Markov constraints, the joint distribution of $(U_1',U_2',X_1,X_2,Y,Z)$ is $p_{U_1',U_2',X_1,X_2,Y,Z}=p_{X_1,X_2,Z}p_{U_1'|X_1}p_{U_2'|X_2}p_{Y|U_1',U_2',Z}$. On marginalizing away $(U_1',U_2')$ from $p_{U_1',U_2',X_1,X_2,Y,Z}$, we obtain the p.m.f. $p_{X_1,X_2,Y,Z}$, which can be seen as follows.
\begin{align}
&\sum_{u_1',u_2'} (p_{X_1,X_2,Z}(x_1,x_2,z)p_{U_1'|X_1}(u_1'|x_1)p_{U_2'|X_2}(u_2'|x_2)\nonumber\\
&\hspace{1cm}\times p_{Y|U_1',U_2',Z}(y|u_1',u_2',z)) \notag\\
&= \sum_{u_1',u_2'} p_{X_1,X_2,Z}(x_1,x_2,z)p_{U_1'|X_1}(u_1'|x_1)p_{U_2'|X_2}(u_2'|x_2) \cdot\nonumber\\
&(\sum_{u_1'',u_2''} p_{U_1''|U_1'}(u_1''|u_1')p_{U_2''|U_2'}(u_2''|u_2')p_{Y|U_1'',U_2'',Z}(y|u_1'',u_2'',z)) \notag\\
&=\! \!\!\sum_{u_1'',u_2''}\!\! p_{X_1,X_2,Z}(x_1,x_2,z) ( \sum_{u_1'} p_{U_1'|X_1}(u_1'|x_1)p_{U_1''|U_1'}(u_1''|u_1') )\nonumber\\
 &( \sum_{u_2'} p_{U_2'|X_2}(u_2'|x_2)p_{U_2''|U_2'}(u_2''|u_2') ) p_{Y|U_1,U_2,Z}(y|u_1'',u_2'',z) \notag\\
& \stackrel{(a)}= \sum_{u_1'',u_2''} p_{X_1,X_2,Z}(x_1,x_2,z) p_{U_1|X_1}(u_1''|x_1) p_{U_2|X_2}(u_2''|x_2) \nonumber\\
&\hspace{1cm}p_{Y|U_1,U_2,Z}(y|u_1'',u_2'',z) \notag\\
&= p_{X_1,X_2,Y,Z}(x_1,x_2,y,z),
\end{align}
where (a) follows from \eqref{eq:pmfcondition}. Hence $p_{U_1' U_2' X_1 X_2 Y Z}$ satisfies \eqref{eq:pmfmarginalize}.

It follows that the joint distribution of the random variables $(U_{1}',U_{2}',X_{1},X_{2},Y,Z)$ converges to the the joint distribution of $(U_{1},U_{2},X_{1},X_{2},Y,Z)$ in the limit $m \to \infty$. As a result, the mutual information terms $I(U_{1}';X_{1}|Z)$, $I(U_{2}';X_{2}|Z)$ and $I(U_1';X_1,Y|X_2,Z)$ converge to $I(U_{1};X_{1}|Z)$, $I(U_{2};X_{2}|Z)$ and $I(U_1;X_1,Y|X_2,Z)$ respectively. Hence, we conclude that the given rate triple $(R_1,R_2,R_{01})$ is a limit point of the set $\underset{K_{1},K_{2} \geq 0}{\bigcup} \mathcal{C}^{K_{1},K_{2}}$.

\vspace{2mm}

\noindent \textbf{Proof Step 2:}\\
It suffices to consider an optimization of the weighted sum term $\lambda_1 I(U_1;X_1|Z)+\lambda_2 I(U_2;X_2|Z)+\lambda_3 I(U_1;X_1,Y|X_2,Z)$ for non-negative reals $\lambda_1,\lambda_2,\lambda_3$, and find new auxiliary random variables whose cardinalities are bounded while not increasing the weighted sum and preserving the conditions \eqref{p.m.f.structure1} and on \eqref{eq:pmfmarginalize} on $p(x_1,x_2,u_1,u_2,y,z)$.

For a given $p(u_1,u_2,x_1,x_2,y,z)$, consider the perturbation defined by
\begin{align*}
p_{\epsilon}(u_1,u_2,x_1,x_2,y,z)=p(u_1,u_2,x_1,x_2,y,z)\left(1+\epsilon \phi(u_1)\right).
\end{align*}
For $p_{\epsilon}(u_1,u_2,x_1,x_2,y,z)$ to be a valid p.m.f., we require that $\left(1+\epsilon \phi(u_1)\right) \geq 0$ for all $u_1$, and $\sum_{u_1}p(u_1)\phi(u_1)=0$. Furthermore, we will consider perturbations $\phi(u_1)$ such that
\begin{align} \label{eq:perturbT3}
&\mathbb{E}\left[\phi(U_1)|X_1=x_1,X_2=x_2,Y=y,Z=z\right] \nonumber\\
&= \sum_{u_1} p(u_1|x_1,x_2,y,z) \phi(u_1) = 0, \:\: \forall \:\: x_1, x_2, y, z.
\end{align}
Observe that such a non-zero perturbation satisfying \eqref{eq:perturbT3} (which also implies $\sum_{u_1}p(u_1)\phi(u_1)=0$) exists as long as $|\mathcal{U}_1| > |\mathcal{X}_1||\mathcal{X}_2||\mathcal{Y}||\mathcal{Z}|$, since the null-space of the constraints has rank at most $|\mathcal{X}_1||\mathcal{X}_2||\mathcal{Y}||\mathcal{Z}|$. For sufficiently small values of $|\epsilon|$, we also have $\left(1+\epsilon \phi(u_1)\right) \geq 0$ for all $u_1$.
Note that this perturbation preserves the distribution of $(X_1,X_2,Y,Z)$, i.e., $p(x_1,x_2,y,z)$. This follows since
\begin{align}
&p_{\epsilon}(x_1,x_2,y,z) = \sum_{u_1,u_2} p_{\epsilon}(u_1,u_2,x_1,x_2,y,z) \nonumber\\
&= \sum_{u_1,u_2} p(u_1,u_2,x_1,x_2,y,z)\left(1+\epsilon \phi(u_1)\right) \notag\\
&= p(x_1,x_2,y,z) \nonumber\\
&\hspace{0.9cm}\{\sum_{u_1} p(u_1|x_1,x_2,y,z)+\epsilon \sum_{u_1} p(u_1|x_1,x_2,y,z) \phi(u_1) \} \nonumber\\
&= p(x_1,x_2,y,z), \label{eq:pertpreserveQ}
\end{align}
where the last step follows from \eqref{eq:perturbT3}.

We now show that the perturbed distribution $p_{\epsilon}(\cdot)$ preserves the structure of the p.m.f. in \eqref{p.m.f.structure1}. To see  this, note that
\begin{align}
&p_{\epsilon}(x_1,x_2,z,u_1,u_2,y)\nonumber\\
&=p(x_1,x_2,z,u_1,u_2,y)\left(1+\epsilon \phi(u_1)\right) \notag\\
&= p(x_1)p(z|x_1)\left\{p(u_1|x_1)\left(1+\epsilon \phi(u_1)\right)\right\}p(x_2|z)p(u_2|x_2)\nonumber\\
&\hspace{1cm}p(y|u_1,u_2,z).
\end{align}
Marginalizing away $Z,X_2,U_2,Y$, we have
\begin{align}
p_{\epsilon}(x_1,u_1)&=p(x_1)p(u_1|x_1)\left(1+\epsilon \phi(u_1)\right).
\end{align}
Since $p_{\epsilon}(x_1)=p(x_1)$ (by \eqref{eq:pertpreserveQ}), we have $p_{\epsilon}(u_1|x_1)=p(u_1|x_1)\left(1+\epsilon \phi(u_1)\right)$. Thus,
\begin{align}
&p_{\epsilon}(x_1,x_2,z,u_1,u_2,y)\nonumber\\
&=p(z)p(x_1|z)p(x_2|z)p_{\epsilon}(u_1|x_1)p(u_2|x_2)p(y|u_1,u_2,z), \label{eq:pmfuxz}
\end{align}
which is of the form \eqref{p.m.f.structure1}.

Now if the distribution $p(u_1,u_2,x_1,x_2,y,z)$ minimizes $\lambda_1 I(U_1;X_1|Z)+\lambda_2 I(U_2;X_2|Z)+\lambda_3 I(U_1;X_1,Y|X_2,Z)$, then for any valid perturbation, we must have the extremality condition:
\begin{align}
&\frac{d}{d \epsilon} (\lambda_1 I_{\epsilon}(U_1;X_1|Z)+\lambda_2 I_{\epsilon}(U_2;X_2|Z)\nonumber\\
&\hspace{1cm}+\lambda_3 I_{\epsilon}(U_1;X_1,Y|X_2,Z)) \Big|_{\epsilon=0} = 0, \label{eq:extrempert}
\end{align}
where the subscript $\epsilon$ in the mutual information terms is used to denote that these are evaluated under the perturbed distribution $p_{\epsilon}(u_1,u_2,x_1,x_2,y,z)$.
We examine the weighted sum term under the perturbed distribution $p_{\epsilon}(\cdot)$.
\begin{align}
&\lambda_1 I_{\epsilon}(U_1;X_1|Z)+\lambda_2 I_{\epsilon}(U_2;X_2|Z)+\lambda_3 I_{\epsilon}(U_1;X_1,Y|X_2,Z) \notag\\
&= \lambda_1\left(H_{\epsilon}(X_1,Z)-H_{\epsilon}(Z)+H_{\epsilon}(U_1,Z)-H_{\epsilon}(U_1,X_1,Z)\right)\nonumber\\
&+\lambda_2\left(H_{\epsilon}(X_2,Z)-H_{\epsilon}(Z)+H_{\epsilon}(U_2,Z)-H_{\epsilon}(U_2,X_2,Z)\right) \notag\\
& +\lambda_3 (H_{\epsilon}(X_1,X_2,Y,Z)-H_{\epsilon}(X_2,Z)+H_{\epsilon}(U_1,X_2,Z)\nonumber\\
&\hspace{1cm}-H_{\epsilon}(U_1,X_1,X_2,Y,Z)) \notag\\
&\stackrel{(a)}= \lambda_1\left(H(X_1,Z)-H(Z)+H_{\epsilon}(U_1,Z)-H_{\epsilon}(U_1,X_1,Z)\right)\nonumber\\
&+\lambda_2\left(H(X_2,Z)-H(Z)+H_{\epsilon}(U_2,Z)-H_{\epsilon}(U_2,X_2,Z)\right) \notag\\
&+\lambda_3 (H(X_1,X_2,Y,Z)-H(X_2,Z)+H_{\epsilon}(U_1,X_2,Z)\nonumber\\
&\hspace{1cm}-H_{\epsilon}(U_1,X_1,X_2,Y,Z)) \notag\\
&\stackrel{(b)}= \lambda_1\Big(H(X_1,Z)-H(Z)+H(U_1,Z)+\epsilon H_{\phi}(U_1,Z)\nonumber\\
&\hspace{1cm}-[\sum_{u_1,z} p(u_1,z)(1+\epsilon \phi(u_1))\log(1+\epsilon \phi(u_1))]\nonumber\\
&\hspace{1cm}-H(U_1,X_1,Z)-\epsilon H_{\phi}(U_1,X_1,Z)\nonumber\\
&\hspace{1cm}+[\sum_{u_1,x_1,z} p(u_1,x_1,z)\left(1+\epsilon \phi(u_1)\right)\log\left(1+\epsilon \phi(u_1)\right)]\Big)\nonumber\\
&+\lambda_2\left(H(X_2,Z)-H(Z)+H_{\epsilon}(U_2,Z)-H_{\epsilon}(U_2,X_2,Z)\right) \notag\\
&+\lambda_3\Big(H(X_1,X_2,Y,Z)-H(X_2,Z)+H(U_1,X_2,Z)+\nonumber\\
&\epsilon H_{\phi}(U_1,X_2,Z)-\nonumber\\
&[\sum_{u_1,x_2,z} p(u_1,x_2,z)\left(1+\epsilon \phi(u_1)\right)\log\left(1+\epsilon \phi(u_1)\right)]\nonumber\\
&-H(U_1,X_1,X_2,Y,Z)-\epsilon H_{\phi}(U_1,X_1,X_2,Y,Z)+\nonumber\\
&[\sum_{u_1,x_1,x_2,y,z} p(u_1,x_1,x_2,y,z)\left(1+\epsilon \phi(u_1)\right)\log\left(1+\epsilon \phi(u_1)\right)]\Big)\nonumber\\
&\stackrel{(c)}= \lambda_1\left(I(U_1;X_1|Z)+\epsilon H_{\phi}(U_1,Z)-\epsilon H_{\phi}(U_1,X_1,Z)\right)\nonumber\\
&+\lambda_2\left(H(X_2,Z)-H(Z)+H_{\epsilon}(U_2,Z)-H_{\epsilon}(U_2,X_2,Z)\right) \notag\\
&+\lambda_3 (I(U_1;X_1,Y|X_2,Z)+\epsilon H_{\phi}(U_1,X_2,Z)\nonumber\\
&\hspace{1cm}-\epsilon H_{\phi}(U_1,X_1,X_2,Y,Z)), \label{eq:wsumT3}
\end{align}
where (a) follows since the joint distribution of $(X_1,X_2,Y,Z)$ is preserved from \eqref{eq:pertpreserveQ}, while in (b) we have defined
\begin{align*}
&H_{\phi}(U_1,Z) = -\sum_{u_1,z} p(u_1,z) \phi(u_1) \log p(u_1,z),\\
&H_{\phi}(U_1,X_1,Z) = -\sum_{u_1,x_1,z} p(u_1,x_1,z) \phi(u_1) \log p(u_1,x_1,z), \\
&H_{\phi}(U_1,X_1,X_2,Y,Z) \nonumber\\
&= -\!\!\sum_{u_1,x_1,x_2,y,z} p(u_1,x_1,x_2,y,z) \phi(u_1) \log p(u_1,x_1,x_2,y,z),
\end{align*}
and (c) follows from the fact that all the sums in the square brackets are equal.

Note that from \eqref{eq:pmfuxz},
\begin{align}
p_{\epsilon}(u_2,x_2,z) = p(u_2,x_2,z).
\end{align}
Hence it follows that
\begin{align}
H_{\epsilon}(U_2,X_2,Z) &= H(U_2,X_2,Z), \label{eq:entr3T3}\\
H_{\epsilon}(U_2,Z) &= H(U_2,Z). \label{eq:entr4T3}
\end{align}
Substituting \eqref{eq:entr3T3} and \eqref{eq:entr4T3} into \eqref{eq:wsumT3}, the weighted sum becomes
\begin{align}
&\lambda_1 I_{\epsilon}(U_1;X_1|Z)+\lambda_2 I_{\epsilon}(U_2;X_2|Z)+\lambda_3 I_{\epsilon}(U_1;X_1,Y|X_2,Z) \notag\\
&= \lambda_1\left(I(U_1;X_1|Z)+\epsilon H_{\phi}(U_1,Z)-\epsilon H_{\phi}(U_1,X_1,Z)\right) \nonumber\\
&\hspace{12pt}+\lambda_2 I(U_2;X_2|Z) +\lambda_3 (I(U_1;X_1,Y|X_2,Z)\notag\\
&\hspace{1cm}+\epsilon H_{\phi}(U_1,X_2,Z)-\epsilon H_{\phi}(U_1,X_1,X_2,Y,Z)). \label{eq:wsum2T3}
\end{align}
Now we apply the first derivative condition \eqref{eq:extrempert} to \eqref{eq:wsum2T3}. This yields
\begin{align}
&\lambda_1\left(H_{\phi}(U_1,Z)-H_{\phi}(U_1,X_1,Z)\right)\nonumber\\
&\hspace{12pt}+\lambda_3 \left(H_{\phi}(U_1,X_2,Z)-H_{\phi}(U_1,X_1,X_2,Y,Z)\right)= 0. \label{eq:wsum5T3} 
\end{align}
Substituting \eqref{eq:wsum5T3} into \eqref{eq:wsum2T3}, we obtain
\begin{align*}
&\lambda_1 I_{\epsilon}(U_1;X_1|Z)+\lambda_2 I_{\epsilon}(U_2;X_2|Z)+\lambda_3 I_{\epsilon}(U_1;X_1,Y|X_2,Z) \nonumber\\
&\!= \!\lambda_1 I(U_1;X_1|Z)+\!\lambda_2 I(U_2;X_2|Z)+\!\lambda_3 I(U_1;X_1,Y|X_2,Z).
\end{align*}
Thus if $p(u_1,u_2,x_1,x_2,y,z)$ attains the minimum of the weighted sum rate, then the latter is preserved for any valid perturbation $p_{\epsilon}(u_1,u_2,x_1,x_2,y,z)$ that satisfies \eqref{eq:perturbT3}. Now we choose $\epsilon$ such that
\begin{align*}
\min_{u_1} \left(1+\epsilon \phi(u_1)\right) = 0,
\end{align*}
and let $u_1=u_1^{*}$ attain this minimum. Clearly, $p_{\epsilon}(u_1^{*})=0$, and hence there exists an $U_1$ with cardinality at most $|\mathcal{U}_1|-1$ such that $\lambda_1 I(U_1;X_1|Z)+\lambda_2 I(U_2;X_2|Z)+\lambda_3 I(U_1;X_1,Y|X_2,Z)$ is preserved. We can proceed by induction until $|\mathcal{U}_1|=|\mathcal{X}_1||\mathcal{X}_2||\mathcal{Y}||\mathcal{Z}|$. When this happens, we are no longer guaranteed the existence of a non-trivial $\phi(u_1)$ satisfying \eqref{eq:perturbT3}. Hence we can restrict the cardinality to $|\mathcal{U}_1| \leq |\mathcal{X}_1||\mathcal{X}_2||\mathcal{Y}||\mathcal{Z}|$.

We next perturb $U_2$. For a given $p(x_1,x_2,u_1,u_2,y,z)$, consider another perturbation defined by
\begin{align*}
p'_{\epsilon}(x_1,x_2,u_1,u_2,y,z)=p(x_1,x_2,u_1,u_2,y,z)\left(1+\epsilon \phi'(u_2)\right).
\end{align*}
We require that $\left(1+\epsilon \phi'(u_2)\right) \geq 0$ for all $u_2$, and $\sum_{u_2}p(u_2)\phi'(u_2)=0$. Furthermore, let $\phi'(u_2)$ be such that
\begin{align} \label{eq:perturbT3U2}
&\mathbb{E}\left[\phi'(U_2)|U_1=u_1,X_1=x_1,X_2=x_2,Y=y,Z=z\right] \nonumber\\
&= \sum_{u_2} p(u_2|u_1,x_1,x_2,y,z) \phi'(u_2) = 0, \forall \:\: u_1, x_1, x_2, y, z.
\end{align}
Such a non-zero perturbation satisfying \eqref{eq:perturbT3U2} exists as long as $|\mathcal{U}_2| > |\mathcal{U}_1||\mathcal{X}_1||\mathcal{X}_2||\mathcal{Y}||\mathcal{Z}|$.
Again, it can be verified that $p'_{\epsilon}(\cdot)$  preserves $p(x_1,x_2,y,z)$ as well as the structure of the p.m.f. in \eqref{p.m.f.structure1}.
We examine the weighted sum term under the perturbed distribution $p'_{\epsilon}(\cdot)$.
\begin{align}
&\lambda_1 I_{\epsilon}(U_1;X_1|Z)+\lambda_2 I_{\epsilon}(U_2;X_2|Z)+\lambda_3 I_{\epsilon}(U_1;X_1,Y|X_2,Z) \notag\\
&= \lambda_1\left(H(X_1,Z)-H(Z)+H_{\epsilon}(U_1,Z)-H_{\epsilon}(U_1,X_1,Z)\right)\nonumber\\
&+\lambda_2\left(H(X_2,Z)-H(Z)+H_{\epsilon}(U_2,Z)-H_{\epsilon}(U_2,X_2,Z)\right) \notag\\
& +\lambda_3 (H(X_1,X_2,Y,Z)-H(X_2,Z)+H_{\epsilon}(U_1,X_2,Z)\nonumber\\
&\hspace{12pt}-H_{\epsilon}(U_1,X_1,X_2,Y,Z)) \notag\\
&\stackrel{(a)}= \lambda_1\left(H(X_1,Z)-H(Z)+H_{\epsilon}(U_1,Z)-H_{\epsilon}(U_1,X_1,Z)\right)\nonumber\\
&+\lambda_2\left(I(U_2;X_2|Z)+\epsilon H_{\phi'}(U_2,Z)-\epsilon H_{\phi'}(U_2,X_2,Z)\right) \notag\\
&+\lambda_3(H(X_1,X_2,Y,Z)-H(X_2,Z)+H_{\epsilon}(U_1,X_2,Z)\nonumber\\
&\hspace{12pt}-H_{\epsilon}(U_1,X_1,X_2,Y,Z)), \label{eq:wsumT3U2} 
\end{align}
where in (a) we have defined $H_{\phi'}(U_2,X_2,Z) = -\sum_{u_2,x_2,z} p(u_2,x_2,z) \phi'(u_2) \log p(u_2,x_2,z)$. Consider the p.m.f. of $(U_1,X_1,X_2,Y,Z)$ under the perturbation $p'_{\epsilon}(\cdot)$.
\begin{align*}
&p'_{\epsilon}(u_1,x_1,x_2,y,z) = \sum_{u_2} p'_{\epsilon}(u_1,u_2,x_1,x_2,y,z)\nonumber\\
& = \sum_{u_2} p(u_1,u_2,x_1,x_2,y,z)\left(1+\epsilon \phi'(u_2)\right) \notag\\
&= p(u_1,x_1,x_2,y,z) \left\{1+\epsilon \sum_{u_2} p(u_2|u_1,x_1,x_2,y,z) \phi'(u_2) \right\} \nonumber\\
&\stackrel{(a)}= p(u_1,x_1,x_2,y,z),
\end{align*}
where (a) follows from \eqref{eq:perturbT3U2}. Hence it follows that the terms $H_{\epsilon}(U_1,X_1,X_2,Y,Z)$, $H_{\epsilon}(U_1,X_2,Z)$, $H_{\epsilon}(U_1,Z)$ and $H_{\epsilon}(U_1,X_1,Z)$ are preserved under $p'_{\epsilon}(\cdot)$. Hence, the weighted sum in \eqref{eq:wsumT3U2} becomes
\begin{align}
&\lambda_1 I_{\epsilon}(U_1;X_1|Z)+\lambda_2 I_{\epsilon}(U_2;X_2|Z)+\lambda_3 I_{\epsilon}(U_1;X_1,Y|X_2,Z) \notag\\
&= \lambda_1 I(U_1;X_1|Z)+\lambda_2 (I(U_2;X_2|Z)+\epsilon H_{\phi'}(U_2,Z)\nonumber\\
&\hspace{12pt}-\epsilon H_{\phi'}(U_2,X_2,Z)) +\lambda_3 I(U_1;X_1,Y|X_2,Z). \label{eq:wsum2Tq3U2}
\end{align}
Now we apply the first derivative condition to \eqref{eq:wsum2Tq3U2}. This yields
\begin{align}
&\lambda_2\left(H_{\phi'}(U_2,Z)- H_{\phi'}(U_2,X_2,Z)\right)= 0. \label{eq:wsum5Tq3U2} 
\end{align}
Substituting \eqref{eq:wsum5Tq3U2} into \eqref{eq:wsum2Tq3U2}, it follows that the weighted sum is preserved. Thus if $p(u_1,u_2,x_1,x_2,y,z)$ attains the minimum of the weighted sum rate, then the latter is preserved for any valid perturbation $p'_{\epsilon}(u_1,u_2,x_1,x_2,y,z)$ that satisfies \eqref{eq:perturbT3U2}. Now the proof is completed by choosing $\epsilon$ such that $
\min_{u_2} \left(1+\epsilon \phi'(u_2)\right) = 0$
and the cardinality of $U_2$ drops by $1$. We can proceed by induction until $|\mathcal{U}_2|=|\mathcal{U}_1||\mathcal{X}_1||\mathcal{X}_2||\mathcal{Y}||\mathcal{Z}|$. 
}

{\blue\section{Proof of Lemma~\ref{lem:cardbndU0}}\label{app:cardbndU0}
Since the arguments are along the same lines as Appendix~\ref{app:cardbnd}, we only outline the differences below. Again, the cardinality bound on $T$ can be obtained using the support lemma~\cite[Appendix C]{el2011network}, and we focus on cardinality bounds for $U_0$, $U_1$ and $U_2$. Let $\mathcal{C}'$ denote the region in $S^\prime_\epsilon$ with $|\mathcal{T}|=1$, i.e., the collection of $(R_1,R_2)$ such that:
\begin{align}
R_1 &\geq I(U_0,U_1;X_1|Z) \label{eq:thm5e1}\\
R_2 &\geq I(U_0,U_2;X_2|Z) \label{eq:thm5e2}\\
R_1+R_2 &\geq I(U_0,U_1,U_2;X_1,X_2|Z), \label{eq:thm5e3}
\end{align}
for some p.m.f. 
\begin{align}
&p(x_1,x_2,u_0,u_1,u_2,y,z)=p(z)p(x_1|z)p(x_2|z)p(u_0)\nonumber\\
&\hspace{24pt}\times p(u_1|x_1,u_0)p(u_2|x_2,u_0)p(y|u_1,u_2,z) \label{eq:pmfstruct2}
\end{align} 
such that 
\begin{align}
\lVert\sum\limits_{u_0,u_1,u_2}p(x_1,x_2,u_0,u_1,u_2,y,z)-q(x_1,x_2,y,z)\rVert\leq \epsilon. \label{eq:pmfmarg2}
\end{align}
As before, we use the perturbation argument of \cite{gohari2012evaluation}, and the proof is in two steps.
\begin{itemize}
\item \textbf{Step 1:} We prove that $\mathcal{C}' = \textup{Closure}\left(\underset{K_{0},K_{1},K_{2} \geq 0}{\bigcup} \mathcal{C}'^{K_{0},K_{1},K_{2}}\right),$ where the region $\mathcal{C}'^{K_{0},K_{1},K_{2}}$ for positive integers $K_0,K_1,K_2$ is given by the union of rate pairs $(R_1,R_2)$ satisfying \eqref{eq:thm5e1}--\eqref{eq:thm5e3} over $(U_0,U_1,U_2,X_1,X_2,Y,Z)$ with cardinality bounds $|\mathcal{U}_0| \leq K_{0}$, $|\mathcal{U}_1| \leq K_{1}$ $\&$ $|\mathcal{U}_2| \leq K_{2}$ and having a joint p.m.f. 
\begin{align*}
&p(z)p(x_1)p(x_2)p(u_0)p(u_1|x_1,u_0)p(u_2|x_2,u_0)\\
&\hspace{12pt}\times p(y|u_1,u_2,z)
\end{align*}
 such that $$\lVert p(x_1,x_2,y,z)-q(x_1,x_2,y,z)\rVert\leq\epsilon$$.
\item \textbf{Step 2:} If $|\mathcal{U}_0| \leq K_{0}$, $|\mathcal{U}_1| \leq K_{1}$ and $|\mathcal{U}_2| \leq K_{2}$ for some constants $K_{0},K_{1},K_{2}$, we show that the auxiliary cardinalities can be brought down to $|\mathcal{U}_0| \leq |\mathcal{X}_1||\mathcal{X}_2||\mathcal{Y}||\mathcal{Z}|$, $|\mathcal{U}_1| \leq |\mathcal{U}_0||\mathcal{X}_1||\mathcal{X}_2||\mathcal{Y}||\mathcal{Z}|$ and $|\mathcal{U}_2| \leq |\mathcal{U}_0||\mathcal{X}_1||\mathcal{X}_2||\mathcal{Y}||\mathcal{Z}|$.
\end{itemize}

\vspace{2mm}

\noindent \textbf{Proof Step 1:}\\
We show that any rate pair $(R_1,R_2) \in \mathcal{C}'$ is a limit point of the set $\underset{K_{0},K_{1},K_{2} \geq 0}{\bigcup} \mathcal{C}'^{K_{0},K_{1},K_{2}}$. Firstly, if $(R_1,R_2) \in \mathcal{C}'$, then random variables $(U_0,U_1,U_2,X_1,X_2,Y,Z)$ satisfying \eqref{eq:thm5e1}--\eqref{eq:thm5e3} exist such that their joint p.m.f. is of the form \eqref{eq:pmfstruct2} and \eqref{eq:pmfmarg2} holds. Assume that $\mathcal{U}_j = \{1,2,\cdots\}$ for $j=0,1,2$. Define a modified version $U_0'$ of the random variable $U_0$ with p.m.f. specified as:
\begin{align} 
p_{U_0'}(i) = \begin{cases}
p_{U_0}(i), &{i=\{1,2,\cdots,m\}}\\
\textup{Pr}(U_0>m), &{i=m+1,}
\end{cases}
\end{align}
where $m$ is an integer. The alphabet of $U_0'$ is $\mathcal{U}_0'=\{1,2,\cdots,m+1\}$.
Likewise, define modified versions $(U_1',U_2')$ of random variables $(U_1,U_2)$ taking values in $\{1,2,\cdots,m\} \cup \mathcal{X}_1 \cup (\mathcal{X}_1 \times \mathcal{U}_0^\prime)$ and $\{1,2,\cdots,m\} \cup \mathcal{X}_2 \cup (\mathcal{X}_2 \times \mathcal{U}_0^\prime)$ respectively. The alphabet cardinalities of $(U_1',U_2')$ are $m+|\mathcal{X}_1|+(m+1)|\mathcal{X}_1|$ and $m+|\mathcal{X}_2|+(m+1)|\mathcal{X}_2|$ respectively. 
Let the conditional p.m.f. of $U_j'$ given $(X_j,U_0')$ for $j \in \{1,2\}$, $x_j \in \mathcal{X}_j$ be specified as:
\begin{align} 
&p_{U_j'|X_j,U_0'}(i|x_j,u_0') \nonumber\\
&\!= p_{U_j|X_j,U_0}(i|x_j,u_0'),\! u_0' \in \{1,2,\cdots,m\}, \!i \in \{1,2,\cdots,m\}, \label{eq:truncpmfU01}\\
&p_{U_j'|X_j,U_0'}((x_j,u_0')|x_j,u_0') \nonumber\\
&= \sum_{i=m+1}^{\infty}p_{U_j|X_j,U_0}(i|x_j,u_0'), u_0' \in \{1,2,\cdots,m\}, \label{eq:truncpmfU03}\\
&p_{U_j'|X_j,U_0'}(x_j|x_j,m+1) = 1, \label{eq:truncpmfU02}\\
&p_{U_j'|X_j,U_0'}(\cdot|\cdot,\cdot) = 0, \: \textup{otherwise}. 
\end{align}

As before, to define $p_{Y|U_1',U_2',Z}$, consider new auxiliary random variables $(U_1'',U_2'')$ whose alphabets are the same as that of $(U_1,U_2)$. Let the conditional p.m.f. $p_{U_1'',U_2''|U_1',U_2'}$ be specified as follows:
\begin{itemize}
\item If ${U}_j^\prime\in\{1,2,\dots,m\}\cup (\mathcal{X}_j\times \mathcal{U}_0)$, $j=1,2$, then, we define
\begin{align}
P_{U_1^{\prime\prime},U_2^{\prime\prime}|U_1^\prime,U_2^\prime}&(u_1'',u_2''|u_1',u_2')\nonumber\\
&=P_{U_1^{\prime\prime}|U_1^\prime}(u_1''|u_1')\cdot P_{U_2^{\prime\prime}|U_2^\prime}(u_2''|u_2'),
\end{align}
where 
\begin{align}
p_{U_j''|U_j'}(i|i) = 1, i=1,2,\cdots,m.
\end{align}
\begin{align}
&p_{U_j''|U_j'}(k|(x_j,i)) = \nonumber\\
&\hspace{-12pt}\begin{cases}
\textup{Pr}(U_j=k|X_j=x_j,U_0=i,U_j>m), &\!\!\!{k\geq m+1}\\
0, &\!\!\!{k=1,2,\cdots,m}.
\end{cases} \label{eq:truncpmfU05}
\end{align}
\item Note that, otherwise, $U_0'=m+1$ and hence, $U_1^\prime\in\mathcal{X}_1$ and $U_2^\prime\in\mathcal{X}_2$. For this case, we define
\begin{align}\label{eq:truncpmfU04}
&p_{U_1'',U_2''|U_1',U_2'}(u_1,u_2|x_1,x_2) =\nonumber\\
& \sum_{u_0} \textup{Pr}(U_0=u_0|U_0>m)p_{U_1|X_1,U_0}(u_1|x_1,u_0)\cdot\nonumber\\
&\hspace{12pt}p_{U_2|X_2,U_0}(u_2|x_2,u_0).
\end{align}
\end{itemize}
It is easy to verify that
\begin{align}
p_{U_1'',U_2''|X_1,X_2}(u_1,u_2|x_1,x_2)=p_{U_1,U_2|X_1,X_2}(u_1,u_2|x_1,x_2).
\end{align}
Further, let $$p_{Y|U_1'',U_2'',Z}(y|u_1,u_2,z) = p_{Y|U_1,U_2,Z}(y|u_1,u_2,z)$$. The conditional distribution $p_{Y|U_1',U_2',Z}$ is defined as
\begin{align}
&p_{Y|U_1',U_2',Z}(y|u_1',u_2',z) \nonumber\\
&= \sum_{u_1'',u_2''} p_{U_1'',U_2''|U_1',U_2'}(u_1'',u_2''|u_1',u_2') p_{Y|U_1'',U_2'',Z}(y|u_1'',u_2'',z) \label{eq:truncpmfU06}\\
&= \sum_{u_1'',u_2''} p_{U_1'',U_2''|U_1',U_2'}(u_1'',u_2''|u_1',u_2') p_{Y|U_1,U_2,Z}(y|u_1'',u_2'',z). 
\end{align}
The above definitions preserve all the Markov chains associated with the original p.m.f. As a result, the joint distribution of $(U_0',U_1',U_2',X_1,X_2,Y,Z)$ factors according to $p_{U_0',U_1',U_2',X_1,X_2,Y,Z}=p_{X_1,X_2,Z}p_{U_0'}p_{U_1'|X_1,U_0'}p_{U_2'|X_2,U_0'}p_{Y|U_1',U_2',Z}$. On marginalizing away $(U_0',U_1',U_2')$ from $p_{U_0',U_1',U_2',X_1,X_2,Y,Z}$, we may verify that we obtain the p.m.f. $p_{X_1,X_2,Y,Z}$.

It follows that the joint distribution of the random variables $(U_{0}',U_{1}',U_{2}',X_{1},X_{2},Y,Z)$ converges to the the joint distribution of $(U_0,U_{1},U_{2},X_{1},X_{2},Y,Z)$ in the limit $m \to \infty$. As a result, the mutual information terms $I(U_{0}',U_{1}';X_{1}|Z)$, $I(U_{0}',U_{2}';X_{2}|Z)$ and $I(U_{0}',U_1',U_{2}',;X_1,X_2|Z)$ converge to $I(U_{0},U_{1};X_{1}|Z)$, $I(U_{0},U_{2};X_{2}|Z)$ and $I(U_{0},U_1,U_{2},;X_1,X_2|Z)$ respectively. Hence, we conclude that the given rate pair $(R_1,R_2)$ is a limit point of the set $\underset{K_{0},K_{1},K_{2} \geq 0}{\bigcup} \mathcal{C}'^{K_{0},K_{1},K_{2}}$.

\vspace{2mm}

\noindent \textbf{Proof Step 2:}\\
As in Appendix~\ref{app:cardbnd}, it suffices to consider an optimization of the weighted sum term $\lambda_1 I(U_0,U_1;X_1|Z)+\lambda_2 I(U_0,U_2;X_2|Z)+\lambda_3 I(U_0,U_1,U_2;X_1,X_2|Z)$ for non-negative reals $\lambda_1,\lambda_2,\lambda_3$, and find new auxiliary random variables whose cardinalities are bounded while not increasing the weighted sum. For a given $p(x_1,x_2,u_0,u_1,u_2,y,z)$, consider the perturbation defined by
\begin{align*}
p_{\epsilon}(x_1,x_2,u_0,u_1,u_2,y,z)&=p(x_1,x_2,u_0,u_1,u_2,y,z)\nonumber\\
&\hspace{12pt}\times\left(1+\epsilon \phi(u_0)\right).
\end{align*}
For $p_{\epsilon}(x_1,x_2,u_0,u_1,u_2,y,z)$ to be a valid p.m.f., we require that $\left(1+\epsilon \phi(u_0)\right) \geq 0$ for all $u_0$, and $\sum_{u_0}p(u_0)\phi(u_0)=0$. Furthermore, we will consider perturbations $\phi(u_0)$ such that
\begin{align} \label{eq:perturbT3U0}
&\mathbb{E}\left[\phi(U_0)|X_1=x_1,X_2=x_2,Y=y,Z=z\right] \nonumber\\
&= \sum_{u_0} p(u_0|x_1,x_2,y,z) \phi(u_0) = 0, \:\: \forall \:\: x_1, x_2 , y, z.
\end{align}
Observe that such a non-zero perturbation satisfying \eqref{eq:perturbT3U0} exists as long as $|\mathcal{U}_0| > |\mathcal{X}_1||\mathcal{X}_2||\mathcal{Y}||\mathcal{Z}|$.
Similar to Appendix~\ref{app:cardbnd}, it can be shown verified that this perturbation preserves $p(x_1,x_2,y,z)$ as well as the structure of the p.m.f. in \eqref{eq:pmfstruct2}.
We examine the weighted sum term under $p_{\epsilon}(\cdot)$.
\begin{align}
&\lambda_1 I_{\epsilon}(U_0,U_1;X_1|Z)+\lambda_2 I_{\epsilon}(U_0,U_2;X_2|Z)\nonumber\\
&\hspace{12pt}+\lambda_3 I_{\epsilon}(U_0,U_1,U_2;X_1,X_2|Z) \notag\\
&= \lambda_1\left(H(X_1|Z)+H_{\epsilon}(U_0,U_1,Z)-H_{\epsilon}(U_0,U_1,X_1,Z)\right)\nonumber\\
&+\lambda_2\left(H(X_2|Z)+H_{\epsilon}(U_0,U_2,Z)-H_{\epsilon}(U_0,U_2,X_2,Z)\right) \notag\\
& +\lambda_3 (H(X_1,X_2|Z)+H_{\epsilon}(U_0,U_1,U_2,Z)\nonumber\\
&\hspace{12pt}-H_{\epsilon}(U_0,U_1,U_2,X_1,X_2,Z)) \notag\\
&\stackrel{(a)}= \lambda_1\Big(I(U_0,U_1;X_1|Z)+\epsilon H_{\phi}(U_0,U_1,Z)\nonumber\\
&\hspace{12pt}-\epsilon H_{\phi}(U_0,U_1,X_1,Z)\Big)\nonumber\\
&+\lambda_2\Big(I(U_0,U_2;X_2|Z)+\epsilon H_{\phi}(U_0,U_2,Z)\nonumber\\
&\hspace{12pt}-\epsilon H_{\phi}(U_0,U_2,X_2,Z)\Big) +\lambda_3 \Big(I(U_0,U_1,U_2;X_1,X_2|Z)\notag\\
&\hspace{12pt}+\epsilon H_{\phi}(U_0,U_1,U_2,Z)-\epsilon H_{\phi}(U_0,U_1,U_2,X_1,X_2,Z)\Big), \label{eq:wsumT3U0}
\end{align}
where in (a) we have defined 
\begin{align*}
H_{\phi}(U_0,&U_1,U_2,X_1,X_2,Z) \nonumber\\
&= -\sum_{u_0,u_1,u_2,x_1,x_2,z} p(u_0,u_1,u_2,x_1,x_2,z) \phi(u_0)\cdot\nonumber\\
&\hspace{36pt} \log p(u_0,u_1,u_2,x_1,x_2,z).
\end{align*}

Now we apply the first derivative condition to \eqref{eq:wsumT3U0}. This yields
\begin{align}
&\lambda_1\left(H_{\phi}(U_0,U_1,Z)- H_{\phi}(U_0,U_1,X_1,Z)\right)\nonumber\\
&+\lambda_2(H_{\phi}(U_0,U_2,Z)- H_{\phi}(U_0,U_2,X_2,Z)) \notag\\
&+\lambda_3 \left(H_{\phi}(U_0,U_1,U_2,Z)- H_{\phi}(U_0,U_1,U_2,X_1,X_2,Z)\right)= 0. \label{eq:wsum5T3U0} 
\end{align}
Substituting \eqref{eq:wsum5T3U0} into \eqref{eq:wsumT3U0}, it follows that the weighted sum is preserved. Thus if $p(u_0,u_1,u_2,x_1,x_2,y,z)$ attains the minimum of the weighted sum rate, then the latter is preserved for any valid perturbation $p_{\epsilon}(u_0,u_1,u_2,x_1,x_2,y,z)$ that satisfies \eqref{eq:perturbT3U0}. Now we choose $\epsilon$ such that $\min_{u_0} \left(1+\epsilon \phi(u_0)\right) = 0$ and let $u_0=u_0^{*}$ attain this minimum. This makes $p_{\epsilon}(u_0^{*})=0$, and the cardinality of $U_0$ can be reduced by one. We can proceed by induction until $|\mathcal{U}_0|=|\mathcal{X}_1||\mathcal{X}_2||\mathcal{Y}||\mathcal{Z}|$. Hence we can restrict the cardinality to $|\mathcal{U}_0| \leq |\mathcal{X}_1||\mathcal{X}_2||\mathcal{Y}||\mathcal{Z}|$.

We next perturb $U_1$. For a given $p(x_1,x_2,u_0,u_1,u_2,y,z)$, consider another perturbation defined by
\begin{align*}
p'_{\epsilon}(x_1,x_2,u_0,u_1,u_2,y,z)&=p(x_1,x_2,u_0,u_1,u_2,y,z)\nonumber\\
&\hspace{12pt}\times\left(1+\epsilon \phi'(u_1)\right).
\end{align*}
We require that $\left(1+\epsilon \phi'(u_1)\right) \geq 0$ for all $u_1$, and $\sum_{u_1}p(u_1)\phi'(u_1)=0$. Furthermore, let $\phi'(u_1)$ be such that
\begin{align} \label{eq:perturbT23U0}
&\mathbb{E}\left[\phi'(U_1)|U_0=u_0,X_1=x_1,X_2=x_2,Y=y,Z=z\right] \nonumber\\
&= \sum_{u_1} p(u_1|u_0,x_1,x_2,y,z) \phi'(u_1) = 0, \forall \:\: u_0, x_1, x_2, y, z.
\end{align}
Such a non-zero perturbation satisfying \eqref{eq:perturbT23U0} exists as long as $|\mathcal{U}_1| > |\mathcal{U}_0||\mathcal{X}_1||\mathcal{X}_2||\mathcal{Y}||\mathcal{Z}|$.
Again, it can be verified that $p'_{\epsilon}(\cdot)$  preserves $p(x_1,x_2,y,z)$ as well as the structure of the p.m.f. in \eqref{eq:pmfstruct2}.
We examine the weighted sum term under the perturbed distribution $p'_{\epsilon}(\cdot)$.
\begin{align}
&\lambda_1 I_{\epsilon}(U_0,U_1;X_1|Z)+\lambda_2 I_{\epsilon}(U_0,U_2;X_2|Z)\nonumber\\
&\hspace{12pt}+\lambda_3 I_{\epsilon}(U_0,U_1,U_2;X_1,X_2|Z) \notag\\
&= \lambda_1\left(H(X_1|Z)+H_{\epsilon}(U_0,U_1,Z)-H_{\epsilon}(U_0,U_1,X_1,Z)\right)\nonumber\\
&+\lambda_2\left(H(X_2|Z)+H_{\epsilon}(U_0,U_2,Z)-H_{\epsilon}(U_0,U_2,X_2,Z)\right) \notag\\
&+\lambda_3 (H(X_1,X_2|Z)+H_{\epsilon}(U_0,U_1,U_2,Z)\nonumber\\
&\hspace{12pt}-H_{\epsilon}(U_0,U_1,U_2,X_1,X_2,Z)) \notag\\
&\stackrel{(a)}= \lambda_1\Big(I(U_0,U_1;X_1|Z)+\epsilon H_{\phi'}(U_0,U_1,Z)\nonumber\\
&\hspace{12pt}-\epsilon H_{\phi'}(U_0,U_1,X_1,Z)\Big)\nonumber\\
&+\lambda_2\left(H(X_2|Z)+H_{\epsilon}(U_0,U_2,Z)-H_{\epsilon}(U_0,U_2,X_2,Z)\right) \notag\\
&+\lambda_3 \Big(I(U_0,U_1,U_2;X_1,X_2|Z)+\epsilon H_{\phi'}(U_0,U_1,U_2,Z)\nonumber\\
&\hspace{12pt}-\epsilon H_{\phi'}(U_0,U_1,U_2,X_1,X_2,Z)\Big), \label{eq:wsumTq3U0}
\end{align}
where in (a) we have defined 
\begin{align*}
H_{\phi'}&(U_0,U_1,U_2,X_1,X_2,Z) \nonumber\\
&= -\sum_{u_0,u_1,u_2,x_1,x_2,z} p(u_0,u_1,u_2,x_1,x_2,z) \phi'(u_1) \cdot\nonumber\\
&\hspace{36pt}\log p(u_0,u_1,u_2,x_1,x_2,z).
\end{align*}
 Consider the p.m.f. of $(U_0,U_2,X_2,Z)$ under the perturbation $p'_{\epsilon}(\cdot)$.
\begin{align*}
&p'_{\epsilon}(u_0,u_2,x_2,z) = \sum_{u_1,x_1,y} p'_{\epsilon}(u_0,u_1,u_2,x_1,x_2,y,z) \nonumber\\
&= \sum_{u_1,x_1,y} p(u_0,u_1,u_2,x_1,x_2,y,z)\left(1+\epsilon \phi'(u_1)\right) \notag\\
&= p(u_0,u_2,x_2,z) \left\{1+\epsilon \sum_{u_1} p(u_1|u_0,u_2,x_2,z) \phi'(u_1) \right\} \notag\\
&\stackrel{(a)}= p(u_0,u_2,x_2,z) \left\{1+\epsilon \sum_{u_1} p(u_1|x_2,u_0) \phi'(u_1) \right\} \nonumber\\
&\stackrel{(b)}= p(u_0,u_2,x_2,z),
\end{align*}
where (a) follows since $U_2 \to (X_2,U_0) \to (U_1,X_1,Z)$, while (b) follows from \eqref{eq:perturbT23U0}. Hence it follows that
\begin{align}
H_{\epsilon}(U_0,U_2,X_2,Z) &= H(U_0,U_2,X_2,Z), \label{eq:entr3T23U0}\\
H_{\epsilon}(U_0,U_2,Z) &= H(U_0,U_2,Z). \label{eq:entr4T23U0}
\end{align}
Substituting \eqref{eq:entr3T23U0} and \eqref{eq:entr4T23U0} into \eqref{eq:wsumTq3U0}, the weighted sum becomes
\begin{align}
&\lambda_1 I_{\epsilon}(U_0,U_1;X_1|Z)+\lambda_2 I_{\epsilon}(U_0,U_2;X_2|Z)\nonumber\\
&\hspace{12pt}+\lambda_3 I_{\epsilon}(U_0,U_1,U_2;X_1,X_2|Z) \notag\\
&= \lambda_1\Big(I(U_0,U_1;X_1|Z)+\epsilon H_{\phi'}(U_0,U_1,Z)\nonumber\\
&\hspace{12pt}-\epsilon H_{\phi'}(U_0,U_1,X_1,Z)\Big)\nonumber\\
& +\lambda_2 I(U_0,U_2;X_2|Z)+\lambda_3 \Big(I(U_0,U_1,U_2;X_1,X_2|Z) \notag\\
&+\epsilon H_{\phi'}(U_0,U_1,U_2,Z)-\epsilon H_{\phi'}(U_0,U_1,U_2,X_1,X_2,Z)\Big). \label{eq:wsum2Tq3U0}
\end{align}
Now we apply the first derivative condition to \eqref{eq:wsum2Tq3U0}. This yields
\begin{align}
&\lambda_1\left(H_{\phi'}(U_0,U_1,Z)- H_{\phi'}(U_0,U_1,X_1,Z)\right)+\nonumber\\
&\lambda_3 \left(H_{\phi'}(U_0,U_1,U_2,Z)- H_{\phi'}(U_0,U_1,U_2,X_1,X_2,Z)\right)= 0. \label{eq:wsum5Tq3U0} 
\end{align}
Substituting \eqref{eq:wsum5Tq3U0} into \eqref{eq:wsum2Tq3U0}, it follows that the weighted sum is preserved. Thus if $p(u_0,u_1,u_2,x_1,x_2,y,z)$ attains the minimum of the weighted sum rate, then the latter is preserved for any valid perturbation $p'_{\epsilon}(u_0,u_1,u_2,x_1,x_2,y,z)$ that satisfies \eqref{eq:perturbT23U0}. Now the proof is completed by choosing $\epsilon$ such that $
\min_{u_1} \left(1+\epsilon \phi'(u_1)\right) = 0$
and the cardinality of $U_1$ drops by $1$. We can proceed by induction until $|\mathcal{U}_1|=|\mathcal{U}_0||\mathcal{X}_1||\mathcal{X}_2||\mathcal{Y}||\mathcal{Z}|$. The same argument can be repeated to make $|\mathcal{U}_2| \leq |\mathcal{U}_0||\mathcal{X}_1||\mathcal{X}_2||\mathcal{Y}||\mathcal{Z}|$ as well.
}
\bibliographystyle{IEEEtran}
\bibliography{mylit}
\begin{IEEEbiographynophoto}
{Gowtham R. Kurri} (Member, IEEE) graduated from the International Institute of Information Technology, Hyderabad, India, with a B.\ Tech.\ degree in Electronics and Communication Engineering, in 2011. He received his M.Sc. and Ph.D. degrees from the Tata Institute of Fundamental Research, Mumbai, India in 2020. Since 2020, he has been a Post-Doctoral Researcher at the School of Electrical, Computer and Energy Engineering at Arizona State University. His research interests are in information theory and statistical machine learning.

From 2011-2012, he worked as an Associate Engineer at Qualcomm India Private Limited, Hyderabad, India. From July to October, 2019, he was a Research Intern in the Blockchain Technology Group at IBM Research, Bangalore, India. 
\end{IEEEbiographynophoto}
\begin{IEEEbiographynophoto}
{Viswanathan Ramachandran} was born in Kerala, India. He received the Ph.D. degree in electrical engineering from the Indian Institute of Technology Bombay in 2020. He worked as a visiting researcher at the Tata Institute of Fundamental Research during 2020. Since 2021, he has been a postdoctoral fellow with the Department of Electrical Engineering, Technical University of Eindhoven, the Netherlands. His research interests are in information theory and multi-user information theory with applications to wireless and optical channels. He was a recipient of the Naik and Rastogi Award for Excellence in Ph.D. research from the Indian Institute of Technology Bombay in 2021, and also a Best Paper Award at the 2019 25th National Conference on Communications (NCC) held at the Indian Institute of Science Bangalore.
\end{IEEEbiographynophoto}
\begin{IEEEbiographynophoto}
{Sibi Raj B. Pillai} (Member, IEEE) received the Ph.D. degree in computer science and communication systems from EPFL, Switzerland, in July 2007. From October 2007 to April 2009, he was a Research Fellow with the University of Melbourne. Since 2009, he has been a Faculty with the Department of Electrical Engineering, Indian Institute of Technology Bombay. His research interests are in network information theory, feedback communications, cross layer scheduling, biological information inheritance, and radar signal processing.
\end{IEEEbiographynophoto}
\begin{IEEEbiographynophoto}
{Vinod M. Prabhakaran} (Member, IEEE) received the M.E. degree from the Indian Institute of Science in 2001 and the Ph.D. degree from the University of California, Berkeley in 2007. He was a Post-Doctoral Researcher at the University of Illinois, Urbana-Champaign from 2008 to 2010 and the Ecole Polytechnique Fédérale de Lausanne, Switzerland in 2011. Since 2011, he has been at the School of Technology and Computer Science at the Tata Institute of Fundamental Research, Mumbai. His research interests are in information theory, communication, cryptography, and signal processing. He was an Associate Editor of Shannon Theory for the IEEE Transactions on Information Theory from 2016 to 2019. He is currently an Associate Editor of Security and Privacy for the IEEE Transactions on Information Theory.
\end{IEEEbiographynophoto}
\end{document}